\documentclass[12pt]{article}
\usepackage{formatting}
\usepackage{shortcuts}

\begin{document}

\begin{titlepage}
\begin{center}
\phantom{ }
\vspace{0.7cm}

{\bf \Large{Quantum microstate counting from Brownian motion:\vskip .3cm
from many-body systems to black holes}}
\vskip 1.4cm
{\large   Enzo Bavaro${}^{\dagger\ddagger}$, Javier M. Mag\'{a}n${}^{\dagger\ddagger}$, and Leandro Martinek${}^{\ddagger}$}
\vskip 1cm

\small{${}^{\dagger}$ \textit{Departament de F\'isica Qu\`antica i Astrof\'isica, Institut de Ci\`encies del Cosmos}}
\vskip -.075cm
{\textit{ Universitat de Barcelona, Mart\'i i Franqu\`es 1, E-08028 Barcelona, Spain}}
\vskip 0.3cm
\small{${}^{\ddagger}$ \textit{Instituto Balseiro, Centro At\'omico Bariloche}}
\vskip -.075cm
{\textit{ 8400-S.C. de Bariloche, R\'io Negro, Argentina}}
\\[1cm]

{
\small{\noindent\texttt{E-mail}:
 \href{mailto:enzo.bavaro@ib.edu.ar}{enzo.bavaro@ib.edu.ar}, \href{mailto:magan@fqa.ub.edu}{magan@fqa.ub.edu}, \href{mailto:leandro.martinek@ib.edu.ar}{leandro.martinek@ib.edu.ar}}}

\end{center}

\vspace{.6cm}

\begin{abstract}
We introduce a new way to produce infinite families of bases of a quantum system's Hilbert space, as well as methods to find its dimension. These families are constructed via Brownian motions in the Hilbert space, defined using disordered, time-dependent couplings. The dimension spanned by them has zero variance over the ensemble of disordered couplings, and it is determined by certain replica partition functions. We apply these methods to finite dimensional $q$-local  systems (spin clusters and SYK) and black holes. For $q$-local systems we find exact expressions at finite $q$, $N$ and $t$, for the replica partition functions, together with universal behavior at large times, and a semiclassical analysis at large-$N$ in appropriate master field variables. The right Hilbert space dimension is obtained for any time $t>0$, $q>1$ and $N>0$. For black holes, the Brownian motions prepare quantitatively generic ensembles of black hole microstates, and the Bekenstein–Hawking entropy is universally reproduced by counting those. There the $n$-replica partition functions are constructed by gluing $n$-traversable wormholes at their future and past. This method demonstrates the robustness of black hole microstate counting methods, avoiding several limitations of previous constructions, including the non-genericity of the microstates and associated interiors, implicit statistics of heavy operators, limited microscopic control over overlaps, and the need for specific limits in the calculation.
\end{abstract}
\vspace{-.1cm}

\end{titlepage}

\cleardoublepage

\setcounter{tocdepth}{2}
\tableofcontents
\noindent\makebox[\linewidth]{\rule{\textwidth}{0.4 pt}}

\cleardoublepage

\section{Introduction}
\label{Sec:intro}

The discovery that black holes obey the entropy law \cite{Bekenstein:1973ur,Hawking:1975vcx}
\be\label{eq:bhe}
S = \frac{A}{4G_{\text{N}}}\,,
\ee
has long raised the question of its quantum statistical origin. It was widely thought that the answer would require understanding the microscopic structure of black hole microstates, an approach most concretely realized in certain classes of supersymmetric black holes in string theory, as pioneered in \cite{Strominger:1996sh}. More recent developments in holography in relation to the black hole information paradox \cite{Penington:2019npb,Almheiri:2019psf,Penington:2019kki,Almheiri:2019hni,Almheiri:2019qdq} have highlighted that non-perturbative effects within the semiclassical theory of gravity might be enough to account for the universality of \eqref{eq:bhe}.

A state-counting derivation of \eqref{eq:bhe} along these lines was proposed in \cite{Balasubramanian:2022gmo,Balasubramanian:2022lnw}. The proposal builds on the West Coast model \cite{Penington:2019kki}, on further developments in the gravitational path integral \cite{Marolf:2020xie,Hsin:2020mfa,Stanford:2020wkf,Sasieta:2022ksu}, and on the construction of black hole interior microstates in two-dimensional gravity \cite{Kourkoulou:2017zaj,Goel:2018ubv,Lin:2022rzw,Lin:2022zxd} and in higher dimensions \cite{Chandra:2022fwi}. Rather than relying on a full non-perturbative formulation of quantum gravity, the key observation is that, once such a theory is assumed to exist, large families of black hole microstates can already be described semiclassically as interior configurations within the low-energy effective field theory of gravity and matter. From this perspective, the problem is not the absence of microstates but rather their overabundance. In \cite{Balasubramanian:2022gmo,Balasubramanian:2022lnw}, this redundancy was made concrete by inserting thin shells of dust particles into the interior of a Schwarzschild black hole, formulating a sharp version of the “bag of gold” paradox. The proposed resolution is that, although the states appear naively orthogonal, they in fact overlap at the non-perturbative level, as revealed by Euclidean spacetime wormhole contributions to the gravitational path integral. These overlaps accumulate and reduce the number of linearly independent states in the family of microstates. The true dimension of the Hilbert space, computed using the gravitational input of these overlaps, reproduces the exponential of the entropy \eqref{eq:bhe}.

Some of these ideas have been further developed in subsequent work, including: \cite{Boruch:2023trc}, which performs a gravitational path integral state-counting derivation of the BPS degeneracy for black holes admitting $\mathcal{N}=2$ super-JT effective descriptions; \cite{Climent:2024trz}, which refines the general construction and emphasizes its universality for charged (and higher dimensional rotating) black holes;  \cite{Barbon:2025bbh}, which discusses possible negative modes of the Euclidean wormholes for the overlaps. See also \cite{Chandra:2023rhx,Maxfield:2023mdj,Boruch:2024kvv,Balasubramanian:2024rek,Balasubramanian:2024yxk,Balasubramanian:2025jeu,Balasubramanian:2025zey,Balasubramanian:2025hns,Geng:2024jmm,Balasubramanian:2025akx,He:2025neu} for related developments. Analogous state-counting methods have also been applied to other gravitational systems, notably to compute the Hilbert space dimension of closed universes; see \cite{Antonini:2023hdh,Antonini:2025ioh}, as well as \cite{Abdalla:2025gzn,Harlow:2025pvj} for further work in this direction. See also \cite{Espindola:2025wjf,deBoer:2025tmh} for related approaches to the state counting derivation of the de Sitter entropy.

Although remarkably successful in reproducing the black hole entropy \eqref{eq:bhe}, all these state-counting derivations come with several potential caveats. We now outline them:

\begin{enumerate}
    {\item {\it Atypical families of black hole microstates}. The proposed bases of microstates are highly atypical. This is clear from the bulk perspective, where the thin shell states are defined by very specific and spherically symmetric initial data in the black hole interior. At the same time, this is precisely what allows calculational control. Does the state-counting derivation of \eqref{eq:bhe} with the gravitational path integral hold for more generic families of states?}

    {\item {\it Microscopics of shell operators and overlaps}. In AdS/CFT, the microscopic description of the thin shell operator is complex and arguably not under full control. Consequently, the microscopic overlaps between thin-shell states are also subtle. The operator is roughly defined as follows \cite{Anous:2016kss}. Given a conformal primary $\phi_\Delta$ with $1 \ll \Delta \ll c$, where $c$ is the central charge of the CFT, one considers an approximately uniform arrangement of $O(c)$ local insertions of this primary around the spatial sphere. For convenience, one may moreover smear the resulting operator in the $s$-wave sector. This yields a surface operator that admits an effective semiclassical description in terms of a homogeneous thin shell of dust fluid with a given total rest mass $m$, which, in AdS units, is proportional to $\Delta$ and to the number of primary insertions on the sphere.
    The dynamics of such operators is not accessible from first principles in the CFT but is, in some sense, inferred from its effective gravitational dual. The issue is that this description is itself approximate and, moreover, corresponds not to a single microscopic operator but to an ensemble, as we discuss in the next item. However, important steps in reproducing the dynamics and statistics of these operators in CFT$_2$ have been given by solving for the relevant conformal block in the large-$c$ expansion \cite{Anous:2016kss, Chen:2024hqu,Liu:2025ztq}, or by relating these operators to Liouville line defects \cite{Chandra:2024vhm}.}

 {\item {\it Implicit statistics}. The gravitational path integral computation of the moments of the Gram matrix of overlaps between different microstates requires the inclusion of Euclidean wormholes. These wormholes, following the seminal works \cite{COLEMAN1988867,GIDDINGS1988890,Saad:2019lba}, impart an implicit statistical character to the calculation. Here, the statistics is interpreted to be hidden in an ensemble of microscopic operators modeling the behavior of the thin shell and compatible with semiclassics \cite{Sasieta:2022ksu,deBoer:2023vsm}. It is then challenging to provide a matching with the microscopic system by independently doing the boundary computations, at least in a way that does not implicitly rely on the gravitational effective field theory to define the CFT ensembles.}

{\item {\it Variance of the Hilbert space dimension}. Due to the previous item, even if small, the Hilbert space dimension might exhibit a variance across the specific microstate family. The gravitational computation of the variance may involve off-shell wormhole contributions, which might only be reliably evaluated in two dimensions and BPS scenarios \cite{Boruch:2023trc}.}

 {\item {\it Large mass limit}. The calculation is performed in the limit $m\ell_{\rm AdS} \rightarrow \infty$ ($\ell_{\rm AdS}$ is the AdS radius), where all shell masses and mass differences are taken to infinity. This limit simplifies the disk and wormhole contributions and ensures that interactions and crossings between different shells are suppressed. However, this limit is quite unphysical, and one would like to understand how the computation works when all shells have finite rest masses.}

 {\item {\it Semiclassical limits}. Recent micro state countings have used the semiclassical approximation, in which the system is described by an effective theory of gravity. This makes it difficult to understand how such a limit affects a putative exact microscopic computation of the dimension. It is then important to have exact methods, and later analyze how the semiclassical approximation interferes in the computation of the dimension.}
 
\end{enumerate}

These caveats are our main motivations. They suggest it would be most desirable to find a universal method to build and count microstates, namely one valid across all quantum systems, and that can be applied to chaotic systems in practice.  The standard approach, based on the counting of energy eigenstates, is not good enough since building particular energy eigenstates is out of reach in systems of interest. Finding such a method would clarify the recent gravitational path integral approaches, as well as demonstrate their robustness.
\vspace{0.5cm}

{\subsection*{This work:}}

The main objective of this work is to propose and develop an approach that avoids the previous issues. This approach originates in the new families of black hole microstates put forward recently in \cite{Magan:2024aet,Magan:2025hce}. More concretely, such microstates are constructed by introducing disordered, time-dependent matter sources in the Euclidean preparation of the states. Such operator insertions can be understood as effecting a certain Brownian motion in the black hole Hilbert space. The corresponding spatial wormholes are long and supported by numerous matter inhomogeneities. These wormholes were referred to as {\it Einstein-Rosen (ER) caterpillars} in \cite{Magan:2024aet,Magan:2025hce}, using a term coined in \cite{LennyTalk}. In this paper, we will refer to them simply as caterpillars, for short. From the bulk perspective, the matter distribution renders the initial data of a caterpillar generic compared to the thin shell states. At the level of the microstates, long caterpillars are quantitatively generic states, insofar as their microscopic wavefunctions exhibit a controlled and arbitrarily high degree of randomness \cite{Magan:2024aet,Magan:2025hce}.

Our first observation is that such a construction of ``pure thermal microstates'', with controlled degrees of randomness and statistics, actually works more broadly, and in principle can be used to build and count microstates in any quantum system. The only modification is the particular set of driving operators (the sources) generating the Brownian motion through the Hilbert space. Given this observation, our simple driving motto is as follows:
\vspace{0.5cm}

\noindent \emph{Present work's motto:} Start from a particular state, and let it follow a Brownian motion through the Hilbert space. At each time, the infinite number of different realizations of the Brownian couplings give rise to an infinite number of different states. Compute the dimension of the Hilbert space spanned by a finite number of those at each time.
\vspace{0.5cm}
\begin{figure}[h]
    \centering
    \includegraphics[width=0.55\linewidth]{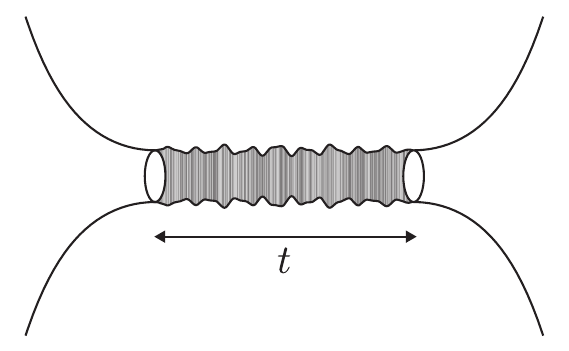}
   \caption{We will count families of caterpillars, corresponding to states of two black holes prepared for a fixed Euclidean time~$t$, with semiclassical ER bridge of length of order~$t$ supported by large numbers of erratic matter inhomogeneities. Using the gravitational path integral, we will show that, for any time (or wormhole length) larger than the scrambling time $t\gg t_\ast$, families of caterpillars span a finite-dimensional Hilbert space of dimension $e^{2S_{\rm BH}}$.}
    \label{fig:caterpi}
\end{figure}

When applied to black holes, this motto will explain their entropy by counting caterpillars, see fig. \ref{fig:caterpi}. But the approach has universal applicability, as described extensively below.

\vspace{0.5cm}

\noindent \textbf{Outline:}
\vspace{0.5cm}

Given this context and objective, the paper is organized as follows. Section~\ref{Sec:framework} introduces the framework: the Brownian motions, the families of states induced by them, the general formulas for evaluating their mutual overlaps, and the state-counting methods. Section~\ref{Sec:applications} considers analytically tractable quantum many-body systems, describing how the counting works both microscopically and through path integral methods in large-$N$ models. In particular, exact formulas at finite $N$ and finite times $t$ for the replica partition functions computing the statistics of overlaps are presented. Section~\ref{Sec:BHs} uses the gravitational path integral to count states of near-extremal black holes, where the analysis can be carried out in detail. There we will compute overlap moments from new wormhole saddles, and use those moments to show that black hole entropy is universally reproduced. We sketch the same computation in higher dimensions for general classes of black holes, where the problem maps to the construction of eternal traversable wormholes via double-trace couplings in higher dimensions. Section~\ref{SecVI} ends with a general discussion, highlighting how the present approach addresses the caveats described above, and commenting on open questions and directions. Finally, several appendices follow with further details.

\vspace{5mm}

\section{General framework: state ensembles and counting methods}
\label{Sec:framework}

This section describes the general framework. In particular we introduce the Brownian motions, the state ensembles induced by them, the particular quantities that need to be analyzed to derive Hilbert space dimensions, and the microstate counting methods.

\subsection{State ensembles}
\label{Sec:ensembles}

We start by reviewing the ensembles of quantum states introduced in \cite{Magan:2025hce,Magan:2024aet}. These ensembles will be used later in the state-counting applications. For convenience, we focus on states on two copies of a quantum system, although the single-copy generalization is straightforward.

\subsubsection{Infinite temperature}
\label{Sec:ensemblesinftemp}

Let $\mathcal{H}$ be a finite-dimensional Hilbert space of dimension $d$. Given a unitary $U(t)$ we consider its Choi state
\be\label{eq:inftempstate} 
\ket{U(t)} = \dfrac{1}{\sqrt{d}} \sum_{i,j=1}^d U(t)_{ij}\ket{i}\ket{j^*}\,.
\ee 
The state \eqref{eq:inftempstate} is a purification of the infinite-temperature density matrix on either subsystem. For us  $t$ parametrizes time as $U(t)$ is the time evolution operator generated via a time-dependent Hamiltonian
\be \label{eq:timeorderedinfinite}
U(t) = \mathsf{T}\left\lbrace \exp\left(-\iw \int_0^t \text{d}s \,H(s)\right)\right\rbrace \,.
\ee 
The ensemble of states at fixed time $t$ is specified by the choice of a disordered Brownian Hamiltonian, defined as follows. Let $\{\mathcal{O}_\alpha\}_{\alpha=1}^K$ be a set of Hermitian operators normalized such that $\mathrm{Tr}(\mathcal{O}_\alpha^2) = d$, ensuring that their eigenvalues are typically of $O(1)$ magnitude.  We consider time-dependent Hamiltonians of the form
\be\label{eq:BrownianH}
H(t) =  \sum_{\alpha =1}^K g_\alpha(t) \mathcal{O}_\alpha\,.
\ee
We further normalize the operators so that the couplings $g_\alpha(t)$ carry dimensions of energy. The couplings are drawn independently from a white-noise correlated Gaussian distribution,
\be\label{eq:statisticscouplings}
\overline{g_\alpha(t)}= 0\,,\quad \overline{g_\alpha(t)g_{\alpha’}(t’)} = J\delta_{\alpha\alpha’} \delta(t-t’)\,,
\ee
where $J$ is an energy scale setting the coupling variance. These couplings define an independent random Hamiltonian at each time. Upon exponentiation, the infinitesimal time-evolution implements a ``random unitary gate'' $\exp(-\iw\delta t\,  H(t))$. The dynamics is thus the continuous-time analog of a random quantum circuit \cite{Lashkari:2011yi}. Via \eqref{eq:inftempstate}, the disorder over couplings generates an ensemble of states at any fixed time $t$.

For rather general choices of drive operators in \eqref{eq:BrownianH}, the associated unitaries $U(t)$ will rapidly become generic. Therefore, the associated ensemble of states \eqref{eq:inftempstate} will also become generic for sufficiently large $t$. For example, if we consider a system of $N$ qubits where $d=2^N$ and restrict the drive operators to be few-body, $U(t)$ will be a scrambling unitary after some $O(\log N)$ time in units of $J^{-1}$ \cite{Lashkari:2011yi}. On the other hand, we do not expect $U(t)$ to become close to a totally generic Haar random unitary for much longer $O(2^N)$ timescales \cite{Jian:2022pvj,Guo:2024zmr}.\footnote{The degree of genericity, or randomness, of an ensemble of unitaries can be naturally characterized in terms of the statistical moments of the ensemble. This leads to the notion of a {\it unitary $k$-design}: an ensemble of unitaries with the same first $k$ moments as the Haar ensemble. Random quantum circuits are known to efficiently generate approximate unitary $k$-designs \cite{Harrow_2009,Brandao:2012zoj}, and $k$ is expected to grow linearly with circuit depth \cite{Emerson_2005,4262758,10.1063/1.2716992,Harrow_2009,Low:2010hyw,Brandao:2012zoj,Hunter-Jones:2019lps,Schuster:2024ajb}, although this linear scaling has not been rigorously proven. Brownian systems provide an advantage because the growth of $k$ can be analyzed explicitly: Jian, Bentsen, and Swingle \cite{Jian:2022pvj} showed that $k$ increases linearly with $t$. In \cite{Guo:2024zmr} this linear growth was shown to be robust with respect to the degree of locality of the Brownian Hamiltonian.} 

Given this state ensemble, we now want to derive concrete expressions for the overlaps. Consider two independent draws of the ensemble of states: $\ket{U_1(t)}$ and $\ket{U_2(t)}$. We will be interested in the overlap
\be\label{eq:overlapdef} 
\bra{U_1(t)}\ket{U_2(t)} = \dfrac{1}{d}\text{Tr}\left(U_1(t)^\dagger U_2(t)\right)\,.
\ee
Instead of facing the intractable problem of computing individual overlaps, we will compute their average value and higher moments over the ensemble of states by explicitly considering the disorder average over the random couplings. The advantage of using Brownian couplings is that the average value of the overlap exactly corresponds to an effective thermal partition function at inverse temperature $2t$,
\be\label{eq:id1} 
\overline{\bra{U_1(t)}\ket{U_2(t)}} = \dfrac{1}{d}\text{Tr}\left(e^{-2t H_{\text{eff}}}\right)\,,
\ee
for the single-replica effective Hamiltonian
\be\label{eq:effHgen1replica}
H_{\text{eff}} = \dfrac{J}{2} \sum_{\alpha=1}^K \Op_\alpha^2\,.
\ee
In Appendix \ref{app:Heff} we provide explicit details on how to derive this kind of relation. The essence is that the Brownian disorder of the real time evolution generates an Euclidean evolution driven by a time-independent effective Hamiltonian (also called the Liouvillian in the context of quantum Markov processes). This will also be true for higher moments, as we will see now, where the effective Hamiltonian in that case is defined in multiple replicas of the original system.

For \eqref{eq:effHgen1replica}, very generally, the drive operators will not all have zero eigenvalues, and $H_{\text{eff}}$ will be positive-definite with an $O(JK)$ minimum eigenvalue, as this follows from Weyl's inequality. Therefore, the average overlap between two different draws decays to zero exponentially fast with an $O(JK)$ exponent. To get an idea of the parametric timescales involved, for the system of $N$ qubits, it is natural to take $JK = O(N)$ so that the energy variance of the Brownian Hamiltonian is extensive in $N$. In that case, the average overlap is $O(\exp(-N))$ at timescales $Jt \gtrsim O(1)$ in the large-$N$ expansion.

Let us now focus on the second moment of the overlap. As represented in Fig. \ref{fig:heff1inftemp}, the square overlap, averaged over couplings, is a two-replica effective thermal partition function \cite{Jian:2022pvj}
\be\label{eq:squareoverlapavginftemp} 
\overline{|\bra{U_1(t)}\ket{U_2(t)}|^2} = \dfrac{1}{d^2}\text{Tr}\left(e^{-2t H_{\text{eff}}^{1\bar{1}}}\right)\,,
\ee
for the two-replica effective Hamiltonian
\be\label{eq:effHgen}
H_{\text{eff}}^{a\bar{b}} = \dfrac{J}{2} \sum_{\alpha=1}^K \left(\Op^{a}_\alpha-\Op_\alpha^{\bar{b}}\,^*\right)^2\,.
\ee
The trace in \eqref{eq:squareoverlapavginftemp} is taken over two replicas of the original Hilbert space, corresponding to the replicas labeled by $a=1$ and $\overline{b}=\overline{1}$, where the different replicas of the operators act in the effective Hamiltonian \eqref{eq:effHgen}. The bar  in the replica label represents the reverse time-orientation of the replica, with the corresponding operator in \eqref{eq:effHgen} appearing transposed, $(\Op_\alpha^{\bar{b}})^T = \Op_\alpha^{\bar{b}}\,^*$.

There are two features of the effective Hamiltonian \eqref{eq:effHgen} worth highlighting. First, it is an interacting Hamiltonian between replicas. The interactions arise solely from the classical disorder average, since the original (pre-averaged) quantity factorizes. Second, these effective interactions are bi-local, of the form $-J\sum_{\alpha=1}^K\Op_\alpha^{a}\Op_\alpha^{\bar{b}}\,^*$, as a direct consequence of the Gaussian nature of the random couplings in our definition of the Brownian Hamiltonian.

It is always true that $H_{\text{eff}}^{a\bar{b}}$ has as a ground state the infinite-temperature thermofield double (TFD) state between the replicas
\be\label{eq:inftempTFD} 
\ket{\mathbf{1}} = \dfrac{1}{\sqrt{d}} \sum_{i=1}^d \ket{i}_a\ket{i^*}_{\,\bar{b}}\,,
\ee
of zero energy, satisfying $H_{\text{eff}}^{a\bar{b}}\ket{\mathbf{1}} = 0$. The star in $|i^*\rangle$ denotes the action of an antiunitary operator. The TFD annihilates each term of the Hamiltonian, since $(\Op^{a}_\alpha-\Op_\alpha^{\bar{b}}\,^*)\ket{\mathbf{1}}=0$. Under very general conditions on the drive operators, namely that there is more than one ($K>1$), that no symmetry is preserved by all of them, and that no subset of degrees of freedom remains non-interacting, $\ket{\mathbf{1}}$ is the unique ground state of \eqref{eq:effHgen}.

As a consequence, at infinite times the thermal partition function in \eqref{eq:squareoverlapavginftemp} receives contributions only from the unique ground state, and the second moment of the overlap approaches
\be
\lim_{t\rightarrow \infty }\overline{|\bra{U_1(t)}\ket{U_2(t)}|^2} = \dfrac{1}{d^2}\,.
\ee
This coincides with the second moment of the overlap $\overline{|\braket{U_1}{U_2}|^2}$ between states generated by independent Haar-random unitaries $U_1$ and $U_2$. In this sense, the states have become generic, at least as diagnosed by the square overlap.\footnote{In the quantum information literature, the second moment of the overlap is known as the \textit{first frame potential} of the ensemble of unitaries (see \cite{Mele2024introductiontohaar} for a review). More generally, the \textit{$k$-th frame potential} is defined as 
$F_k(t) = d^{2k} \, \overline{|\langle U_1(t) | U_2(t) \rangle|^{2k}}$. 
For the Haar ensemble one has $F_k^{\text{Haar}} = k!$ for $k<d$.
}

\begin{figure}[h]
    \centering
    \includegraphics[width=0.55\linewidth]{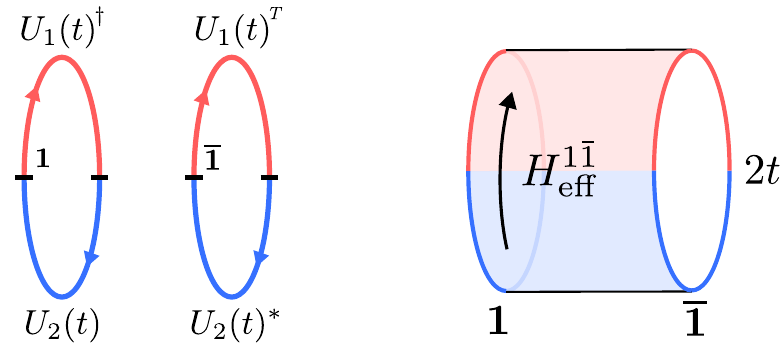}
    \caption{On the left, the quantity $|\mathrm{Tr}(U_1(t)^\dagger U_2(t))|^2$ is proportional to the square of the overlap between the two states. On the right, averaging over the Brownian couplings generates local-in-time effective interactions between the two replicas, resulting in Euclidean evolution governed by a time-independent Hamiltonian $H_{\mathrm{eff}}^{1\bar{1}}$. The average overlap is the thermal partition function of this effective Hamiltonian at inverse temperature $2t$.}
    \label{fig:heff1inftemp}
\end{figure}
 
In the next section, and through a different method in Appendix \ref{app:resol}, we will show that to extract the Hilbert space dimension from a state-counting calculation using the present ensembles, it suffices to consider the cyclic $n$-th moment of the overlap,
\be\label{eq:Zndefinftemp}
Z_n(t) = \overline{\bra{U_1(t)}\ket{U_2(t)}\bra{U_2(t)}\ket{U_3(t)}\cdots\bra{U_n(t)}\ket{U_1(t)}} \,.
\ee
For $n = 1$, we have $Z_1(t) = 1$, while for $n = 2$ we recover the case in \eqref{eq:squareoverlapavginftemp}. For general $n$, the moment \eqref{eq:Zndefinftemp} can be expressed as a matrix element in a $2n$-replica Hilbert space:
\be\label{eq:Znetaoverlapinftemp}
Z_n(t) = \bra{\eta} e^{-t H_{\text{eff}}^{1\bar{1}}} \otimes \cdots \otimes e^{-t H_{\text{eff}}^{n\bar{n}}} \ket{\eta}\,,
\ee
where, for any permutation $\sigma \in \text{Sym}(n)$, the state
\be\label{eq:defpermstate}
\ket{\sigma} \equiv \mathop{\otimes}\limits_{a=1}^n \ket{\mathbf{1}}_{{a\overline{\sigma(a)}}}
\ee
is built from pairwise infinite-temperature TFDs between the $a$-th and $\overline{\sigma(a)}$-th replicas, and $\eta = (12\cdots n)$ is the cyclic permutation. We explicitly write replica labels in $\ket{\mathbf{1}}_{a\overline{b}}$ to denote that the TFD is defined between the $a$-th and $\overline{b}$-th replicas. We illustrate \eqref{eq:Znetaoverlapinftemp} for $n=3$ in Fig. \ref{fig:Z3inftemp}.

\begin{figure}[h]
    \centering
    \includegraphics[width=.92\linewidth]{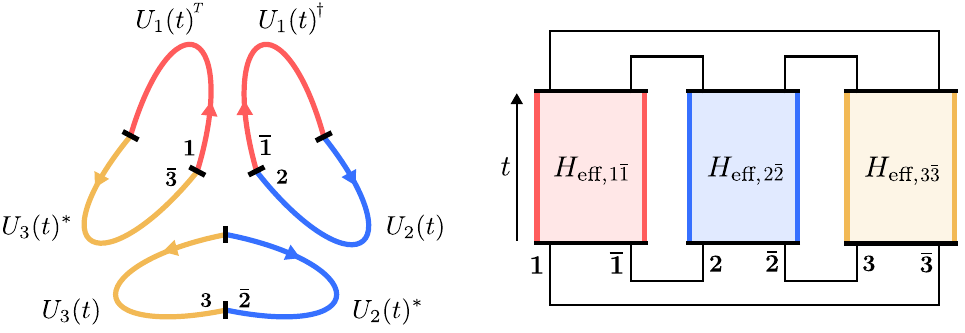}
    \caption{On the left, the quantity $\mathrm{Tr}(U_1(t)^\dagger U_2(t))\,\mathrm{Tr}(U_2(t)^\dagger U_3(t))\,\mathrm{Tr}(U_3(t)^\dagger U_1(t))$ is represented so that the time-ordering is aligned across replicas. On the right, averaging over the Brownian couplings generates local-in-time effective interactions between each pair of replicas, resulting in Euclidean evolution governed by time-independent Hamiltonians $H_{\mathrm{eff}}^{a\bar{a}}$ for $a = 1,2,3$. Following the replica structure of the original quantity, the resulting object $Z_3(t)$ corresponds to the survival amplitude under Euclidean evolution by time $t$ in the cyclic state $\ket{\eta}$ in the six-replica Hilbert space.}
    \label{fig:Z3inftemp}
\end{figure}

At infinite times, under the same assumptions as above, $Z_n(t)$ will generally approach the squared overlap between the cyclic state $\ket{\eta}$ and the identity permutation state $\ket{e}$,
\be
\lim\limits_{t\rightarrow\infty} Z_n(t) = \left|\bra{\eta}\ket{e} \right|^2 = \dfrac{1}{d^{2(n-1)}}\,.
\ee
This value again coincides with the value of the cyclic moment $\overline{\bra{U_1}\ket{U_2}\bra{U_2}\ket{U_3}\cdots\bra{U_n}\ket{U_1}}$ for unitaries $U_i$ drawn independently from the Haar distribution.

\subsubsection{Finite temperature}
\label{sec:finitetempstates}

We now generalize the ensemble of states introduced above to describe a finite-temperature equilibrium situation, where the notion of equilibrium is defined with respect to a reference Hamiltonian $H_0$. Although this generalization introduces many additional complications, it is necessary to describe finite-entropy black holes.

We introduce a finite-temperature operator $W(t,\beta) \equiv e^{-\frac{\beta}{4}H_0}W(t)e^{-\frac{\beta}{4}H_0}$ and consider its corresponding Choi state
\be \label{eq:state1}
\ket{W(t,\beta)} = \dfrac{1}{\sqrt{\mathcal{N}}} \sum_{i,j} e^{-\frac{\beta}{4}(E_i+E_j)}\,W(t)_{ij} \ket{E_i} |E_j^*\rangle\,,
\ee
where $\ket{E_i}$ are eigenstates of $H_0$ with eigenvalues $E_i$. For now, we treat $\beta$ as a free parameter controlling the energy of the state. The normalization factor of the state is
\be\label{eq:norm}
\mathcal{N} = \text{Tr}\left(e^{-\frac{\beta}{2}H_0} W(t)^\dagger e^{-\frac{\beta}{2}H_0} W(t)\right)\,.
\ee

The operator $W(t)$ is chosen to implement a finite-temperature version of the Brownian time evolution. One way to construct it, following \cite{Magan:2025hce,Magan:2024aet}, is to interchange a sequence of infinitesimal Lorentzian random evolutions, alternating with Euclidean cooling steps generated by $H_0$,
\be\label{eq:coolrcircuit}
W(t) = e^{-\delta \beta H_0}  U_I(n_t,\delta t)  e^{-\delta \beta H_0} \cdots e^{-\delta \beta H_0}  U_I(1,\delta t)\,.
\ee
The subscript $I$ in $U_I(j,\delta t)$ indicates that the Lorentzian evolution is driven by a purely ``interaction'' Brownian Hamiltonian of the form \eqref{eq:BrownianH}. The number of cooling steps is $n_t = t / \delta t$, and we consider the continuum limit $\delta t = \delta \beta \to 0$ (so that $n_t \to \infty$), where the relative scaling between $\delta t$ and $\delta \beta$ is absorbed into the definition of the coupling variance $J$ in \eqref{eq:statisticscouplings}.

Formally, in the continuum limit, we can equivalently define $W(t)$ in \eqref{eq:coolrcircuit} as a time-ordered exponential
\be\label{eq:W(t)torder} 
W(t) = \mathsf{T}\left\lbrace \exp\left(-\iw \int_0^t \text{d}s \,H(s)\right)\right\rbrace \,,
\ee  
driven by a time-dependent (non-Hermitian) Hamiltonian\footnote{In \cite{Magan:2025hce,Magan:2024aet} a prescription was given to implement this Hamiltonian in terms of the bare and total Hamiltonians using infinitesimal time-folds.}
\be\label{eq:coolrhamilt}
H(t) = -\iw H_0 + \sum_{\alpha=1}^K g_\alpha(t)\Op_\alpha \,.
\ee
The ensemble of states at fixed time, out of which $\ket{W(t,\beta)}$ is drawn, is induced by the white-noise correlated Gaussian random couplings \eqref{eq:BrownianH}.

Given this new state ensemble, we now want to derive expressions for the overlaps. Consider two independent draws from the ensemble of states: $\ket{W_1(t,\beta)}$ and $\ket{W_2(t,\beta)}$. We are interested in computing the overlap,
\be\label{eq:overlapfintemp} 
\bra{W_1(t,\beta) }\ket{W_2(t,\beta)} = \dfrac{1}{\sqrt{\mathcal{N}_1 \mathcal{N}_2}}\, \mathrm{Tr}\left(e^{-\frac{\beta}{2}H_0} W_1(t)^\dagger e^{-\frac{\beta}{2}H_0} W_2(t)\right)\,.
\ee
As before, we focus on the moments of the overlap \eqref{eq:overlapfintemp} averaged over the ensemble of states, corresponding to different realizations of the random couplings. Compared to the infinite-temperature case, the complication here is that the normalization depends on each realization, which makes the computation of the moments of \eqref{eq:overlapfintemp} more subtle. In what follows, we compute the ``annealed averages’’ of \eqref{eq:overlapfintemp}, where we consider states normalized on average by the average norm over the ensemble $\overline{\mathcal{N}}$. In practice, in \eqref{eq:overlapfintemp} this implies taking the average of the numerator, and replacing the denominator by $\overline{\mathcal{N}}$. As we will later emphasize in Section \ref{Sec:countingframework}, this distinction is irrelevant for the purposes of this paper. Thus, from now on, in the finite-temperature setting an overline will implicitly denote an annealed average.

First of all, note that we can write the relevant quantities in terms of survival amplitudes
\begin{gather}
\mathrm{Tr}\left(e^{-\frac{\beta}{2}H_0} W_1(t)^\dagger e^{-\frac{\beta}{2}H_0} W_2(t)\right) = Z(\beta) \bra{\text{TFD}}  W_2(t) \otimes W_1(t)^*\ket{\text{TFD}}\,,\label{eq:numfintemppreavg}\\[.2cm]
\mathcal{N} = Z(\beta) \bra{\text{TFD}} W(t)\otimes W(t)^*\ket{\text{TFD}}\,.\label{eq:normfintemppreavg}
\end{gather}
Here we have introduced the finite-temperature TFD state
\be
\ket{\text{TFD}} = \dfrac{1}{\sqrt{Z(\beta)}} \sum_{i} e^{-\frac{\beta}{2} E_i}\, \ket{E_i} \ket{E_i^*}\,,
\ee
and $Z(\beta)=\mathrm{Tr}(e^{-\beta H_0})$ is the canonical partition function.

The average value of the overlap is
\be\label{eq:avoverlapfintemp} 
\overline{\bra{W_1(t,\beta) }\ket{W_2(t,\beta)}} = \dfrac{\bra{\text{TFD}}  e^{-tH_{\eff}} \otimes e^{-tH_{\eff}^*}\ket{\text{TFD}}}{\bra{\text{TFD}}  e^{-tH_{\eff}^{1\bar{1}}} \ket{\text{TFD}}}\,,
\ee
with the single- and two-replica finite-temperature effective Hamiltonians given by
\begin{gather}
H_{\mathrm{eff}} = H_0 + \dfrac{J}{2} \sum_{\alpha=1}^K \Op_\alpha^{2}\,,\\
H_{\mathrm{eff}}^{a\bar{b}} = H_0^{a} + H_0^{\bar{b}}\,^* + \dfrac{J}{2} \sum_{\alpha=1}^K \left(\Op_\alpha^{a} - \Op_\alpha^{\bar{b}*} \right)^2\,.\label{eq:tworeplicaeffHfintemp}
\end{gather}
An explicit derivation of the relevant moments of $W(t)$, arriving at these relations is provided in Appendix~\ref{app:Heff}. Below we will use the same notation for the effective Hamiltonians as in the infinite-temperature case; the intended meaning should be clear from the context.

Again, as discussed in the next section, we will be interested in computing the cyclic $n$-th moment of the overlap,
\be\label{eq:znfintemp}
Z_n(t,\beta) = \overline{\bra{W_1(t,\beta) }\ket{W_2(t,\beta)}\bra{W_2(t,\beta) }\ket{W_3(t,\beta)}...\bra{W_n(t,\beta) }\ket{W_1(t,\beta)}}\,.
\ee
This moment can be written in terms of the two-replica effective Hamiltonian as
\be\label{eq:Znetaoverlapfintemp}
Z_n(t,\beta) = \dfrac{\bra{\eta,\beta} e^{-tH_{\eff}^{1\bar{1}}}\otimes ...\otimes e^{-tH_{\eff}^{n\bar{n}}}\ket{\eta,\beta}}{(\bra{\text{TFD}}  e^{-tH_{\eff}^{1\bar{1}}} \ket{\text{TFD}})^n}\,.
\ee
where, for any permutation $\sigma \in \text{Sym}(n)$, we define
\be\label{eq:defpermstatefintemp}
\ket{\sigma,\beta} = \mathop{\otimes}\limits_{a=1}^n \ket{\text{TFD}}_{{a\overline{\sigma(a)}}}\,,
\ee
and again $\eta = (12...n)$ is the cyclic permutation.

\subsection{Microstate counting methods}
\label{Sec:countingframework}

We now outline the methods used to determine the Hilbert space dimension spanned by families of states from their mutual overlaps, and then describe the expectations that follow when applying them to the ensembles introduced in the previous section.

\subsubsection{Gram matrices and Hilbert space dimensions}
\label{Sec:countingframeworkdim}

Starting from a finite-dimensional Hilbert space $\mathcal{H}$, we consider a family of $\Omega$ states in it
\be
\mathsf{F}_\Omega = \lbrace \ket{\Psi_{i}}\in \mathcal{H}: i=1, \dots, \Omega\rbrace\,.
\ee
Our goal is to extract the dimension of the Hilbert subspace that these states span, 
\be
d_\Omega = \text{dim}(\text{span}\lbrace\mathsf{F}_\Omega\rbrace)\,.
\ee
To do this, we consider
the $\Omega \times \Omega$ Gram matrix $G$ of overlaps between states in the family,
\be\label{eq:Gram}
G_{ij} = \bra{\Psi_{i}} \ket{\Psi_{j}}\,, \qquad i,j = 1,\dots ,\Omega\,.
\ee
The dimension of interest is the rank of the Gram matrix
\be\label{eq:Gramdimensionrankker}
d_\Omega = \text{rank}(G) = \Omega - \text{Ker}(G)\,.
\ee 
In turn, the $\text{rank}(G)$ can be evaluated in ways other than direct diagonalization. One particularly convenient method is to analytically continue the moments of the Gram matrix and evaluate its rank via the limit
\be\label{eq:n0limitZn}
d_\Omega = \lim_{n\rightarrow 0}\, \text{Tr}(G^n)\,.
\ee
The approach of extracting the dimension from the overlap moments is especially natural in statistical settings where $G$ itself is a random matrix drawn from an ensemble. In our case, the randomness of $G$ originates from the ensemble of microstates prepared by the explicitly disordered couplings used to define the Brownian motion.

More concretely, we consider families of $\Omega$ states $\mathsf{F}_\Omega = \lbrace \ket{\Psi_{i}}\in \mathcal{H}: i=1, \dots, \Omega\rbrace$ drawn randomly and independently from one of the ensembles described above. We seek to compute $\text{Tr}(G^n)$, and then average the moment over different realization of random couplings. First, consider the limit in which we take an infinite number of draws, $\Omega \rightarrow \infty$, with all other parameters held fixed. In this case, the dominant contribution to $\mathrm{Tr}(G^n)$ comes from terms in the sum with all indices distinct, as terms with repeated indices are suppressed by powers of $1/\Omega$ in their multiplicity. After averaging:
\be\label{zncyc}
\overline{\mathrm{Tr}(G^n)} = \sum_{i_1 \neq \cdots \neq i_n = 1}^\Omega \overline{\bra{\Psi_{i_1}}\ket{\Psi_{i_2}} \bra{\Psi_{i_2} }\ket{\Psi_{i_3}} \cdots \bra{\Psi_{i_n} }\ket{\Psi_{i_1}}} = {\Omega \choose n} Z_n\,,
\ee
where $Z_n$ is the cyclic moment,
\be
Z_n \equiv \overline{\bra{\Psi_1} \ket{\Psi_2} \bra{\Psi_2 }\ket{ \Psi_3} \cdots \bra{\Psi_n }\ket{\Psi_1}}\,.
\ee
Recall that we introduced these cyclic moments for the ensembles of states in Section~\ref{Sec:ensembles}. Using  \eqref{eq:n0limitZn}, we find that the dimension is given by
\be\label{eq:dimensioninfty}
d_\infty = \lim_{\Omega \rightarrow \infty} \lim_{n \rightarrow 0} \dfrac{\Gamma(\Omega+1)}{\Gamma(n+1)\Gamma(\Omega - n + 1)} Z_n = \lim_{n \rightarrow 0} Z_n\,.
\ee
In this context, the moment $Z_n$ can be understood as the $n$-th R\'enyi entropy of the average state $\rho \equiv \overline{\ket{\Psi}\bra{\Psi}}$, given by $Z_n = \text{Tr}(\rho^n)$, as we discuss more properly below. The analytic continuation and corresponding limit $\lim\limits_{n\rightarrow 0}Z_n$ computes the rank of $\rho$.

We note that we do not write an average on the left-hand side intentionally, since it is straightforward to show that the variance of $d_\infty$ vanishes in this limit. This follows from the factorization
\be
\overline{\mathrm{Tr}(G^n)\,\mathrm{Tr}(G^m)} = \overline{\mathrm{Tr}(G^n)}\,\overline{\mathrm{Tr}(G^m)} \qquad \forall\, n, m \geq 1\,,
\ee
in the limit of large $\Omega$, as connected index contractions contribute fewer powers of $\Omega$.

But, in turn, the result \eqref{eq:dimensioninfty} already shows that the dimension of the linear span of any $\mathsf{F}_\Omega$ with $\Omega$ finite, is, with unit probability, bounded above by $d_\infty$, which must be finite if the states live in finite-dimensional Hilbert spaces or in finite-entropy microcanonical windows.

To understand the dimension at finite $\Omega$ generally, one may invoke the following observation. Let $\mathcal{H}$ be a finite-dimensional Hilbert space of dimension $d$, and let 
$\ket{\Psi_1},\ldots,\ket{\Psi_\Omega}\in \mathcal{H}$ be random vectors with joint distribution $\mu$ on $\mathcal{H}^{\otimes \Omega}$. Suppose $\mu$ is absolutely continuous with respect to the product Haar measure on $\mathcal{H}^{\otimes \Omega}$. Then, with unit probability, the rank $d_{\Omega}$ of the Gram matrix satisfies
\be\label{eq:thmdim-dep}
       d_{\Omega} = \min\lbrace \Omega, d \rbrace\,.
\ee
Assume first $\Omega<d$. Then $d_\Omega=\Omega$ unless $\det(G)=0$, where $\det(G)$ is a real-analytic polynomial of total degree $2\Omega$ in the wavefunctions. Since this polynomial is not identically zero, its vanishing set has Haar measure zero, and thus probability zero for any absolutely continuous distribution. If $\Omega\ge d$, then $d_\Omega\le d$ by construction, while the alternative $d_\Omega<d$ would require all $d\times d$ minors of $G$ to vanish. These minors are nontrivial real-analytic polynomials (of degree $2d$ in the wavefunctions), whose common zero set has Haar measure zero.

More intuitively, this is stating something self-evident. Consider say a Hilbert space of $d=2$ and $\Omega=2$. For absolutely continuous distributions, the probability that in a particular realization both vectors are exactly parallel is zero, in the same way that the probability $p(x=y)=0$ for any $y$, for an  absolutely continuous distribution $p(x)$ in the real line. Equivalently, the probability distributions  $\mu_2$ of two vectors on $\mathcal{H}^{\otimes 2}$ that fall outside the scope of this observation must be of the form $\mu_2\sim \delta (\vert\psi_1\rangle-\vert\psi_2\rangle) \mu_1 (\vert\psi_1\rangle)$, and are thus singular as functions.

From this observation it follows that the saturation value \eqref{eq:dimensioninfty} computed from the overlaps indeed corresponds to the Hilbert space dimension where the states live. That is we have
\be\label{eq:transitiondimension}
d_{\Omega} = \min\lbrace \Omega, \lim_{n \rightarrow 0} Z_n\rbrace \,,
\ee
where we omit the average in the left hand side again, as the result holds with unit probability under the stated assumption that the $\Omega$ states are drawn from a smooth probability distribution. Equivalently, although the Gram matrix is random, the dimension spanned by the $\Omega$ states is not random, i.e. it displays zero variance. The ensembles defined above using Brownian motions are expected to be described by absolutely continuous probability distributions in the Hilbert space since they emerge from Gaussian averages over real couplings appearing in unitary operators. We verify the variance of $d_{\Omega}$ vanishes for Brownian circuits numerically in Appendix \ref{app:rank}.

We note that, strictly speaking, the states constructed in Section \ref{Sec:ensemblesinftemp} are maximally entangled and thus belong to a submanifold $\mathcal{M}\subset \mathcal{H}$. But the same observation applies for the induced (Riemannian) measure on $\mathcal M$ because the linear span of $\mathcal{M}$ is the full Hilbert space. Then, another simple consequence is that if we project the states $\mathsf{F}_\Omega = \lbrace \ket{\Psi_{i}}\in \mathcal{H}: i=1, \dots, \Omega\rbrace$ to a proper subspace $\tilde{\mathcal{H}}\subset \mathcal{H}$, where $d_{\tilde{\mathcal{H}}}<d_{\mathcal{H}}$, then the new family
\be
\mathsf{F}_\Omega^{P_{\tilde{\mathcal{H}}}} = \lbrace P_{\tilde{\mathcal{H}}}\ket{\Psi_{i}}\in \mathcal{N}: i=1, \dots, \Omega\rbrace \,,
\ee
where $P_{\tilde{\mathcal{H}}}$ is the projector into the subspace, spans a Hilbert space with dimension
\be\label{eq:transitiondimensionsub}
d_{\Omega} = \min\lbrace \Omega, \lim_{n \rightarrow 0} Z_n^{P_{\mathcal{N}}}\rbrace \,,
\ee
where $Z_n^{P_{\mathcal{N}}}$ is the cyclic moment of the projected states. Again, the reason is that there is measure zero for any state of the ensemble to have a zero projection into any particular direction of the Hilbert subspace if the distribution is absolutely continuous. This is useful to understand the variance of the state-countings in finite temperature scenarios, where a projection into a given microcanonical window is required. We will comment on this in due time below.

We emphasize that none of the considerations in this section depend on whether the states are normalized or not. Indeed, a generic rescaling $\ket{\Psi_i}\rightarrow \lambda_i \ket{\Psi_i}$ leaves the resulting dimension unchanged. While this is also obvious, it makes clear that using annealed averages for the finite-temperature ensembles introduced in Section \ref{sec:finitetempstates}, which amounts to normalizing the states only on average, does not affect the dimension of the Hilbert space spanned.

We end by noting that, since Ref. \cite{Penington:2019kki}, state-countings methods have proceed by computing the resolvent of the Gram matrix of overlaps, allowing a derivation of the density of states and associated rank of such matrix. It is worth emphasizing that such a method is not needed in the present approach. At every time, we can choose as many realizations of states as we want, and the moments of the Gram matrix simplify enormously as in \ref{zncyc}. Equivalently, we do not need to take a limit in which $\Omega$ and $d^2$ grow large with the ratio fixed. Such limit forces us to sum all planar diagrams. Here, since we have microscopic control over the state preparation and associated statistics at finite $t$, and we know the probability is continuous, we can just take the limit $\Omega\gg d^2$, in which the previous replica partition function \ref{zncyc} becomes isolated. For finite $\Omega$ we just use the previous observation. Still, in appendix \ref{app:resol} we describe the resolvent-type computations for the present state ensembles. This lead to an improved Schwinger-Dyson equation that takes into account non-vanishing one-point functions of the inner products.

\subsubsection{Expected behavior for the ensembles}
\label{Sec:countingframeworkexpectation}

In the next sections we will apply the framework of \eqref{eq:transitiondimension} to the ensembles described in Section~\ref{Sec:ensembles}, with the aim of extracting the dimension spanned by their mutual overlaps. Let us build first the most basic intuition for what to expect. Consider the infinite-temperature case of Section \ref{Sec:countingframeworkdim}, where we look at a family of $\Omega$ draws taken at a fixed time $t$,
\be\label{eq:familyinftemp} 
\mathsf{F}_\Omega^t = \lbrace \ket{U_{i}(t)}\in \mathcal{H}: i=1, \dots, \Omega\rbrace\,.
\ee 
The Brownian couplings generate, for any time $t > 0$, an ensemble of states defined by a smooth probability distribution in the Hilbert space $\mathcal{H}$. Then, the probability of linear dependence for an undercomplete set of independent states drawn from this distribution is zero. Therefore, using \eqref{eq:transitiondimension}, we expect that, with unit probability, the family $\mathsf{F}_\Omega^t$ spans a subspace of dimension
\be\label{eq:saturation}
d_\Omega(t) = \min\lbrace \Omega, d^2 \rbrace \qquad \forall\, t > 0\,.
\ee
At $t = 0$ all draws coincide with the infinite-temperature TFD. The distribution at this time is a delta-function peaked around the initial state. We trivially have $d_\Omega(0) = 1$. The interesting development here is that we get a basis of the Hilbert space for any $t > 0$, no matter how small $t$ is. Nevertheless, we note the Choi states $\ket{U_i(t)}$ become mutually distinguishable and attain a very low fidelity only after a finite time.

At finite temperature, the situation is more subtle. We consider the family of states 
\be\label{eq:familyfintemp} 
\mathsf{F}_\Omega^{t,\beta} = \lbrace \ket{W_i(t,\beta)}\in \mathcal{H}: i=1, \dots, \Omega\rbrace\,.
\ee 
It was shown in \cite{Magan:2024aet} that the average state $\overline{\vert \psi\rangle\langle \psi\vert}$ approaches $\vert \rho_{\text{eq}}\rangle\langle \rho_{\text{eq}}\vert$ at large times, with a stationary density matrix $\rho_{\text{eq}}$, whose specific form will be reviewed later below and derived in the examples. Then, in these scenarios we expect that the dimension $d$ appearing in \eqref{eq:saturation} will be replaced by $\text{rank}(\rho_{\text{eq}})$  following the same argument. The variance again vanishes.

We will verify these expectations analytically in several models below and in gravity. Numerical verifications can be found in Appendix~\ref{app:rank}.

\subsubsection{Comparison to random quantum circuits}

Before proceeding, let us compare the above setting to the discrete-time analogue of the Brownian time evolution: a random quantum circuit. Consider a circuit model on $N$ qubits defined by a finite universal gate set $\mathcal{G}$ of two-qubit gates (e.g. the two-qubit Clifford group together with the $\tfrac{\pi}{4}$ phase shift $T$), which may act on different pairs of qubits. Each layer applies $K$ such gates, and time is measured in units of circuit time. At each time step, the gates are chosen independently and uniformly at random from $\mathcal{G}$, each with probability $1/|\mathcal{G}|$.

\begin{figure}[h]
\centering
\includegraphics[width=.7\linewidth]{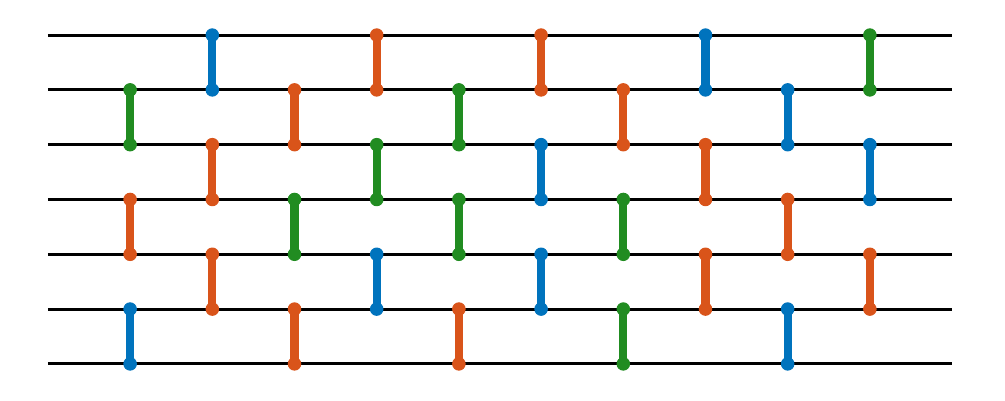}
\caption{A random quantum circuit preparing a maximally entangled state $\ket{U(t)}$.}
\label{fig:randcirc}
\end{figure}

Since $\mathcal{G}$ consists of two-body gates, we take $K \gtrsim O(N)$ so that $O(N)$ qubits interact in every layer. An example of a random quantum circuit with gates restricted to act locally (on nearest neighbors) is shown in Fig.~\ref{fig:randcirc}. Each realization of the random circuit yields a unitary $U(t)$ at circuit depth $t$, which defines a maximally entangled two-sided state $\ket{U(t)}$. Since the random circuit serves as the discrete-time counterpart of Brownian time evolution, this setup provides the discrete-time analog of the infinite-temperature two-sided states introduced in Section~\ref{Sec:ensemblesinftemp}.

Given the discreteness of the quantum circuit, the total number of distinct states that the random circuit can generate at finite time $t$ is bounded above by
\be
\Omega_{\max}(t) = |\mathcal{G}|^{tK}\,.
\ee
Under this restriction, the argument above implies that a family of states obtained by drawing $\Omega$ samples at fixed depth $t$ will, with unit probability, span a Hilbert subspace of dimension
\be\label{eq:saturationdiscreterc}
d_\Omega(t) = \min\left\lbrace \Omega,\ \Omega_{\max}(t),\ 4^N \right\rbrace \qquad \forall\, t \geq 0\,.
\ee
Hence, generating the full Hilbert space requires a circuit depth of at least
\be
t > \dfrac{2N}{K\log_2 |\mathcal{G}|}\,.
\ee
For the example in Fig.~\ref{fig:randcirc} with $N = 7$ and $|\mathcal{G}| = K = 3$, this yields $t \gtrsim 2.9$.

Heuristically, the Brownian motion corresponds to the limit $|\mathcal{G}| \to \infty$, since the continuous random couplings effectively generate an infinite gate set. This allows the Hilbert space to be spanned instantaneously. This provides an intuitive explanation of \ref{eq:saturation}.

\section{Counting states of many-body quantum systems}
\label{Sec:applications}
 
Having established the general framework, in this section we illustrate its use on a sequence of quantum many-body systems. We begin by 
analytically tractable cases where we obtain exact expressions for the overlap moments $Z_n(t)$, and use them to determine the corresponding dimension through \eqref{eq:transitiondimension}.  We the study the problem in the SYK scenario using techniques that are conceptually nearer the gravitational one, i.e. we analyze $Z_n(t)$ using large-$N$ saddle-point approximations to collective field path integrals. This connects directly to the computation of black hole entropy from spacetime wormhole contributions to the overlap moments, the focus of the next section. We end by describing a system-independent derivation of the dimension that applies more generally to families drawn from any of the ensembles introduced in Section~\ref{Sec:ensemblesinftemp}.

\subsection{Microscopic overlap moments and dimension}

In each of the cases below, we study the family
\be
\mathsf{F}_\Omega^t = \{\ket{U_i(t)} \in \mathcal{H} : i = 1, \dots, \Omega\}\,,
\ee
whose members are independent draws from the ensemble defined in Section~\ref{Sec:ensemblesinftemp}, determined by a specific choice of Brownian Hamiltonian. For these ensembles, we compute the overlap moments $Z_n(t)$ through explicit diagonalization of the effective Hamiltonian, and obtain the corresponding dimension $d_\Omega(t)= \text{dim}(\text{span}(\mathsf{F}_\Omega^t))$ via analytic continuation of the moments \eqref{eq:transitiondimension}.

\subsubsection{Brownian GUE}
\label{sec:BrownianGUE}

We start by considering a general Hilbert space $\mathcal{H}$ of dimension $d$, and a time-evolution unitary $U(t)$ generated by time-dependent Hamiltonians drawn from the Brownian Gaussian Unitary Ensemble (Brownian GUE) \cite{Guo:2024zmr}; see also \cite{Tang:2024kpv}. Concretely, at each time step the Hamiltonian is drawn independently from the GUE. Equivalently, the operators $\Op_\alpha$ in \eqref{eq:BrownianH} correspond to the $d^2$ matrix elements in an orthonormal basis $\lbrace \ket{i}\rbrace_{i=1}^d$, with $\alpha = (ij)$ labeling each matrix element. Expanding the Hamiltonian explicitly, we have
\be\label{Hgue}
H(t) = \sum_{i = 1}^d x_{ii}(t) \Pi_{ii} + \sum_{i<j}^d \left[ x_{ij}(t) \left( \frac{\Pi_{ij} + \Pi_{ji}}{\sqrt{2}} \right) + \iw \, y_{ij}(t) \left( \frac{\Pi_{ij} - \Pi_{ji}}{\sqrt{2}} \right) \right]\,,
\ee
where $\Pi_{ij}=\ket{i}\bra{j}$. The real ``couplings'' $x_{ii}(t)$, $x_{ij}(t)$, and $y_{ij}(t)$ are taken to be independent white-noise correlated Gaussian random variables with zero mean and variance
\be\label{eq:couplingsBGUE}
\overline{x_{ij}(t)x_{kl}(t’)} = \overline{y_{ij}(t)y_{kl}(t’)} = \frac{J}{d}\,\delta(t-t’)\,\delta_{ik}\delta_{jl}\,.
\ee
We choose this normalization of $J$ to avoid carrying extra factors of $d$ below.

In this case, the two-replica effective Hamiltonian is proportional to the orthogonal projector into perpendicular direction to the infinite-temperature TFD state \cite{Guo:2024zmr}
\be \label{eq:HeffGUE}
H_{\text{eff}}^{a\bar{b}} = J \left(\mathbf{1} - \ket{\mathbf{1}}\bra{\mathbf{1}}\right)\,.
\ee
Then, as anticipated from general considerations, the infinite-temperature TFD state $\ket{\mathbf{1}}$ is the unique ground state of the effective Hamiltonian \eqref{eq:HeffGUE}. All remaining eigenstates have energy $\Egap = J$, so the first excited level is highly degenerate, with $N_\ast = d^2 - 1$. Consequently, the Euclidean evolution operator is
\be\label{eq:efftevolBGUE}
e^{-tH_{\text{eff}}^{a\bar{b}}} = \left(1 - e^{- J t}\right)\ket{\mathbf{1}}\bra{\mathbf{1}} + e^{- J t}\,\mathbf{1} \,.
\ee
Plugging this into \eqref{eq:squareoverlapavginftemp} and taking the trace, we obtain the second moment of the overlap
\be
Z_2(t) = \dfrac{1}{d^2}\left(1 + (d^2 - 1)e^{-2J t}\right)\,.
\ee
For $n>2$, we can also obtain closed-form expressions for $Z_n(t)$. Starting from \eqref{eq:Znetaoverlapinftemp}, we replace \eqref{eq:efftevolBGUE} for the Euclidean time-evolution operators. In total we have $2^n$ terms, depending on whether we choose the identity or the projector \eqref{eq:efftevolBGUE} in each replica. Grouping terms by the number $p$ of orthogonal projectors, the contribution for a fixed $p<n$ is
\be\label{eq:znbcontrib}
\frac{1}{d^{n+p}} \left(1 - e^{-J t}\right)^p e^{-(n-p)J t}\, d^{n-p} = \left(\frac{1 - e^{-J t}}{d^2}\right)^{p} e^{-(n-p)J t}\,.
\ee
On the left-hand side, the factor $d^{-(n+p)}$ comes from the normalization of $\ket{\eta}$ together with the $p$ orthogonal projectors, whereas the factor $d^{\,n-p}$ counts the number of closed index loops when computing the corresponding matrix element. A separate treatment is required for $p=n$, where the contribution is
\be
\frac{1}{d^{2n}} (1 - e^{-J t})^n\, d^2 = \frac{1}{d^{2(n-1)}}\left(1 - e^{-J t}\right)^n\,.
\ee
The extra factor of $d^2$ relative to \eqref{eq:znbcontrib} comes from having two closed index loops rather than none; here the overlap reduces to $|\bra{\eta}\ket{e}|^2$.

Summing these two types of contributions yields
\begin{align}
Z_n(t) &= \sum_{p=0}^{n-1} \binom{n}{p} \left(\frac{1 - e^{-J t}}{d^2}\right)^{p} e^{-(n-p)J t} + \frac{1}{d^{2(n-1)}}\left(1 - e^{-J t}\right)^n \nonumber \\
&=  \dfrac{1}{d^{2n}} \left[(1 + e^{-J t}(d^2 - 1))^{n} + (d^2 - 1)(1 - e^{-J t})^{n}\right]\,, \label{eq:ZnBrownianGUE}
\end{align}
where in the last step we used the binomial expansion.

To evaluate the dimension $d_\Omega$ spanned by the family $\mathsf{F}_\Omega^t$, we use \eqref{eq:transitiondimension} and substitute the analytic continuation of \eqref{eq:ZnBrownianGUE} as $n\to 0$. There are two cases. At $t=0$, \eqref{eq:ZnBrownianGUE} gives $Z_n(0)=1$, whose analytic continuation is $1$. For $t>0$, the analytic continuation yields $d^{2}$. Altogether,
\be
d_\Omega(t)=\min\!\left\{\Omega,\lim_{n\to 0} Z_n(t)\right\}
=\begin{cases}
1\, & t=0,\\[2pt]
\min\{\Omega,d^2\}\, & t>0.
\end{cases}
\ee
This agrees with the general expectation discussed in Section \ref{Sec:countingframeworkexpectation}.

\subsubsection{Brownian spin cluster}
\label{S:brownianspincluster}

We now consider a system of $N$ spins and restrict the drive operators $\Op_\alpha$ in the time-dependent Hamiltonian to be $q$-body. With this choice, we expand the Hamiltonian in the $q$-body sector of the Pauli basis. Let $\{\sigma_0,\sigma_1,\sigma_2,\sigma_3\}$ denote the single-spin identity and Pauli matrices, and for a multi-index $\alpha=(\alpha_1,\ldots,\alpha_N)\in\{0,1,2,3\}^N$ define the Pauli string $P_\alpha = \sigma_{\alpha_1}\otimes ... \otimes \sigma_{\alpha_N}$. The Hamiltonian takes the form
\be
    H(t) = \sum_{|\alpha|=q} g_\alpha (t) P_\alpha\,,
\label{eq:hamiltonianspincluster}
\ee
where we have defined the weight $|\alpha|$ of the Pauli string as the number of non-identity elements $|\alpha| = \sum_{i=1}^N(1-\delta_{\alpha_i0})$. The couplings $g_\alpha(t)$ are independent Brownian Gaussian couplings with zero mean and variance given by \eqref{eq:statisticscouplings}. In this case, the two-replica effective Hamiltonian is \cite{Jian:2022pvj}
\be\label{eq:Heffspins}
    H_{\text{eff}}^{a\bar{b}} = \frac{J}{2} \sum_{|\alpha|=q} \left (P_\alpha^a - P_\alpha^{\bar{b}\,*} \right )^2\,.
\ee
For $q>0$, it is straightforward to see that $\ket{\mathbf{1}}$ is the unique ground state of the effective Hamiltonian, as expected. To see this, consider the operator $A$ associated with the Choi state $\ket{A}$. Because the Hamiltonian is frustration-free, the ground state condition for $\ket{A}$ is essentially that $[A,P_\alpha]=0$ for all Pauli strings of length $|\alpha| =q$. However, the $q$-body Pauli strings generate the full operator algebra. It follows by Schur’s lemma that the only operator commuting with all of them is the identity up to a scalar.

The diagonalization of $H_{\text{eff}}^{a\bar{b}}$ was performed in \cite{Jian:2022pvj} for $q=4$; we summarize the result here for completeness and generalize to $q$-body. First, the Choi states of the Pauli strings $\ket{P_\alpha}$ are the eigenstates of $H_{\text{eff}}^{a\bar b}$. To see this, note the commutation relation of the Pauli strings
\be
P_\alpha P_\gamma = (-1)^{k(\alpha, \gamma)} P_\gamma P_\alpha\,,
\label{commutation}
\ee
where $0 \leq k(\alpha, \gamma)\leq N$ is the number of sites where $P_\alpha$ and $P_\gamma$ have different non-identity Pauli matrices. For example, for $\alpha = (1,2,3,0,0,…,0)$ and $\gamma = (1,0,1,0,0,…,0)$, we have $k(\alpha, \gamma) = 1$ since there are distinct non-identity Pauli matrices only at the third site. Using this relation, it is straightforward to see that $H_{\eff}^{a\bar{b}}\ket{P_\alpha} = E_\alpha \ket{P_\alpha}$ with
\be\label{eq:specBspinHeff}
E_\alpha = J \sum_{|\gamma| = q } \left (1- (-1)^{k(\gamma, \alpha)} \right)\,.
\ee
Since the states $\ket{P_\alpha}$ form an orthonormal basis of the two-sided Hilbert space, these are all possible eigenstates. 

Now, in \eqref{eq:specBspinHeff} the resulting energy $E_\alpha$ can only depend on permutation-invariant data, i.e., on the size of the Pauli string $|\alpha|$. The spectrum thus consists of values $E(s)$ for $s=0,1,...,N$, where $s$ represents the size of the corresponding strings. We find that:
\begin{lemma}\label{lemma:bspin1}
The Hamiltonian \eqref{eq:Heffspins} is diagonal in the Pauli string basis ${\ket{P_\alpha}}$, with eigenvalues
\be\label{eq:specBspinHefffinal}
E(s) = JK\!\left[1 - {}_2F_1\!\left(-q,\,-s;\,-N;\,\tfrac{4}{3}\right)\right]\,,
\ee
where $s = |\alpha|$ and $K = 3^q {N \choose q}$. Each level is degenerate, with multiplicity $3^s{N \choose s} $.
\end{lemma}
We leave the proof to Appendix \ref{app:exactBrownian}, where we also generalize the analysis for scenarios where there is at most $q$-local interactions.

Note that the spectrum \eqref{eq:specBspinHefffinal} predicts the following gap for the effective Hamiltonian
\be\label{eq:gapBrownianspin} 
\Egap = \dfrac{4q}{3} \,\J \,,
\ee 
where we have defined the rescaled coupling 
\be\label{eq:normJ}
\J  = \dfrac{JK}{N}\,.
\ee
We remark that the above expressions hold at finite $N$ and finite $q$. However, when scaling the system as $N\rightarrow \infty$, keeping $\J $ fixed rather than $J$ is more natural, since this choice makes the Brownian Hamiltonian have an energy variance extensive in the system size $N$. Doing this, the effective Hamiltonian has an $O(1)$ gap in the large-$N$ limit.
\be\label{eq:linearspecapprox} 
E(s) = s \Egap\,,
\ee
for $s\ll \frac{\J N}{\Egap} = \frac{3N}{4q}$, where the single-particle excitations are the Pauli matrices at each site.

We now move to computing the overlap moments $Z_n(t)$ at finite $N$. Given that we have diagonalized the two-replica effective Hamiltonian, we can use that
\be
    e^{-tH_{\text{eff}}^{a\bar{b}}} = \sum_{\alpha} e^{-tE_\alpha} \ket{ P_{\alpha} } \bra{P_{\alpha}} \,,
\ee
to get a formula for the overlap moments using \eqref{eq:Znetaoverlapinftemp},
\be\label{eq:ZntBrownianspin}
 Z_n(t) = \frac{1}{4^{nN}} \sum_{\alpha_1, \ldots \alpha_n} e^{-t(E_{\alpha_1} + \ldots + E_{\alpha_n})} \left |\text{Tr}(P_{\alpha_1} \ldots P_{\alpha_n}) \right |^2\,.
\ee
What remains is to evaluate this expression. Doing this we find that:
\begin{lemma}\label{lemma:bspin2}
The overlap moment \eqref{eq:ZntBrownianspin} is given by the $n$-th spectral moment 
\be\label{eq:lemmabrownianZn(t)}
Z_n(t) = \Tr(\rho_t^n)\,,
\ee
of the operator $\rho_t = \sum_{\alpha} \lambda_t(\alpha)\, \ket{P_\alpha}\bra{P_\alpha}$. Its eigenvalues depend only on the length $s = |\alpha|$ of the Pauli string,
\be\label{eq:ptsspin}
\lambda_t(s) = \frac{1}{4^N} \sum_{l=0}^N 3^l e^{-tE(l)} \sum_{j=0}^s \left (-3 \right )^{-j} \binom{s}{j} \binom{N-s}{l-j} \,.
\ee
\end{lemma}
For the proof, see Appendix \ref{app:exactBrownian}. As we will see later, $\rho_t$ is in fact a density matrix, although this is not immediately evident from the form of its eigenvalues \eqref{eq:ptsspin}. It represents the average state of the ensemble $\rho_t=\overline{\vert \psi_t\rangle\langle\psi_t\vert}$. We note that the proof of this relation does not depend on the particular functional form of $E(s)$. Equivalently, as long as the effective Hamiltonian is diagonalized in the basis of Pauli strings and the energies depend only on the size of the strings, then the previous formula holds.

We now use Lemma \ref{lemma:bspin2} and take the limit of $n\rightarrow 0$ of \eqref{eq:lemmabrownianZn(t)} to obtain 
\be 
\lim_{n\rightarrow 0} Z_n(t) = {\rm rank}(\rho_t)\,.
\ee
At $t=0$, we have $\lambda_0(s)=\delta_{s,0}$, so $\rank(\rho_0)=1$. 
For $t>0$, none of the $\lambda_t(s)$ vanish, implying $\rank(\rho_t)=4^N$ is maximal. The reason is that we can write $\lambda_t(s)=\sum_l c_l(s)\,e^{-tE(l)}$, with time-independent coefficients $c_l(s)$. 
As a function of the variables $\vec\omega=(e^{-tE(1)},\ldots,e^{-tE(N)})$, 
$\lambda_t(s)$ is real-analytic, and therefore it is either identically zero or its zeros form a discrete (countable) set of measure zero. 
At large times, $\lambda_{\infty}(s)=4^{-N}$ for all $s$, so $\lambda_t(s)$ is not identically zero, 
and any zeros in $\vec\omega$, and hence in $t$, can only occur on a countable, measure-zero subset of times. Using \eqref{eq:transitiondimension} this implies
\be\label{eq:resultbspin}
d_\Omega(t)=\min\!\left\{\Omega,\,\mathrm{rank}(\rho_t)\right\}
=\begin{cases}
1\, & t=0\,,\\[4pt]
\min\{\Omega,\,4^N\}\, & t>0^{(*)}\,,
\end{cases}
\ee
in agreement with the expectation discussed in Section~\ref{Sec:countingframeworkexpectation}. 
The asterisk indicates that the statement holds up to a countable, measure-zero subset of times. This subset might be empty, and in fact we expect it to be empty following the lemma described in the previous section concerning absolutely continuous probability distributions. The fact that the subset is empty and that all eigenvalues are positive for all times can be verified numerically for particular values of $N$ and $q$, but we were not able to give a direct proof.

We can also be more explicit and compute $Z_n(t)$ in the quasi-particle approximation \eqref{eq:linearspecapprox} for the spectrum of the effective Hamiltonian. Since the excitations are dilute in this part of the spectrum, we can, to a good approximation, neglect any possible overlap between different Pauli strings. Therefore, for the trace in \eqref{eq:ZntBrownianspin} not to vanish, the strings must come in pairs. We then have $\tbinom{n}{2} 3N $ fermionic excitations, each with energy $2\Egap$. The overlap moment is the partition function of such fermions
\be\label{eq:Bspinoverlapmom} 
Z_n(t) \approx \dfrac{1}{4^{N(n-1)}}  (1+e^{-2t\Egap})^{\frac{3Nn(n-1)}{2}} \,.
\ee 
The analytic continuation of \eqref{eq:Bspinoverlapmom} likewise leads to~\eqref{eq:resultbspin}. 
Regarding the regime of validity of this approximation, note that it neglects Pauli strings of length at least $O(\sqrt{N})$ which in fact begin to overlap with $O(1)$ probability \cite{erdHos2014phase}. 
Such configurations correspond to states with energies of order $O(N)$ in units of $\J$ and therefore decay rapidly. 
Since the number of these high-energy states is $O(4^N)$, 
their contribution becomes negligible for $\J t \gtrsim O(1)$ when compared to the low-lying modes. 
Consequently, expression~\eqref{eq:Bspinoverlapmom} is valid within this time regime.

\subsubsection{Brownian SYK}
\label{sec:BrownianSYKexact}

The case of Brownian SYK can be treated in close analogy to the previous example, so we will be brief here (see \cite{Saad:2018bqo,Sunderhauf:2019djv,Jian:2022pvj,Milekhin:2023bjv,Guo:2024zmr} for details on the definition of the model). We include it, however, because it will be important later on.  We consider $N$ Majorana fermions $\psi_{i}$ with $i=1,\dots,N$, with $\lbrace \psi_i,\psi_j\rbrace = 2\delta_{ij}$. We restrict the drive operators of the Hamiltonian to be $q$-body Majorana strings $\psi_{\alpha} = \psi_{i_1}\cdots\psi_{i_q}$, where $\alpha = (i_1,\dots,i_q)$ with $i_1 < \cdots < i_q$. We further take $q$ to be even, so that the operators are bosonic. The time-dependent Hamiltonian then takes the form
\be\label{eq:timedepHSYK}
    H(t) = (-\iw)^{\frac{q}{2}}\sum_{i_1 < ... < i_q} g_\alpha (t) \psi_\alpha\,.
\ee
The independent random couplings drawn from \eqref{eq:statisticscouplings} lead to the two-replica effective Hamiltonian \cite{Jian:2022pvj,Guo:2024zmr}
\be\label{eq:HeffSYK}
    H_{\text{eff}}^{a\bar{b}} =  \frac{J }{2}\,\iw^q\sum_{|\alpha|=q} \left (\psi_\alpha^a - (-1)^{\frac{q}{2}}\psi_\alpha^{\bar{b}} \right )^2\,.
\ee
This system has additional subtleties due to global symmetries. First, fermion parity $(-1)^F$ commutes with the Brownian Hamiltonians \eqref{eq:timedepHSYK} and hence with the effective Hamiltonian \eqref{eq:HeffSYK}. Second, for $N \equiv 4 \pmod{8}$ there is an extra discrete ``particle–hole'' symmetry that induces degeneracies within each fermion parity sector. For simplicity we here restrict to $N \not\equiv 4 \pmod{8}$. But we remark that, at the end of the day, the formula we derive below is valid for all $N$, including odd $N$, and we comment about this in Appendix \ref{app:exactBrownian}.

As a result, for $N$ even and $N \not\equiv 4 \pmod{8}$, there are two ground states, given by the infinite-temperature TFD states in each fermion parity sector, $\ket{\mathbf{1},\pm}$. This simply indicates that the time-evolution unitaries $U(t)$ are block-diagonal, and that they only become random within each fermion parity sector, rather than on the full Hilbert space. The overlap moments at infinite time become
\be 
Z_n(\infty) = |\bra{\eta}\ket{e,+}|^2 + |\bra{\eta}\ket{e,-}|^2 = \dfrac{1}{2^{(n-1)(N-1)}}\,,
\ee 
where $\ket{e,\pm}$ is defined as in \eqref{eq:defpermstate} with respect to a fixed fermion parity TFD state $\ket{\mathbf{1},\pm}$, whereas $\ket{\eta}$ is defined with respect to the total TFD $\frac{1}{\sqrt{2}}(\ket{\mathbf{1},+} + \ket{\mathbf{1},-})$.

Building on the intuition from the previous examples, consider now the Majorana string states $\ket{\psi_\alpha,\pm}= \psi^{a}_\alpha \ket{\mathbf{1},\pm}$. We find:

\begin{lemma}\label{lemma:bSYK1}
The Hamiltonian \eqref{eq:HeffSYK} is diagonal in the Majorana basis $\ket{\psi_\alpha,\pm}$, with eigenvalues
\be\label{eq:specSYKHefffinal}
E(s)=JK\Bigl[1-{}_2F_1\!\left(\!-q,\,-s;\,-N;\,2\right)\Bigr]\,,
\ee
where $s = |\alpha|$ and $K = {N \choose q}$. Each level is degenerate, with multiplicity $2{N \choose s}$.
\end{lemma}
We leave the proof to Appendix \ref{app:exactBrownian}. Note that, in order to consider a basis of eigenstates, we need to restrict $|\alpha|\leq N/2$. This is consistent with $E(N-s) = E(s)$ for \eqref{eq:specSYKHefffinal}.

Defining the normalized coupling $\J = JK/N$, we have that the gap in this case reads
\be 
\Egap  = 2q \,\J \,.
\ee 
In the large-$N$ limit the excitations are additive $E(s) = s\Egap$, as for the Brownian spin cluster. 

We now want to compute the overlap moments $Z_n(t)$. We replace the form of the operator
\be
    e^{-tH_{\text{eff}}^{a\bar{b}}} = \sum_{\alpha} e^{-tE_\alpha} \left(\ket{ \psi_{\alpha},+ } \bra{\psi_{\alpha},+} + \ket{ \psi_{\alpha},-} \bra{\psi_{\alpha},-}  \right)\,,
\ee
into the definition of the moments \eqref{eq:Znetaoverlapinftemp} to get
\be\label{eq:SYKmoment} 
Z_n(t)= \frac{1}{2^{n(N-1)}} \sum_{\alpha_1, \ldots \alpha_n} e^{-t(E_{\alpha_1} + \ldots + E_{\alpha_n})}\sum_{f_1,...,f_n = \pm 1}\left |\text{Tr}(\Pi_{f_1}\psi_{\alpha_1}\Pi_{f_2}\psi_{\alpha_2} \ldots \Pi_{f_n}\psi_{\alpha_n}) \right |^2 \,.
\ee
where $\Pi_f = \tfrac{1}{2}(1 + f\,(-1)^F)$ is the projector onto the $f$ fermion parity sector, and where we recall that the sum over $\alpha$ is restricted to $|\alpha|\leq N/2$. We can simplify this expression a bit. First we notice we can extend the sums over all chains of size $s = 0, \ldots,N$, adding a factor $1/2^n$. Second, one notes that each term has $n$ factors $\Pi_f$ and $n$ chains of fermions between them. If we commute them in order to have all the $n$ factors $\Pi_f$ together (taking into account that $\Pi_f \psi_\alpha = \psi_\alpha \frac{1}{2} (1 +f (-1)^\abs{\alpha} (-1)^F)$), we see that given $\psi_{\alpha_1}, \ldots,\psi_{\alpha_n}$ chains, only $2$ changes of signs will endure since $\Pi_+ \Pi_- = 0$. In those cases, noting $\Pi_f^n = \Pi_f$ we arrive at
\be
Z_n(t) = \frac{1}{d^{2n}} \frac{1}{2} \sum_{\alpha_1, \ldots,\alpha_n}^{\abs{\alpha} \leq N} e^{-t(E_{\alpha_1} + \ldots +  E_{\alpha_n})} \left [\abs{\text{Tr}(\psi_{\alpha_1} \ldots \psi_{\alpha_n})}^2 + \abs{\text{Tr}((-1)^F \psi_{\alpha_1} \ldots \psi_{\alpha_n})}^2 \right ] \, ,
\label{eq:ZnSYKNeven}
\ee
where we can separate the module square of the trace because they cannot be non-zero at the same time. We find the replicated partition function to be an average of two partition functions: $Z_n^+(t)$ (first term) built over the ground state with no fermions, and $Z_n^-(t)$ (second term) built over the ground state with all fermions.\footnote{The label ``$\pm$'' is not related to the parity sector, it will be relevant in the proof.}

What remains is to evaluate this expression. Doing this we find that:
\begin{lemma}\label{lemma:syk2}
The overlap moment \eqref{eq:SYKmoment} is given by the $n$-th spectral moment 
\be\label{eq:lemmasykZn(t)}
Z_n(t) = \Tr(\rho_t^n)\,,
\ee
of the operator $\rho_t = \sum_{\alpha} \lambda_t(\alpha)\, \ket{P_\alpha}\bra{P_\alpha}$, with $\alpha$ even. Its eigenvalues depend only on the length $s = |\alpha|$ of the Majorana string,
\be\label{eq:ptssyk}
\lambda_t(s) = \frac{1}{2^{N}} \sum_{l=0}^{N} e^{-tE(l)} \sum_{j=0}^{s} (-1)^j \binom{s}{j} \binom{N-s}{l-j} \,.
\ee
\end{lemma}
For the proof, see Appendix \ref{app:exactBrownian}. As before, $\rho_t$ is in fact a density matrix, although this is not immediately evident. It represents the average state of the ensemble. Also, again the proof of this relation does not depend on the particular functional form of $E(s)$. Equivalently, as long as the effective Hamiltonian is diagonalized in the basis of Majorana strings and the energies depend only on the size of the strings, then the previous formula holds. We remark that although we started assuming $N$ even and $N \not\equiv 4 \pmod{8}$, the formula is valid for all $N$. We explain this in Appendix \ref{app:exactBrownian}.

We now use Lemma \ref{lemma:syk2} and take the limit of $n\rightarrow 0$ of \eqref{eq:lemmasykZn(t)} to obtain 
\be 
\lim_{n\rightarrow 0} Z_n(t) = {\rm rank}(\rho_t)\,.
\ee
At $t=0$, we have $\lambda_0(s)=\delta_{s,0}$, so $\rank(\rho_0)=1$. 
For $t>0$, none of the $\lambda_t(s)$ vanish, implying $\rank(\rho_t)=2^{N-1}$ is maximal. The reasons are the same as for the Brownian spin cluster. Here, the $-1$ arises since we are exploring one single fermion parity sector of the Hilbert space. This is half of the Hilbert space of $2N$ Majoranas. Then, using \eqref{eq:transitiondimension} this implies
\be\label{eq:resultbspin2}
d_\Omega(t)=\min\!\left\{\Omega,\,\mathrm{rank}(\rho_t)\right\}
=\begin{cases}
1\, & t=0\,,\\[4pt]
\min\{\Omega,\,2^{N-1}\}\, & t>0^{(*)}\,,
\end{cases}
\ee
in agreement with the expectation discussed in Section~\ref{Sec:countingframeworkexpectation}, and where the asterisk means the same as for the Brownian spin cluster, see above.

In the large-$N$ limit, the excitations are dilute and their energies become additive. Consequently, they must appear in pairs in order to contribute to the trace (the trace of $\psi_i$ also vanishes within each parity sector). Furthermore, for any given configuration, simultaneously changing all $f_i \to -f_i$ leaves the contribution invariant, so there is a two-fold degeneracy associated to the degeneracy of ground states. The overlap moment therefore reduces to the partition function of ${N \choose 2}$ fermionic excitations,
\be\label{eq:SYKmomentapprox} 
Z_n(t) \approx \dfrac{1}{2^{(n-1)(N-1)}}  (1+e^{-2t\Egap})^{\frac{Nn(n-1)}{2}} \,,
\ee
a formula valid for $Jt\gtrsim \mathcal O(1)$.

\subsection{Large-$N$ path integral}\label{sec:largeNsyk}

We now turn to the regime most relevant for gravity, where the overlap moments are accessible only through collective field path integral methods. For concreteness, we focus on the Brownian $q$-SYK model, building on previous work \cite{Jian:2022pvj,Saad:2018bqo,Maldacena:2018lmt,Maldacena:2016hyu}, though similar methods also apply to the large-$N$ Brownian spin cluster \cite{Jian:2022pvj}.

\subsubsection{Method 1: Full path integral}\label{subsec:largeNsyk}

We begin by noting that each replica of the overlap \eqref{eq:overlapdef} can be expressed as a closed path integral of the Majorana fermions, governed by the Hamiltonian \eqref{eq:timedepHSYK}. We seek to compute the replica partition function $Z_n$. For $n$ even, we can transpose half of the traces (for example, the even factors) and use the cyclic property to write
\be
Z_n(t) = \frac{1}{2^{nN/2}} \overline{\text{Tr}(U^\dagger_1(t) U_2(t))\text{Tr}(U^*_2(t) U^t_3(t))\text{Tr}(U^\dagger_3(t) U_4(t)) \ldots,\text{Tr}(U^*_n(t)U^t_1(t))} \, ,
\ee
so that we have $n$ unitaries paired with their complex conjugate. 

The idea is to use the path integral to compute these traces over Majorana fermions, and approximate the path integral in the large-$N$ limit. In each (odd) trace $\text{Tr}(U^\dagger_j(t) U_{j+1}(t))$, we will have two sets of fermions $\{\psi_a^{j,1}(\tau)\}$,$\{\psi_a^{j,2}(\tau)\}$, $a = 1 \ldots,N$, the first one evolving with $-H_j(\tau)$ and the second with $H_{j+1}(\tau)$, where the subindex corresponds to a particular set of couplings $g_\alpha(\tau)$. In each (even) trace $\text{Tr}(U^*_j(t) U^t_{j+1}(t))$, the first set of fermions will evolve with $-H^*_j(\tau)$ and the second one with $H^*_{j+1}(\tau)$. The boundary conditions for both cases are
\be\label{boundc}
\psi_a^{j,1}(t) = \psi_a^{j,2}(0) \, , \;\;\;\; \psi_a^{j,2}(t) = -\psi_a^{j,1}(0) \, .
\ee
So, in general we have $2n$ sets of fermions $\{\psi_a^{j,1}\}, \{\psi_a^{j,2}\}$ ($j = 1, \ldots,n$), with boundary conditions given in pairs. The full path integral then reads 
\be\label{pathav}
\int \mathcal{D}\psi^j_a \exp{\iw \int_0^t d\tau \frac{i}{4} \psi^j_a(\tau) \partial_\tau \psi^j_a(\tau) - \sum_{\abs{\alpha} = q} (-1)^j g_\alpha^j(\tau) \left ((-\iw)^{q/2}\psi_\alpha^{j,1}(\tau) - \iw^{q/2} \psi_\alpha^{j-1,2}(\tau) \right )} \, ,
\ee
where we are summing over $a = 1, \ldots,N$, over $j = 1, \ldots,n$ ($j=0$ is equivalent to $j=n$), and, if we do not indicate, it is implicit the sum over $1,2$. Performing the average over the Brownian couplings, we obtain
\be
Z_n(t) = \frac{1}{2^{nN/2}} \int \mathcal{D}\psi^j_a \exp{-\int_0^t d\tau \frac{1}{4} \psi^j_a(\tau) \partial_\tau \psi^j_a(\tau) + J \sum_{\abs{\alpha} = q} \left (n-\psi_\alpha^{j,1}(\tau) \psi_\alpha^{j-1,2}(\tau) \right )} \, .
\ee
To compute this path integral we follow standard large-$N$ techniques. We first introduce new variables $G_j(\tau), \Sigma_j(\tau)$, $j = 1, \ldots,n$, so that the expression reads
\begin{align}
\notag Z_n(t) \approx \frac{1}{2^{nN/2}} &\int \mathcal{D}G_j \mathcal{D} \Sigma_j \exp{-\frac{1}{4}\int_0^t d\tau ~ 4J K \left (n-\iw^q G_j^q(\tau) \right ) + N G_j(\tau) \Sigma_j(\tau)} \times \\
\times & \int \mathcal{D}\psi^j_a \exp{-\frac{1}{4} \int_0^t d\tau ~ \psi^j_a(\tau) \partial_\tau \psi^j_a(\tau) -\Sigma_j(\tau) \psi_a^{j,1}(\tau)\psi_a^{j-1,2}(\tau)} \, .
\end{align}
The $\approx$ symbol refers to the large-$N$ approximation. In deriving this expression we used that for large-$N$, $K = \binom{N}{q} \approx N^q / q!$, and we just kept the highest order in $N$. Although it is transparent from the previous expression, we notice that if we integrate over the Lagrange multipliers $\Sigma_j(\tau)$, their equations of motion fix $G_j(\tau) = \frac{1}{N} \psi_a^{j,1}(\tau) \psi_a^{j-1,2}(\tau)$. If we introduce this constraint back in the action, we recover the original fermion path integral \ref{pathav}.

The next step is to integrate over the fermions to obtain an effective action for $\Sigma_j$ and $G_j$. This is possible because the expression is a quadratic action over the $2n-$vector $\Psi_a(\tau)$ of Majorana fermions with certain mass matrix $\mathbb{\Sigma}$. This part of the path integral then reads
\be
\label{Majoranadet}
\int \mathcal{D}\Psi_a \exp{-\frac{1}{4} \int_0^t d\tau ~ \Psi_a^\dagger \left ( \partial_\tau - \mathbb{\Sigma}(\tau) \right ) \Psi_a} = \det (\partial_\tau - \mathbb{\Sigma}(\tau))^{N/2} \, .
\ee
The particular form of the mass matrix $\mathbb{\Sigma}$, as well as the details of the computation of the determinant, taking into account the boundary conditions \ref{boundc}, can be found in Appendix \ref{app:ferdet}. The final result is 
\be
\det (\partial_\tau - \mathbb{\Sigma}(\tau))^{N/2} = e^{-I^{\textrm{eff}}_n(t)} \, , \;\;\; I^{\textrm{eff}}_n(t) = -\frac{N}{2} \ln{(2)} - \frac{N}{2} \ln\, (f_n(t))\, ,
\ee
where we have defined 
\be
f_n(t) \equiv \sum_{i=0}^{n/2} \sum_{\abs{\beta} = 2i}\prod_{j \in \beta} \cos{\left ( \int_0^t d\tau \frac{\Sigma_j(\tau)}{2} \right )} - \prod_{j=1}^n \sin{\left ( \int_0^t d\tau \frac{\Sigma_j(\tau)}{2} \right )} \, ,
\label{poldetMajorana}
\ee
where $\abs{\beta} = 2i$ means all the $\binom{n}{2i}$ subsets of $\lbrace \Sigma_1 , \ldots,\Sigma_n \rbrace$ with $2i$ elements. Then, the replica partition function can be written as $Z_n(t) = e^{-I_n(t)}$, where the total effective action is
\be
I_n(t) = I^{\textrm{eff}}_n(t) + n\frac{N}{2} \ln{(2)} + nJKt - \sum_{j=1}^n JK \iw^q \int_0^t d\tau ~ G^q_j(\tau) + \frac{N}{4} \int_0^t d\tau ~ G_j(\tau) \Sigma_j(\tau) \, .
\ee
To find the contribution for large $t$, we remember that for sufficiently long times, on average the effective Hamiltonian drives the circuit to its vacuum, which is the infinite temperature thermofield double between replicas. In such a state the correlations between fermions in different replicas saturate to a constant value. This suggests that to find the solution at large $t$, it is natural to assume a scenario where $\Sigma_i, G_i$ are constant. The same assumption was taken in Ref. \cite{Saad:2018bqo} to understand the spectral form factor in Brownian SYK. As explained there, such solutions play the role of ``wormholes'' in this system.

To understand the structure of the solutions, we start by analyzing the simplest case, namely $n=2$. The full effective action simplifies to
\be
I_2(t) = N \ln{(2)} - N \ln{\left (2 \cos{ \left (\frac{(\Sigma_1 + \Sigma_2) t}{4} \right )} \right )} + 2JKt - \iw^q JKt \left (G_1^q + G_2^q \right ) + \frac{N}{4} t \left ( G_1 \Sigma_1 + G_2\Sigma_2 \right ) \, .
\ee
Given that $\Sigma_1, \Sigma_2$ are (constant) complex numbers, for large $t$ we have 
\be
\ln{\left (2 \cos{ \left (\frac{(\Sigma_1 + \Sigma_2) t}{4} \right )} \right )} \approx \pm \iw \frac{(\Sigma_1 + \Sigma_2) t}{4} \, ,
\ee
where the $\pm$ depends on the sign of the imaginary part of $\Sigma_1 + \Sigma_2$ (opposite to it). By solving the equations of motion, we find the two saddle points for $\Sigma_j$ and $G_j$,
\be
G_j = \pm \iw \, , \;\;\; \Sigma_j = \mp \iw \frac{4qJK}{N} .
\ee
As in the case of the double-cone in the computation of the spectral form factor \cite{Saad:2018bqo}, the part of the action that depends on $t$ cancels out in these saddle points. Note also we have $2$ saddle points instead of $4$ because it must be $\Sigma_1 = \Sigma_2$ in order to obtain a non-decaying contribution for large $t$. Then, in the strict limit of large times we obtain
\be
Z_2(\infty) = \frac{1}{2^{N-1}} \, .
\ee
By taking the next to leading term in the expansion of the logarithm
\be
\ln{\left (2 \cos{ \left (\frac{(\Sigma_1 + \Sigma_2) t}{4} \right )} \right )} \approx \pm \iw \frac{(\Sigma_1 + \Sigma_2) t}{4} + \ln{\left (1 + e^{\mp \iw \frac{(\Sigma_1 + \Sigma_2) t}{2}} \right ) } \, ,
\ee
we can find the first correction in $t$ to the two replica partition function
\be
Z_2(t) \approx \frac{1}{2^{N-1}} \left (1 + N e^{- \frac{4qJKt}{N}} \right ) \, .
\ee
This is the same expression as the one obtained by exact methods in the appropriate time regime, see \eqref{eq:SYKmomentapprox} and take the large-$N$ limit. In Fig. \ref{fig:Z2BSYK} we give a graphical representation of these non-trivial saddles, as well as the disconnected saddles corresponding to $G_i = \Sigma_i = 0$, whose contribution decays like $e^{-JKt}$.

\begin{figure}[h]
\centering
\includegraphics[width=.9\linewidth]{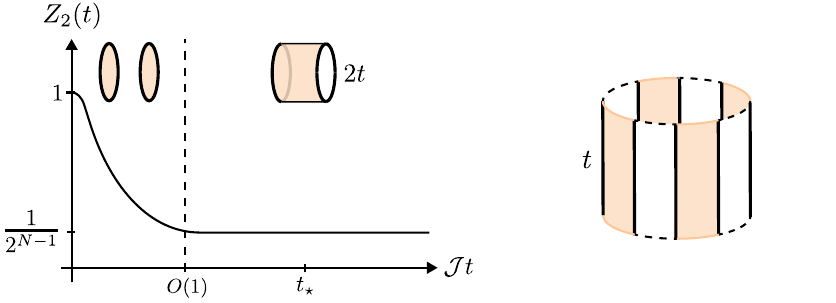}
\caption{On the left, the ``wormhole'' saddle point dominates $Z_2(t)$ at $O(1)$ times and reproduces the contribution from the tower of quasi-particle excitation with mass $E_{\text{gap}} = 2 q \J$. On the right, the corresponding replica wormhole providing the analogous contribution to $Z_4(t)$.}
\label{fig:Z2BSYK}
\end{figure}

We can follow a similar procedure for general even $n$. At large times we can simplify (\ref{poldetMajorana}) to
\be
f_n(t) \approx \prod_{j=1}^n \cos{\left (\frac{\Sigma_j t}{2} \right )} - \prod_{j=1}^n \sin{\left (\frac{\Sigma_j t}{2} \right )}  \, ,
\ee
so the total effective action reads
\be
I_n(t) \approx (n-1) \frac{N}{2} \ln{(2)} - \frac{N}{2} \ln{\left ( \prod_{j=1}^n \cos{\left (\frac{\Sigma_j t}{2} \right )} - \prod_{j=1}^n \sin{\left (\frac{\Sigma_j t}{2} \right )} \right )} + nJKt - t \sum_{j=1}^n JK \iw^q G_j^q + \frac{N}{4}  G_j \Sigma_j \, .
\ee
In this general case, the saddles are still the same $G_j = \pm \iw$, $\Sigma_j = \mp \iw \frac{4qJK}{N}$ (see Fig. \ref{fig:Z2BSYK} for a graphical representation of $Z_4$), but only some particular combinations will have a time-independent contribution for large $t$. Defining $x \equiv \frac{4qJK}{N}$, we have 
\be
\left(\prod_{j=1}^n \cos{\left (\frac{\Sigma_j t}{2} \right )} - \prod_{j=1}^n \sin{\left (\frac{\Sigma_j t}{2} \right )}\right)_{\textrm{Saddles}} = \cosh^n{ \left ( \frac{xt}{2} \right )} \pm \sinh^n{\left (\frac{xt}{2} \right )} \, ,
\ee
where the $\pm$ depends on how many $\Sigma_j$ are $\iw x$ and how many are $-\iw x$, being the $+$ cases the ones non-decaying. It is easy to see that half of the $2^n$ possibilities will contribute, obtaining
\be
Z_n(\infty) = \frac{1}{2^{(n-1)(N-1)}} \, .
\ee
The first correction in $t$ comes again from considering the next to leading term in the expansion of the logarithm
\be
\frac{N}{2}\ln{ \left (f_n(t) \right )} \approx  -(n-1)\frac{N}{2} \ln{(2)} + \frac{N}{4}xt + \ln{ \left (1 + N\binom{n}{2} e^{-xt} \right )} \, ,
\ee
leading to our main result in this section
\be
Z_n(t) \approx \frac{1}{2^{(n-1)(N-1)}} \left (1 + N \binom{n}{2}e^{-\frac{4qJKt}{N}} \right )\, ,
\ee
which coincides with the exact expression in the appropriate regime of large-$N$ and sufficiently large times. This shows we can recover the dimension of the Hilbert space of the theory in the semiclassical limit, and this precisely parallels the exact microscopic analysis. One interesting avenue for future work is to see how far one can go with the semiclassical analysis in relation to the exact formulas derived above.

\subsubsection{Method 2: Gluing energy propagators}

In the previous section we obtained the replica partition functions and associated Hilbert space dimension in the large-$N$ limit for times greater than the gap time and for $n$ even. In this section we take another more general approach to the computation.

To explain the idea, let us first start by writing the replicated partition functions as 
\be
\notag Z_n(t) = \langle \eta \vert \exp{-tH_{\text{eff}}}^{\otimes n} \vert \eta \rangle\;. 
\ee
Now we introduce partitions of the identity using the basis of strings of Majorana fermions (of length $\leq N/2$) acting over the two vacuums $\vert \pm \rangle$. Given that these are eigenstates of $H_\text{eff}$, we obtain the expression
\be
Z_n(t) = \sum_\alpha \abs{\langle \eta \vert \Pi_{f_1}\psi_{\alpha_1}, \ldots,\Pi_{f_n}\psi_{\alpha_n} \rangle}^2 e^{-t(E_{\alpha_1} + \ldots +  E_{\alpha_n})} \, , 
\ee
where we used that
\be
\langle \Pi_f\psi_\alpha \vert e^{-tH_\text{eff}} \vert \Pi_{f'}\psi_{\alpha'} \rangle  = \delta_{\alpha, \alpha'} \delta_{f,f'} e^{-tE_\alpha} \, .
\ee 
This is of course the path we took to arrive at the exact microscopic formula. The idea now is to re-obtain this last expression from the large-$N$ limit using the macroscopic variables $G$, $\Sigma$. Thus, we want to compute $\langle \psi_\alpha \vert \overline{U(t) \otimes U^*(t)} \vert \psi_{\alpha'} \rangle$, where $\psi_\alpha$, $\psi_{\alpha'}$ are strings of Majorana fermions acting on $\mathcal{H}^1$, via de the semiclassical path integral. To this end we consider again the path integral expression
\begin{align}
\notag &\int \mathcal{D}G \mathcal{D} \Sigma \exp{-\frac{1}{4}\int_0^t d\tau ~ 4JK \left (1-\iw^qG^q(\tau) \right ) + N G(\tau) \Sigma(\tau)} \times \\
\times & \int \mathcal{D}\psi^1_a \mathcal{D}\psi^{\bar{1}}_a \exp{-\frac{1}{4} \int_0^t d\tau ~ \psi^1_a(\tau) \partial_\tau \psi^1_a(\tau) + \psi^{\bar{1}}_a(\tau) \partial_\tau \psi^{\bar{1}}_a(\tau) -\Sigma(\tau) \psi_a^1(\tau)\psi_a^{\bar{1}}(\tau)} \, .
\end{align}
The difference lies in the boundary conditions. While previously we considered the trace with antisymmetric boundary conditions, in this case we need to enforce particular initial and final boundary conditions. With this in mind, to integrate the Majorana fermions, we can form $N$ Dirac fermions $\Psi_a(\tau) = \frac{1}{2} \left (\psi_a^1(\tau) + \iw \psi_a^{\bar{1}}(\tau)\right )$ with masses $\Sigma(\tau)/2$. Given that the path integral factorizes in the $N$ Dirac fermions, we can consider each one separately. Formally, we have
\be
\exp{\iw \int_0^t d\tau \frac{\Sigma(\tau)}{4} } \int_{\Psi(0) = \Psi_0}^{\Psi(t) = \Psi_t} \mathcal{D} \bar{\Psi} \mathcal{D} \Psi \exp{\iw \int_0^t d\tau ~ \iw \bar{\Psi}(\tau)\partial_\tau \Psi(\tau) - \frac{\Sigma(\tau)}{2}\bar{\Psi} (\tau)\Psi(\tau)} \, .
\ee
This is an evolution of a free fermion with mass $\Sigma(\tau)/2$. Defining $\vert 0 \rangle$ from $ \Psi \vert 0 \rangle = 0$ and $\Psi^\dagger \vert 0 \rangle \equiv \vert 1 \rangle$, we obtain
\begin{align}
&\exp{\iw \int_0^t d\tau \frac{\Sigma(\tau)}{4} }  \langle \Psi_t \vert \exp{-\iw \Psi^\dagger \Psi \int_0^t d\tau \frac{\Sigma(\tau)}{2}} \vert \Psi_0 \rangle \\
&= \delta_{\Psi_t, \Psi_0} \left ( \delta_{\Psi_0,0} \exp{\iw \int_0^t d\tau \frac{\Sigma(\tau)}{4}} + \delta_{\Psi_0,1} \exp{-\iw \int_0^t d\tau \frac{\Sigma(\tau)}{4}} \right ) \, .
\end{align}
This is an exact result. It tells us the path integral is diagonal in the basis of strings of Majorana fermions, as expected from the microscopic description. Also, if the string $\psi_\alpha$ contains the Majorana fermion $\psi_a$, the corresponding Dirac state $\Psi_a$ is $\vert 1 \rangle$; if not, the corresponding state is $\vert 0 \rangle$. Thus, if $\abs{\alpha} = s$, the effective action for $G$ and $\Sigma$ is
\be
I_{\textrm{eff}}(t) = -\iw \frac{N-2s}{4}  \int_0^t d\tau ~ \Sigma(\tau) + JKt - \iw^q JK \int_0^t d\tau ~ G^q(\tau) + \frac{N}{4} \int_0^t d\tau ~ G(\tau) \Sigma(\tau) \, .
\ee
This action is linear in $\Sigma(\tau)$, so we can compute the path integral exactly, obtaining
\be
I_{\textrm{eff}}(t) = tJK \left [1-\left (1 - \frac{2s}{N} \right )^q \right ] \equiv t\tilde{E}_s \, ,
\ee
where $\tilde{E}_s$ is the large-$N$ approximation of (\ref{eq:specSYKHefffinal}). Therefore, we find
\be
\langle \psi_\alpha \vert \overline{U(t) \otimes U^*(t)} \vert \psi_{\alpha'} \rangle = \delta_{\alpha, \alpha'} e^{-t\tilde{E}_\alpha} \, .
\ee
Then, we can compute the replicated partition function in the exact same way as we did above, in order to obtain (\ref{eq:lemmasykZn(t)}) and (\ref{eq:ptssyk}). The only thing that changes is $E_s \to \tilde{E}_s$, i.e. we need to change the spectrum from the exact to the semiclassical one. We note this a remarkably transparent demonstration of the robustness of semiclassical methods to compute Hilbert dimensions in this framework. The formula for the replicated partition function is structurally the same, and the only thing that gets modified is the energy spectrum, from exact to semiclassical. But this obviously has nothing to do, and therefore cannot modify, the rank of the average state and therefore the dimension of the Hilbert space spanned by these states.

To end this semiclassical discussion, let us note that we can reinterpret what we have done in a manner that will reappear when discussing black hole caterpillars. We can write
\be
Z_n(t) = \sum_\alpha \abs{\langle \eta \vert \Pi_{f_1} \psi_{\alpha_1}, \ldots,\Pi_{f_n}\psi_{\alpha_n} \rangle}^2 e^{-t(E_{\alpha_1} + \ldots +  E_{\alpha_n})} = \sum_{\alpha,\alpha'} I_{\alpha\alpha'} \,G_{\alpha_1\alpha_1'}\cdots G_{\alpha_n\alpha_n'}\, , 
\ee
where we have defined
\be 
G_{\alpha_j\alpha_j'}\equiv\langle \psi_{\alpha_j} \vert e^{-t H_{\textrm{eff}}^{j\bar{j}}} \vert\psi_{\alpha_j'}\rangle\;,
\ee
and
\be 
I_{\alpha\alpha'}\equiv\langle \eta \vert \Pi_{f_1} \psi_{\alpha_1}, \ldots,\Pi_{f_n} \psi_{\alpha_n} \rangle \langle \Pi_{f_1'} \psi_{\alpha_1'}, \ldots,\Pi_{f_n'} \psi_{\alpha_n'} \vert \eta \rangle\;.
\ee
This has a wormhole interpretation, analog to Fig. \ref{eq:figurewhJT} in the caterpillar scenario discussed below (see Fig. \ref{fig:Z2BSYK} too). Such wormhole resembles a ``birdcage'', where the propagators $G_{\alpha\alpha'}$ compute the partition functions (energy propagators) in the strips (the bars of the birdcage), while the kernel $I_{\alpha\alpha'}$ glue those propagators at the top and at the bottom (the caps of the birdcage).

\subsection{Further generalities}
\label{sec:overlapgeneralexp}

We end this section by showing that the conclusions obtained in the examples above actually hold for generic ensembles constructed following the prescription of Section~\ref{Sec:ensemblesinftemp}. We hope to make transparent that with very mild input, the framework achieves almost complete universality.

To begin, we introduce the transverse transpose of the relevant two-replica operator
\be\label{eq:defrhot}
\rho_t \equiv \big(e^{-t H_{\eff}^{a\bar{b}}}\big)^{T_\perp}\,.
\ee
This is just the Choi state version of the two-replica averaged evolution $\overline{U(t)^1 \otimes U(t)^{\bar{1}}\,^*} =  e^{-tH^{1\bar{1}}_{\eff}}$. With this definition in place, we obtain the following lemma:
\begin{lemma}\label{lemma:densitymatrix}
The operator $\rho_t$ defines a normalized density matrix corresponding to the average state of the ensemble.
\end{lemma}
The proof of this lemma is described in Appendix \ref{app:Heff}. With this definition, the overlap moments \eqref{eq:Znetaoverlapinftemp} can then be written as R\'enyi moments
\be
Z_n(t) = \Tr\left(\rho_t^n\right)\,.
\ee
This is precisely the form found in the previous examples, see \ref{lemma:bspin2} and \ref{lemma:syk2}. This expression shows that the saturation value of the dimension is determined by the rank of the average state:
\be
\lim_{n\to 0} Z_n(t) = \rank(\rho_t)\,,
\ee
and thus the dimension of the Hilbert space spanned by the states is
\be\label{eq:d_Omegagen}
    d_\Omega(t) =  \text{min}\lbrace \Omega, \rank(\rho_t)\rbrace\,.
\ee 
Equivalently, the dimension $d_\infty(t)$ is the operator Schmidt rank of $e^{-t H_{\eff}^{a\bar{b}}}$.

At $t=0$ we have $\rho_0 =\mathbf{1}^{T_\perp}\,\propto\, \ket{\mathbf{1}}\bra{\mathbf{1}}$ and thus $\rank(\rho_0) =1$. To understand the rank at $t>0$, we can formally diagonalize the two-replica effective Hamiltonian
\be\label{eq:effHgeneral}
H_{\text{eff}}^{a\bar{b}} = \sum_{i=1}^{d^2} E_i \ket{V_i}\bra{V_i}\,,
\ee
ordering the eigenvectors $\ket{V_i} \in \mathcal{H}_a\otimes \mathcal{H}_{\bar{b}}$ by increasing energy, with $0=E_1  \leq E_2 \leq E_3  \leq \cdots $. Each $\ket{V_i}$ is the Choi state associated with a single-replica Hermitian operator $V_i$, normalized such that $\text{Tr}(V^2_i) = d$. The operator $\rho_t$ can then be written as
\be\label{eq:T_t} 
\rho_t = \sum_{i=1}^{d^2} e^{-t E_i} (V_i \otimes V^*_i)\,.
\ee 
Under very general conditions (like the absence of symmetries) we expect that for $t>0$, the rank of \eqref{eq:T_t} to be constant, and moreover we expect it to be maximal so that
\be\label{eq:d_Omegagen2}
    d_\Omega(t) =  \text{min}\lbrace \Omega, d^2\rbrace\,,\qquad \forall\, t > 0\,.
\ee 
This validates the general expectation of Section \ref{Sec:countingframeworkexpectation}.

A general argument for this works as follows. The determinant $\det(\rho_t)$ is a degree-$d^2$ polynomial in the weight vector
\be\label{eq:omegavect}
\vec\omega=(e^{-tE_1},\ldots,e^{-tE_{d^2}})\,.
\ee
As we show below, for sufficiently late times the determinant is nonzero (and the rank maximal), if we assume that $\ket{V_1}=\ket{\mathbf{1}}$ is the unique ground state of $H_{\text{eff}}^{a\bar{b}}$. This property is a very mild input that is simply attained in all scenarios. In this regime the average state $\rho_t \rightarrow \mathbf{1}_{d^2}/d^2$ tends to the maximally mixed state.\footnote{In the presence of global symmetries, as in the SYK example of Section~\ref{sec:BrownianSYKexact}, the average state becomes a weighted mixture of maximally mixed states within each superselection sector. A similar argument points that the rank is constant $\frac{\mathrm{d}}{\mathrm{d}t}\mathrm{rank}(\rho_t)=0$ for $t>0$, even when the rank is submaximal.} The zeros in weight space are therefore isolated. In general, they are complex, but even if some of them lie on the real $\vec\omega$-plane, a trajectory of the form \eqref{eq:omegavect} would not intersect any of them for any $t>0$, in general. Therefore, on quite solid grounds, the rank remains maximal for $t>0$. This is what we found for the Brownian spin clusters.

A more explicit general analysis is possible for a single drive operator. Here, evolution is effectively time-independent and the states $\ket{U(t)}$ are restricted to a diagonal subspace. We study this scenario explicitly in Appendix~\ref{app:H}. Different realizations of the Brownian coupling at fixed $t$ correspond to unitaries at different Hamiltonian times. Therefore, the unitary never becomes Haar random, as shown in detail in \cite{Guo:2024zmr}. But it does instantaneously generate a spanning set of states within the diagonal subspace, provided the spectrum of the Hamiltonian is non-degenerate.

\subsubsection{Dimension from explicit overlap moments}

For the purposes of what follows, it is useful to make the relation between the dimension and the overlap moments $Z_n(t)$ more explicit. At late enough times, the dimension can be obtained directly from the explicit evaluation of the overlaps. Using the generic diagonal form of the effective Hamiltonian \eqref{eq:effHgeneral}, we can write
\be \label{eq:Zntgeneralinftemp}
Z_n(t) = \dfrac{1}{d^{2n}} \sum_{i_1,...,i_n} e^{-t(E_{i_1} + ... + E_{i_n})} \,|\text{Tr}(V_{i_1}...V_{i_n})|^2\,.
\ee 
At late enough times, the sum is dominated by low-energy eigenstates of the effective Hamiltonian. Under the general condition that there are no symmetries or non-interacting clusters among the degrees of freedom, $\ket{V_1} = \ket{\mathbf{1}}$ is the unique ground state. At sufficiently late times we can then expand the replica partition function by retaining low energy modes
\be\label{eq:expansionZ_ngen}
Z_n(t) = \frac{1}{d^{2(n-1)}}\left(1 + n\,\chi_1(t) + \binom{n}{2}\,\chi_2(t) + \cdots\right)\,,
\ee
where 
\begin{align}\label{subtbig}
\chi_1(t) &= \frac{1}{d^2}\sum_{i=2}^{d^2} e^{-tE_i}\,|\Tr(V_i)|^2 = 0\,, \\[.3cm]
\chi_2(t) &= \frac{1}{d^2}\sum_{i,j=2}^{d^2} e^{-t(E_i+E_j)}\,|\Tr(V_i V_j)|^2
= \sum_{i=2}^{d^2} e^{-2tE_i} \,.
\end{align}
The vanishing of $\chi_1(t)$ follows from the orthogonality condition $\braket{\mathbf{1}}{V_i} = 0$ for $i \ne 1$, which implies $\text{Tr}(V_i) = 0$. For $\chi_2(t)$, we use the Hermiticity of the effective Hamiltonian, which implies the orthonormality relation $\delta_{ij} = \braket{V_i}{V_j} = \tfrac{1}{d}\text{Tr}(V_i V_j)$. The resulting expression can be written in terms of the second overlap moment $\chi_2(t) =Z_2(t)-1$, demonstrating that little information is needed to derive the dimension. The dots in \eqref{eq:expansionZ_ngen} represent suppressed contributions from three or more energy levels.

The truncation where we neglect the higher-order terms in \eqref{eq:expansionZ_ngen} is valid once the contributions from all states except the $N_\ast$ lowest excited levels of the effective Hamiltonian, each with energy $\Egap$, have decayed. The corresponding parametric timescale is the {\it gap time}
\be\label{eq:gaptime}
t_\ast =  \Egap^{-1}\log N_\ast\,.
\ee 
For times much greater than such time, using \ref{subtbig}, we can further approximate \ref{eq:expansionZ_ngen} simply as
\be\label{eq:expansionZ_ngen2}
Z_n(t) \simeq \frac{1}{d^{2(n-1)}}\left(1  + N_\ast\, \binom{n}{2}\,e^{-2t\Egap} + \cdots\right)\,,
\ee
Analytically continuing the truncated expansion \eqref{eq:expansionZ_ngen} to $n\to 0$, we find on general grounds that the dimension spanned by the family of states $\mathsf{F}_\Omega^t$ is
\be
    d_\Omega(t) =  \text{min}\lbrace \Omega, d^2\rbrace \,\qquad t\gg t_\ast\,.
\ee
This provides a general derivation of the expectation in Section \ref{Sec:countingframeworkexpectation} directly from the form of the overlap moments, valid at times larger than the gap time and under the stated assumptions. Of course, we expect that on general grounds only $t>0$ is really needed, but demonstrating this requires more information about the ensembles, as we had when considering the Gaussian unitary ensemble and the spin clusters.

\subsubsection{Gap time $\sim$ scrambling time} 
\label{sec:gaptime}

How does the gap time $t_\ast$ scale with system size? Table \ref{tab:gaptime} summarizes our findings from the previous subsections for a system of $N$ qubits. Measured in units of the normalized coupling $\J ^{-1}$, the gap time scales as $O(\log N)$ for few-body Brownian Hamiltonians, and as $O(1)$ for non-local Brownian Hamiltonians such as the Brownian GUE.

\begin{table}[h]
\centering
\setlength{\tabcolsep}{12pt} 
\renewcommand{\arraystretch}{1.2} 
\begin{tabular}{|c|c|c|c|c|}
\hline
& $q$ & $\Egap/\J$ & $N_\ast$ & $\J t_\ast$ \\ \hline
{$q$-local} & $O(1)$ & 
\makecell{$\tfrac{4q}{3}$ (spin) \\[2pt] $2q$ (SYK)} 
& \makecell{$3N$ (spin) \\[2pt] $N$ (SYK)} &  $O(\log N)$  \\ \hline
{non-local} & $O(N)$ & $O(N)$ & $O(2^N)$ &  $O(1)$  \\
\hline
\end{tabular}
\caption{$N$-scaling of the gap time in spin systems. The natural coupling is $\J  = JK/N$, with $J$ the variance of the Brownian couplings and $K$ the number of drive operators in each case.}
\label{tab:gaptime}
\end{table}

We see that the gap time $t_\ast$ scales in the same way as the “scrambling time” \cite{Shenker:2013pqa} for these systems. The connection to scrambling is natural: the gap time governs the approach of the ensemble of time-evolution operators $U(t)$ to a unitary $k$-design \cite{Jian:2022pvj,Guo:2024zmr}, and unitary $2$-designs already appear scrambled when probed by average notions of the four-point OTOC \cite{Roberts:2016hpo,Sunderhauf:2019djv}. Nevertheless, it is unclear whether this relation holds for higher dimension gravity scenarios. More research is needed to understand the similarities and differences.

Furthermore, considering states after the gap time is particularly interesting, as this is precisely when the Choi states of different scrambling unitaries, $\ket{U_1(t)}$ and $\ket{U_2(t)}$, become distinguishable with high probability. The small value of $Z_2(t)$ quantifies this loss of fidelity. In this regime, we are therefore counting families of {\it distinguishable} quantum states.

\section{Counting caterpillar states of black holes}
\label{Sec:BHs}

We now turn to black holes and apply the construction of Section~\ref{sec:finitetempstates}, which is the appropriate one for scenarios at finite temperature. Our goal is to compute the overlap moments $Z_n(t)$ of the corresponding states using the gravitational path integral and to extract from them the dimension of the Hilbert space they span. We begin with near-extremal black holes, where the analysis can be carried out explicitly. Then, we comment on the expected features of more general cases in higher dimensions.

\subsection{Semiclassical geometries}\label{sec:semiclgeom}

Near-extremal black holes develop a nearly AdS$_2$ near-horizon throat (times a compact manifold), whose metric in global coordinates takes the form $\text{d}s^2 =  \tfrac{1}{\sin^2\sigma}(-\text{d}t^2+\text{d}\sigma^2 )$ for $\sigma \in [0,\pi]$, in units where $\ell_{\rm AdS}=1$. The boundary gravitational degree of freedom $t(u)$ determines the breaking of conformal symmetry at low energies. Its intrinsic dynamics is governed by the Schwarzian action \cite{Maldacena:2016upp}. In our case, we additionally consider a set of matter operators $\Op_\alpha$ in $\mathsf{PSL}(2,\mathbf{R})$ representations of lowest weight $\Delta_\alpha$. We drive the boundary with a time-dependent source for the operators $\Op_\alpha$ with time-dependent couplings, so that the action has the form
\be\label{eq:actiongaugenot}
I = I_{\rm Schw}+ \iw \int \text{d}u \, \sum_{\alpha=1}^K g_\alpha (u) \Op_\alpha\,,\qquad 
\ee 
where 
\be 
I_{\rm Schw}=\phi_b \int \text{d}u \left\lbrace {\tanh(\dfrac{t(u)}{2})},u\right\rbrace \,. 
\ee 
and $\phi_b$ is the boundary value of the dilaton. The first term in $I$ is the Schwarzian action, which performs the cooling allowing to keep the temperature finite. The second term introduces the time-dependent sources along the evolution.


\begin{figure}[h]
\centering
\includegraphics[width=.7\linewidth]{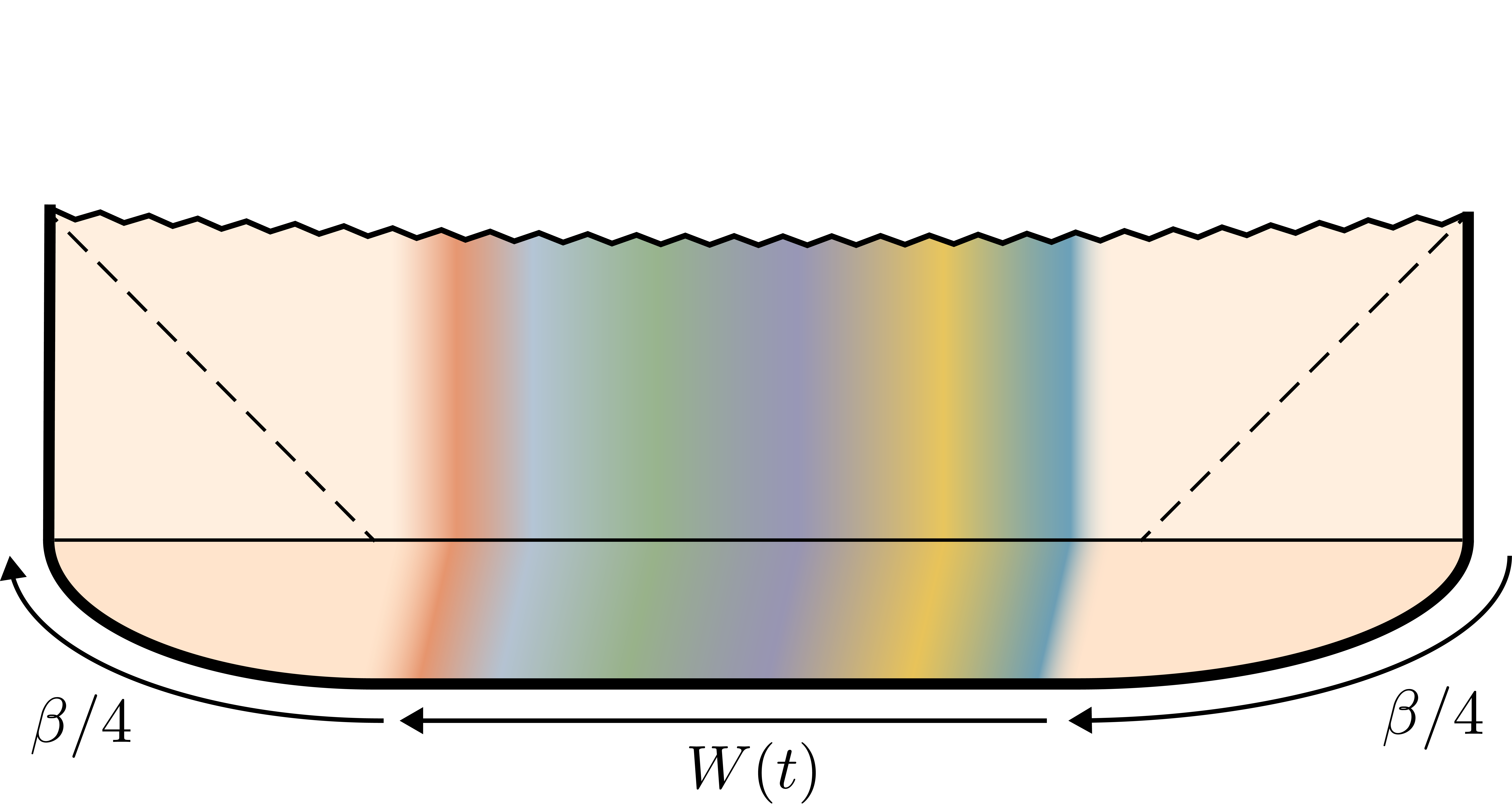}
\caption{A caterpillar, where $W(t)$ is prepared by the Hamiltonian \eqref{eq:coolrhamilt}.}
\label{fig:caterpimodel}
\end{figure}

A schematic representation of the preparation of a state $\ket{W(t,\beta)}$ and its semiclassical dual is shown in Fig.~\ref{fig:caterpimodel}. The time-dependent couplings generate an inhomogeneous matter distribution in the black hole interior. Although the detailed structure of individual realizations is difficult to analyze, coarse-grained geometric features can be extracted by averaging over the ensemble of states prepared at a given time $t$. Such averages, including the norm \eqref{eq:normfintemppreavg} or $\mathsf{LR}$ two-point correlation functions, reveal two key properties \cite{Magan:2024aet,Magan:2025hce}:
\begin{itemize}
\item The average wormhole length grows linearly with the preparation time $t$.
\item The average transverse area (or the value of the dilaton in JT gravity) remains approximately constant, up to small fluctuations, along the wormhole in the black hole interior.
\end{itemize}

\subsubsection{A discrete toy model}

To understand individual instances of the caterpillars it is useful to replace $W(t)$ by some discrete toy version
\be 
\widetilde{W}(t) = U_{t} e^{-\delta t H_0} ...U_2 e^{-\delta t H_0}U_1 \,.
\ee 
where $\delta t$ is finite and $U_i=\exp(-\iw \delta t \Op_i)$ are independent random draws of a ``gate set'' $\mathcal{G} = \{ \exp(-\iw \delta t \Op_\alpha)\}$, where $\Op_\alpha$ is a $\mathsf{PSL}(2,\mathbf{R})$ primary of conformal dimension $\Delta_\alpha$.\footnote{Similar models have appeared in \cite{Lin:2022rzw,Lin:2022zxd,Antonini:2025rmr}, differing in that the infinitesimal unitaries $U_i$ are replaced by operators $\Op_i$ (and that the Euclidean evolution between insertions might taken to be infinite in the BPS context). In those models, the $\mathsf{LR}$ entanglement is found to decrease with the number of insertions, and the states no longer approach generic states in the doubled Hilbert space, but rather products of generic states in single Hilbert spaces. The limit to a vanishing ``semi-quenched'' entropy was demonstrated in \cite{Antonini:2025rmr}.} The corresponding state is shown in Fig.~\ref{fig:toymodel}. In this toy model construction, the geometry contains a discrete number $t$ of matter particles inside the black hole, each with mass $m_i \approx \Delta$. The wormhole length then grows linearly with the total operator insertions $t$. This behavior also characterizes individual realizations, since for $t \gg 1$ the length becomes self-averaging across the ensemble by the central limit theorem.

\begin{figure}[h]
\centering
\includegraphics[width=.7\linewidth]{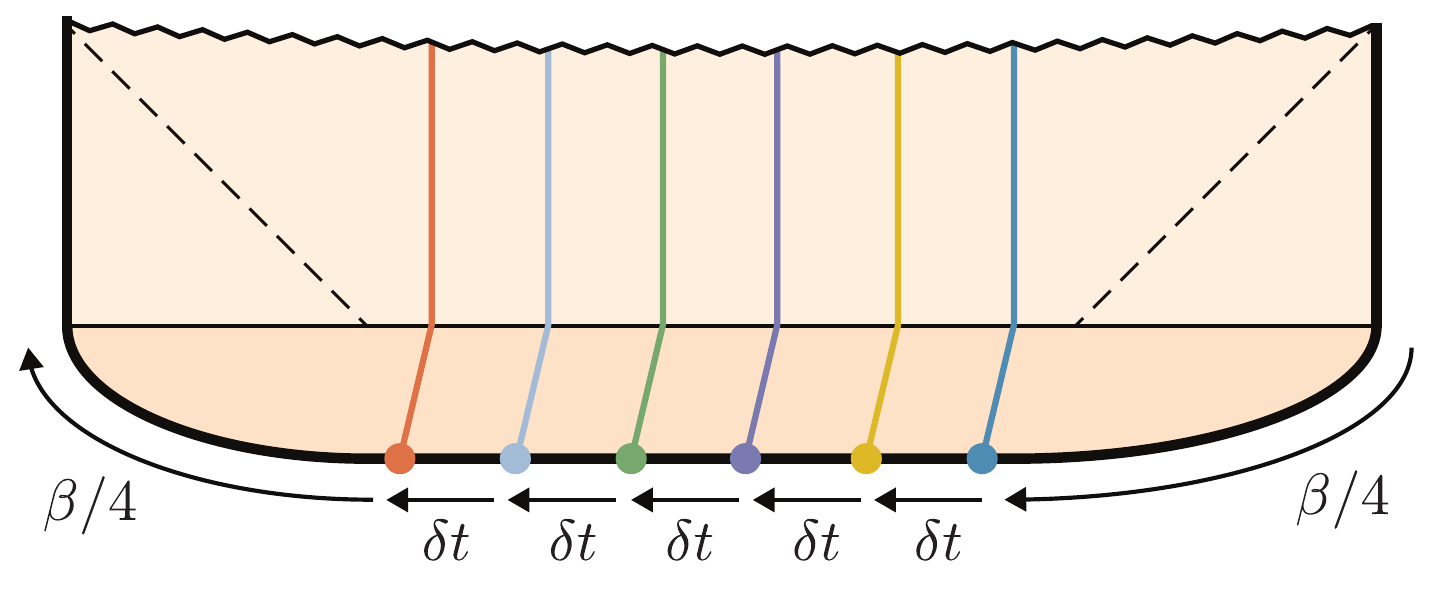}
\caption{Toy model of a caterpillar in terms of heavy local matter insertions with flavors.}
\label{fig:toymodel}
\end{figure}

While this simplified model is convenient for illustrating individual instances in Fig. \ref{fig:caterpimodel}, it shares most of the shortcomings discussed in the introduction and involves an artificially large number of flavors. Nevertheless, one further advantage is that it can be easily considered in general dimensions expanding on \cite{Balasubramanian:2022gmo,Balasubramanian:2022lnw}, and we develop the state-counting in these types of models in Appendix \ref{app:toymodel}.

\subsection{Overlap moments from a replica wormhole}

We consider two cutoff boundaries $a$ and $\bar{b}$, with action \eqref{eq:actiongaugenot} and its complex conjugate, respectively, so that the Brownian couplings are correlated at each point of the two boundaries. We want to compute quantities which are averaged over the couplings. Performing the Gaussian integral over couplings we get an effective action 
\be\label{eq:avgcouplingsJT} 
I_{a\bar{b}}\equiv-\log \left(\int \prod_{\alpha}\mathcal{D} g_\alpha(u) \,e^{-\frac{1}{2 J} \int \text{d}u \sum_{\alpha} g^2_\alpha(u)} e^{-I_a - I_{\bar{b}}^*}\right) = I_{{\rm Schw},a} + I^*_{{\rm Schw},\bar{b}}  + I_{\rm int}\;,
\ee 
where the two boundaries are locally interacting in time via
\be\label{eq:Iint} 
I_{\rm int} = -\frac{\J N}{K} \int \text{d}u \, \sum_{\alpha=1}^K \Op^a_\alpha(u)  \Op^{\bar{b}}_\alpha(u)\,^*\,.
\ee 
This provides the bulk path-integral derivation of the effective Hamiltonian (recall that the operator formalism derivation is presented in Appendix \ref{app:Heff}). Here we used that the renormalized coincident operators square to the identity and omitted additive constants in the action. After imposing the $\mathsf{PSL}(2,\mathbf{R})$ gauge constraints, the two replicas are governed by the gauge-invariant Euclidean effective action of Maldacena-Qi (MQ) \cite{Maldacena:2018lmt}
\be\label{eq:liouvilleaction} 
I_{a\bar{b}} = \dfrac{N}{2} \int \text{d}\tu \,\left(\dfrac{1}{2}\dot{\ell}^2 + 2 e^{-\ell} - \eta e^{-\Delta \ell }\right)  \,,
\ee
where $\tu =  \frac{N}{\phi_b} u = \frac{\mathcal{J}}{\alpha_S}u$ is a dimensionless boundary time (and the total time is $\tilde{t} =\tfrac{\mathcal{J}}{\alpha_S}t$). The relevant degree of freedom is the renormalized geodesic length $\ell = -2\log \dot{\tau}(\tu)$ between boundaries.

We now compute $Z_n(t)$ with the gravitational path integral. The boundary conditions include $n$ closed asymptotically AdS boundaries, with intervals of length $t$ which interact in pairs via the interaction \eqref{eq:Iint}. Any configuration without $a\bar{a}$ correlation (for $a=1,...,n$) will have at least $O(\mathcal{J} N t )$ action \eqref{eq:Iint} and thus its contribution will decay exponentially and become negligible at $\mathcal{J}t \sim O(1)$. Therefore, we seek for a fully connected configuration whose contribution dominates at late times. In particular we seek a configuration which saturates at a non-vanishing value at late times, as expected from the general considerations of Sec. \ref{Sec:countingframework}.

\subsubsection{The propagator}

To construct such connected configuration we proceed in two steps. We first evaluate the Euclidean propagator of the $\ell$-particle in the potential $V(\ell) = 2 e^{-\ell} - \eta e^{-\Delta \ell }$, i.e. the heat kernel
\be 
G^\eta_t(\ell,\ell') = \bra{\ell} e^{-tH(\ell)}\ket{\ell'} \,.
\ee 
The potential is shown in Fig \ref{fig:plot_potential}. From such figure it is transparent the potential includes bound states $\psi_{E_n}(\ell)$ (with discrete energies $E_n<0$) and a continuum of scattering eigenstates $\phi_{E}(\ell)$ (with $E>0$). Therefore, the propagator decomposes as

\begin{figure}[h]
    \centering
    \includegraphics[width=0.5\linewidth]{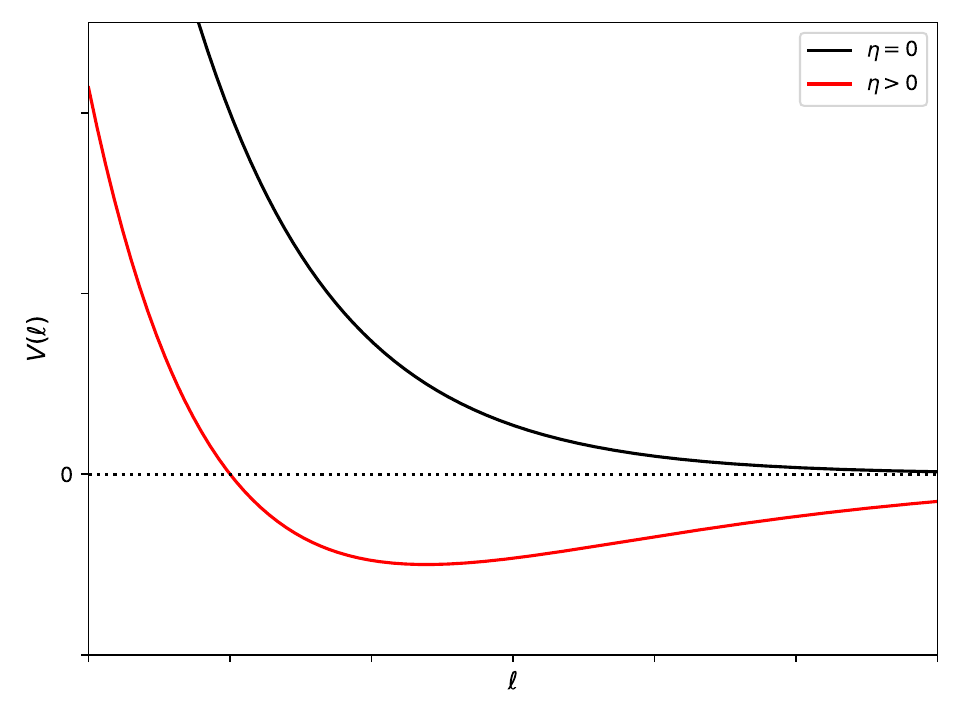}
    \caption{
    Plot of the potential $V(\ell) = 2 e^{-\ell} - \eta e^{-\Delta \ell }$ as a function of $\ell$ for $\Delta = 1/2$ and different values of $\eta$. For $\eta\neq 0$ the spectrum contains both scattering and bound eigenstates.
    }
    \label{fig:plot_potential}
\end{figure}

\be \label{propjt}
G^\eta_t(\ell,\ell')= \int_{E>0} \text{d}E \,e^{-t E}\, \phi_{E}(\ell)\phi^*_{E}(\ell') + \sum_{n=0}^{N_b-1} e^{-t E_n} \psi_{E_n}(\ell)\psi^*_{E_n}(\ell')\;,
\ee 
where $N_b$ is the number of bound states.

The eigenfunctions solve the Schrödinger equation
\be 
\left (-\frac{1}{2}\partial_\ell^2 + 2 e^{-\ell} - \tilde{\eta} e^{-\Delta \ell} \right )\Psi_{E}(\ell) = \frac{E N}{2} \Psi_{E}(\ell)\;,
\ee 
where we have shifted $\ell \rightarrow \ell + 2\log \left (N/2 \right )$ and we have defined $\tilde{\eta} \equiv \eta \left (N/2 \right )^{2-2\Delta}$.

To be fully explicit, let us consider the case of $\Delta = \tfrac{1}{2}$ (free chiral fermion), which is analytically solvable. The scattering eigenstates are
\be
\phi_E(\ell) = \frac{\Gamma(\frac{1}{2}-2\iw s-\tilde{\eta} )}{8^{\frac{1}{2} + 2 \iw s} \, \Gamma(-4\iw s)}  e^{\ell/4} W_{\tilde{\eta}, 2\iw s} \left(8e^{-\ell/2} \right)\,,\qquad s = \sqrt{EN}\,,
\ee
where $W_{\kappa, \mu}(z)$ is the Whittaker function. In particular, $e^{\ell/4} W_{0, 2\iw s}(8e^{-\ell/2}) = \sqrt{\frac{8}{\pi}} K_{2\iw s}(4e^{-\ell/2})$ and this recovers the scattering eigenstates in the Liouville wall \cite{Harlow:2018tqv}. The normalization is chosen in the usual way so that the ingoing coefficient is $1$ as $\ell\rightarrow +\infty$, and the reflection coefficient is a pure phase, given that for large $\ell$, equivalently for $x \equiv 8 e^{-\ell/2} \ll 1$ we have
\be
\label{eq:Whapprox0}
x^{-1/2} W_{\kappa, \mu}(x) \approx x^{-\mu} \frac{\Gamma(2 \mu)}{\Gamma (\frac{1}{2} + \mu - \kappa)} + x^\mu \frac{\Gamma(-2 \mu)}{\Gamma(\frac{1}{2} - \mu - \kappa)} \, .
\ee

For the bound states, the solutions are
\be
\psi_E(\ell) = c_E e^{\ell/4} W_{\tilde{\eta}, 2s} \left(8e^{-\ell/2} \right)\,,\qquad s = \sqrt{-EN}\,,
\ee
where the normalizability at $\ell \to +\infty$ requires the coefficient of the negative power to vanish (first term of (\ref{eq:Whapprox0})), imposing $\Gamma\!\left(\frac{1}{2} + 2s -\tilde{\eta}\right)^{-1} = 0$. This forces a discretization in the energy spectrum of the bound states. More concretely, the discrete energy levels $E_n = -s_n^2/N$ are
\be 
E_n = -\frac{1}{4N}(\tilde{\eta} - \tfrac{1}{2} - n)^2,
\qquad
n = 0,1,\dots,\big\lfloor \tilde{\eta} - \tfrac{1}{2}\big\rfloor\,.
\ee
The number of bound states is $N_b = \lfloor \tilde{\eta} - \tfrac{1}{2} \rfloor +1 \sim \eta N  \gg 1$. Further, for these particular values, the Whittaker function reduces to a Laguerre polynomial, 
\be
W_{\tilde{\eta},2 s_n}(x) \,\propto\, x^{2s_n+\frac{1}{2}} e^{-x/2}\, L_n^{(4s_n)}(x) \, .
\ee
The normalized bound–state wavefunctions are thus
\be
\psi_n(\ell) = c_n\, e^{-s_n\ell}\, e^{-4e^{-\ell/2}}\, L_n^{(4s_n)}\!\left(8e^{-\ell/2}\right)\,,
\ee
where $c_n$ is a normalization constant. Inserting the scattering and bound eigenstates in (\ref{propjt}) we obtain an explicit, albeit quite involved, expression for such propagator. 

\subsubsection{The replica wormhole}

The replica wormhole computing the replica partition function $Z_n$ is built from two building blocks. One of them is the propagator associated to the preparation of the states. For the finite temperature scenario considered here, the ensemble was defined in \ref{sec:finitetempstates}, by composing the random circuit $W(t)$ with cooling Euclidean evolution as  $W(t,\beta) \equiv e^{-\frac{\beta}{4}H_0}W(t)e^{-\frac{\beta}{4}H_0}$. The propagator associated with this ensemble is then obtained by convoluting the MQ propagator ($\eta>0$) described by eq. \eqref{propjt}, with a pair of JT propagators ($\eta =0$), see Fig. \ref{fig:bone}. This leads to
\be 
\mathscr{G}_{t,\beta}(\ell_1,\ell_2) =\int \text{d}\ell \text{d}\ell' \,G^0_{\beta/4}(\ell_1,\ell)G^\eta_t(\ell,\ell') G^0_{\beta/4}(\ell',\ell_2)\,.
\ee

\begin{figure}[h]
\centering
\includegraphics[width = .4\textwidth]{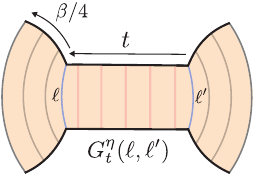}
    \caption{Bone-like propagator, obtained by convoluting the amplitude associated to the MQ eternal traversable wormhole, eq. \eqref{propjt}, appropriate for non-vanishing $\eta>0$ interactions between replicas, with the Euclidean evolution of the non-interacting (Schwarzian) Hamiltonian, as required by the ensemble definition in \ref{sec:finitetempstates}.}
\label{fig:bone}
\end{figure}

The second building block is the gluing of $n$ such propagators to form the ``birdcage wormhole'' represented in Fig. \ref{eq:figurewhJT}. This is done by projecting into the $n$ boundary state prepared by the JT path integral on the hyperbolic $n$-gon, whose wavefunction in the length basis is \cite{Yang:2018gdb,Penington:2019kki}
\be 
I_{n}(\ell_{1\bar{1}},...,\ell_{n\bar{n}}) = 2^{n} \int_0^\infty \text{d}s \,\varrho(s)\, K_{2\iw s}(4 e^{-\frac{\ell_{1\bar{1}}}{2}})...\,K_{2\iw s}(4 e^{-\frac{\ell_{n\bar{n}}}{2}})\,,
\ee 
where 
\be 
\varrho(s) = \dfrac{s}{2\pi^2} \sinh(2\pi s)\,,
\ee
is the density of states of the Schwarzian action \cite{Maldacena:2016upp,Stanford:2017thb,Mertens:2017mtv}.

\begin{figure}[h]
\centering
\includegraphics[width=.9\linewidth]{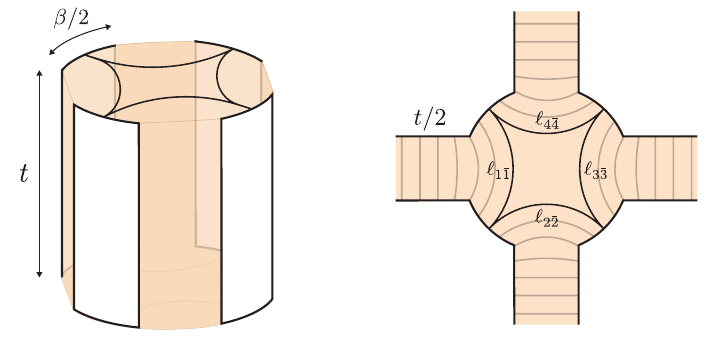}
\caption{Left: The birdcage wormhole has the topology $\Sigma_{0,n}$ of an $n$-punctured two-sphere and provides the dominant contribution to $Z_n(\beta,t)$ for $Jt \gg O(1)$. We depict the case $n=4$. This wormhole plays the analogous role of the pinwheel wormhole in \cite{Penington:2019kki}. Right: Deconstruction of half of the birdcage wormhole into propagators glued by the wavefunction $I(\ell_{1\bar{1}},\ell_{2\bar{2}},\ell_{3\bar{3}},\ell_{4\bar{4}})$, prepared by the JT path integral on the hyperbolic square.}
\label{eq:figurewhJT}
\end{figure}

Therefore, the replicated partition function $Z_n$ is explicitly given by
\be \label{Zngrav}
\notag Z_n(t,\beta)  = \frac{1}{\mathcal{N}^{n}}\int \text{d} \vec{\ell}\,\text{d}\vec{\ell}\,'\,I_{n}(\ell_{1\bar{1}},...,\ell_{n\bar{n}})I_{n}(\ell_{1\bar{1}}',...,\ell_{n\bar{n}}')\,\mathscr{G}_{t,\beta}(\ell_{1\bar{1}},\ell_{1\bar{1}}')...\mathscr{G}_{t,\beta}(\ell_{n\bar{n}},\ell_{n\bar{n}}')\;,
\ee
where $\mathcal{N}$ is the normalization of the states in the ensemble, defined in \eqref{eq:norm}, i.e. it is just $Z_1$. To simplify the dependence on $n$, we notice this expression can be rewritten as
\be
\notag Z_n(t,\beta)  = \frac{1}{\mathcal{N}^{n}} \int ds ~ ds' \varrho(s) \varrho(s') \left ( \alpha^\beta_t(s,s') \right )^n \; ,
\ee
where we have defined
\be
\alpha_t^\beta(s,s') = \int d\ell ~ d\ell' K_{2\iw s}(4 e^{-\frac{\ell}{2}}) \mathscr{G}_{t,\beta}(\ell,\ell') K_{2\iw s'}(4 e^{-\frac{\ell'}{2}}) \;.
\ee
The $s$ dependence of this expression reads
\be 
\alpha_t^\beta(s,s')\,\propto\, e^{-\frac{\beta}{4} \left (E_s + E_{s'} \right )} \int d\ell ~ d\ell' K_{2\iw s}(4 e^{-\frac{\ell}{2}})\, G_t^\eta(\ell,\ell')\, K_{2\iw s'}(4 e^{-\frac{\ell'}{2}}) \, ,
\ee
where we used the orthogonality relation for Bessel functions \cite{Whittaker} and ``$\propto$'' means independent of $s,s'$ factors.

Although these are involved expressions, they simplify at times greater than the gap time, as in all previous cases. Indeed, at large $t$, we can approximate the propagator \eqref{propjt} by keeping the first two bound states
\be
G^\eta_t(\ell, \ell') \approx e^{-tE_0} \left (\psi_0(\ell) \psi_0^*(\ell') + e^{-tE_{\text{gap}}} \psi_1 (\ell) \psi_1^*(\ell') \right ) \, .
\ee
In this approximation we get
\be
\alpha_t^\beta(s,s') \approx e^{-tE_0} \left (\alpha_0(s)\, \alpha^*_0(s') + e^{-tE_{\text{gap}}} \alpha_1(s) \,\alpha^*_1(s') \right ) \, ,
\ee
with
\be
\alpha_n(s)\, \propto\, e^{-\frac{\beta}{4}E_s}\int d\ell K_{2\iw s}(4 e^{-\frac{\ell}{2}}) \psi_n(\ell) \, .
\ee
The proportionality factor does not matter since it will cancel out with the denominator in the expression of the replica partition function $Z_n$. Then, defining
\be 
\varrho_n(s)\equiv e^{-\frac{\beta}{4}E_s} \int d\ell K_{2\iw s}(4 e^{-\frac{\ell}{2}}) \psi_n(\ell)\;,
\ee
and also
\be
\rho_\equil(s) \equiv \frac{\varrho_0(s)}{\int_0^\infty ds' \varrho(s') \varrho_0(s')} \,,\hspace{1cm}  f(s) \equiv \frac{\varrho_1(s)}{\int_0^\infty ds' \varrho(s') \varrho_0(s')} \; ,
\ee
we finally obtain the following formula for the replicated partition function, valid at times much greater than the gap time
\begin{align}\label{znjt}
\notag Z_n(t, \beta) \approx \left ( \int_0^\infty ds \varrho(s) \rho_{\equil}^n(s) \right )^2 &+ n e^{-tE_{\text{gap}}} \left [ \left (\int_0^\infty ds \varrho(s) \rho_\equil^{n-1}(s) f(s) \right )^2 - \right . \\
& \left . - \left ( \int_0^\infty ds \varrho(s) f(s) \right )^2 \left ( \int_0^\infty ds \varrho(s) \rho_{\equil}^n(s) \right )^2 \right ] \;.
\end{align}

\subsection{Microcanonical projection and black hole entropy}

As explained in Sec. \ref{Sec:countingframework}, the dimension of the Hilbert space spanned by $\Omega$ states drawn from the Brownian ensemble is
\be\label{eq:transitiondimensionjt}
d_{\Omega} = \min\lbrace \Omega, \lim_{n \rightarrow 0} Z_n\rbrace \,.
\ee
Expression \eqref{znjt}, albeit seemingly complicated, has an extremely simple $n$ dependence. In fact
\be \label{znjtl}
\lim_{n \rightarrow 0}\, Z_n(t, \beta)=\lim_{n \rightarrow 0}\,\left ( \int_0^\infty ds \varrho(s) \rho_{\equil}^n(s) \right )^2\, .
\ee
This follows because the second term is linear $n$, and the dependence in $n$ of the terms in parenthesis has no effect (the limit leads to the inverse of probabilities, which are positive numbers). Defining the following equilibrium state in a single replica Hilbert space
\be
\rho_{\equil}\equiv \int ds \varrho(s) \rho_{\equil}(s)\,\vert E_s\rangle\langle E_s\vert\;,
\ee
which is diagonal in the energy basis of the Schwarzian theory, we can write the limit as
\be 
\lim_{n \rightarrow 0}\, Z_n \left (t, \beta \right )_{\textrm{Single-sided}} =\text{rank}(\rho_{\text{eq}}) \, .
\ee
This is precisely the expectation from general grounds described in \ref{Sec:countingframeworkexpectation}. Now, of course, the rank of such matrix is infinite here. This is also expected, since the state ensemble produces non-vanishing and random tails for all energy sectors of the theory. But the replica partition functions have the usual simple form, namely disregarding the square in the right hand side of \ref{znjtl}, that only arises since we are considering a two boundary scenario, we have that the partition function has the form $ \int_0^\infty ds \varrho(s) \rho_{\equil}^n(s)$. One can then consider a microcanonical projection. More concretely, following \cite{Penington:2019kki}, where the function $y(s)$ there  gets replaced by $\rho_{\equil}(s)$ here, the projected microcanonical partition functions become $\mathsf{Z}_n (s)\equiv \varrho(s) \rho_{\equil}^n(s)\Delta s$.\footnote{We could include the corrections at large times as well in the projected partition functions, but these are still linear in $n$ and do not contribute to the computation of the Hilbert space dimension.} The Hilbert space dimension spanned by the projected ensemble is then
\be 
d_{\Omega}(t) = \text{min}\lbrace \Omega, d_s^2\rbrace\;,
\ee 
where
\be 
d_s = \dfrac{s}{2\pi^2} \,\sinh(2\pi s)\,.
\ee 
This is the microcanonical Hilbert space dimension of the near-extremal black hole \cite{Maldacena:2016upp,Stanford:2017thb,Mertens:2017mtv,Iliesiu:2020qvm}.

Two remarks follow. First, we note that in the previous formulas we have always disregarded the topological term $e^{S_0}$ that appears in JT gravity besides the Schwarzian action \cite{Maldacena:2016upp}. Including such contribution is straightforward by changing $\varrho(s)\rightarrow e^{S_0}\varrho(s)$. Second, in Sec. \ref{Sec:countingframework}, it was shown that the variance of the dimension of the Hilbert space spanned by the microstates drawn from the Brownian ensembles is zero, and that this remains the same when projecting into a subspace. This implies that the variance in the present computation only comes from the definition of the microcanonical energy window. More precisely, if we define a microcanonical energy window as $E_i\in [E,E+\Delta]$, the dimension of the Hilbert space spanned by the Brownian circuits above will be the number of eigenvalues $n_{E,\Delta}$ within such window. This statement has zero variance when averaging over the disordered couplings defining the Brownian motion. But the number itself has a variance in the associated JT random matrix ensemble \cite{Saad:2019lba} that will be analyzed elsewhere.

\subsection{States in two SYK models and overlaps}
\label{sec:SYKfintemp}

It is clarifying to have a specific microscopic model of the previous JT caterpillars and overlaps in mind. This helps to visualize that the counting remains structurally the same whether we perform it at the microscopic or gravitational level.\footnote{This of course parallels the discussion in Secs. \ref{sec:BrownianSYKexact} and \ref{sec:largeNsyk} concerning the infinite temperature scenario.} To this end, following \cite{Maldacena:2018lmt}, we consider two copies of the SYK model, and define the ensemble of states $\ket{W(t,\beta)}$ of Section \ref{sec:finitetempstates} in this Hilbert space. Microscopically, the operator $W(t)$ is generated as in \eqref{eq:W(t)torder} by a cooling $p$-body SYK Hamiltonian together with a real-time Brownian $q$-SYK Hamiltonian,
\be\label{eq:HamiltonianSYKcooling}
H(t) = -\iw H_{\text{SYK}} + \iw^{\frac{q}{2}}\sum_{|\alpha|=q} g_{\alpha}(t)\,\psi_\alpha\,,
\ee
where the $p$-body SYK Hamiltonian is defined in the standard manner, see e.g. \cite{Maldacena:2018lmt} for the purposes of this discussion. As usual, we now consider a family of $\Omega$ draws of the ensemble
\be 
\mathsf{F}^t_{\Omega} = \lbrace\ket{W_i(t,\beta)}: i=1,...,\Omega \rbrace \;,
\ee
and we ask what the dimension $d_\Omega(t)$ of the linear span of $\mathsf{F}^t_{\Omega}$ is.

The overlap moments $Z_n(t,\beta)$, defined in \eqref{eq:Znetaoverlapfintemp}, are governed by the two-replica effective Hamiltonian $H_{\eff}^{a\bar{b}}$, defined generally in \eqref{eq:tworeplicaeffHfintemp}. In the present scenario, this Hamiltonian coincides with the MQ Hamiltonian \cite{Maldacena:2018lmt},
\be\label{eq:MQ}
H_{\eff}^{a\bar{b}} = H^a_{\text{SYK}} + H^{\bar{b}\,*}_{\text{SYK}} - \frac{\J N}{K} \sum_{|\alpha|=q} \psi^a_\alpha \psi^{\bar{b}}_\alpha + \J N \, .
\ee
Here we have made use of the fact that Majorana strings square to the identity and introduced the rescaled coupling $J = \frac{\J N}{K}$ with $K = {N \choose q}$. Recall that $q$ is even, so that \eqref{eq:HamiltonianSYKcooling}, and therefore $W(t)$, are bosonic operators. At low energies (in units of the SYK couplings), such Hamiltonian admits a semiclassical description in the regime analyzed in \cite{Maldacena:2018lmt}. The only relevant properties for us is that such Hamiltonian is gapped with a gap $E_{\rm gap} = O(N^0)$ and a unique ground state $\ket{\rm GS}$ which admits a semiclassical description as a connected spatial wormhole. The $\ket{\rm GS}$ is similar to a finite-temperature thermofield double state $\ket{\rm TFD}\propto \sum_i e^{-\beta E_i/2}\vert E_i, E_i\rangle$ at some effective temperature. The number of first excited states is $N_* = N$. Hence the associated gap time given by \eqref{eq:gaptime} is $t_\ast = O(\log N)$, of the order of the scrambling time of the black hole.

Now, defining $\rho_t = (e^{-t H_{\eff}^{a\bar{b}}})^{T_\perp}$, see \eqref{eq:defrhot} and lemma \ref{lemma:densitymatrix}, we have that the overlap moments correspond to the following R\'enyi moments
\be\label{eq:momentSYKfintemp} 
Z_n(t,\beta) = \dfrac{\text{Tr} {\rho}^n_{t,\beta}}{\left(\text{Tr} {\rho}_{t,\beta}\right)^n}\,,
\ee 
for the unnormalized density matrix 
\be 
{\rho}_{t,\beta} \equiv \left(e^{-\frac{\beta}{4}H_{\rm SYK}} \otimes e^{-\frac{\beta}{4}H_{\rm SYK}}\right) \rho_t \left(e^{-\frac{\beta}{4}H_{\rm SYK}} \otimes e^{-\frac{\beta}{4}H_{\rm SYK}}\right)\;.
\ee 
This is again the unnormalized density matrix associated to the average state of the ensemble. Taking the limit $n\rightarrow 0$ and using \eqref{eq:transitiondimension} we recover 
\be\label{eq:d_OmegaSYKfintemp}
    d_\Omega(t) =  \text{min}\lbrace \Omega, \rank(\rho_{t,\beta})\rbrace\,.
\ee 
Again we have that $\rank(\rho_{0,\beta}) =1$ and we can argue very generally that  $\rank(\rho_{t,\beta})$ is constant (and maximal) for any $t>0$.

To see this explicitly, note that after the gap time $t \gg t_\ast$, we can use the properties of the MQ Hamiltonian reviewed above to expand it as \footnote{For simplicity, here we subtract the ground-state energy of the MQ Hamiltonian. This does not affect the overlap moments.}
\be 
e^{-tH_{\text{eff}}^{a\bar{b}}} = \ket{\text{GS}}\bra{\text{GS}} + e^{-t\Egap}\,\Pi_{\text{gap}} + \cdots\,,
\ee 
where $\Pi_{\text{gap}}$ is the orthogonal projector onto the subspace of first excited states, and the dots denote higher excitations that decay more rapidly. The $p$-norm distance 
$\|e^{-tH_{\text{eff}}^{a\bar{b}}} - \ket{\text{GS}}\bra{\text{GS}}\|_p$ 
is therefore small, allowing us to approximate the operator by the ground-state projector.

Substituting this large time expansion into \eqref{eq:momentSYKfintemp} yields 
${\rho}^n_{t,\beta} \approx \rho_\equil \otimes \rho_\equil$, where we now define the single-replica (unnormalized) equilibrium density matrix
\be
\rho_\equil \;\equiv\; e^{-\frac{\beta}{4}H_{\text{SYK}}}\, \sqrt{\rho_{\mathsf{L}}}\,e^{-\frac{\beta}{4}H_{\text{SYK}}}, 
\qquad 
\rho_{\mathsf{L}} \;=\; \Tr_{\mathsf{R}}\!\big[\ket{\text{GS}}\bra{\text{GS}}\big].
\ee
Hence, after the gap time, the overlaps approach constant values,
\be
Z_n(t,\beta) \approx
\frac{\big[\Tr(\rho_\equil^n)\big]^2}{\big[\Tr(\rho_\equil)\big]^{2n}} + \cdots,
\qquad t \gg t_\ast\,.
\ee
From this expression it follows that the corresponding dimension becomes constant, corresponding to the square of the Schmidt rank of $\ket{\text{GS}}$ across the $\mathsf{L}|\mathsf{R}$ bipartition,
\be 
\rank(\rho_{t,\beta}) \;=\; \rank(\rho_\equil)^2  \;=\; \rank(\rho_{\mathsf{L}})^2\,.
\ee
Since $\ket{\rm GS}$ is close to a finite temperature TFD state, the rank will be maximal.

\subsection{Higher dimensions and universal structure}

The JT and SYK analysis suggests that it is possible to extend this discussion to higher dimensions, although we leave the specific details for future work. The main reason behind this is that the structure of the replica partition functions is expected to be completely universal. The only assumption required to achieve such universality is the one described in \cite{Magan:2024aet,Magan:2025hce}, when analyzing black hole caterpillars in higher dimensions. There it was assumed that, as for SYK or JT gravity, the effective Hamiltonian
\be
H_{\mathrm{eff}}^{a\bar{b}} = H_0^{a} + H_0^{\bar{b}}\,^* + \dfrac{J}{2} \sum_{\alpha=1}^K \left(\Op_\alpha^{a} - \Op_\alpha^{\bar{b}*} \right)^2\,,
\ee
associated to the random circuit of a holographic CFT, where $H_0$ is the CFT Hamiltonian and $\mathcal{O}_\alpha$ are operators in such CFT, has a unique vacuum
\be
H_{\mathrm{eff}}^{a\bar{b}}\vert \text{GS}\rangle = E_0\,\vert \text{GS}\rangle\;,\label{eq:tworeplicaeffHfintemp2}
\ee
with large entanglement between the replicas and diagonal in the energy basis
\be\label{eq:GSdef} 
\ket{\GS} = \sum_{i} c_i \ket{E^*_i} \otimes \ket{E_i}\,,\quad\quad \rho(\GS) = \sum_{i} c_i \ket{E_i} \bra{E_i}\,,
\ee 
i.e. resembling a thermofield double state. Notice $\rho(\GS)$ is the Choi operator that corresponds to $\ket{\GS}$. To ensure this, the CFT Hamiltonian has to be part of the drive operators as well. See \cite{toappear} and prior \cite{Cottrell:2018ash,Magan:2024aet} for arguments and examples supporting these statements.

Given such assumption and the definition of the circuits \ref{sec:finitetempstates}, namely through the Choi state associated to the operator $W(t,\beta) \equiv e^{-\frac{\beta}{4}H_0}W(t)e^{-\frac{\beta}{4}H_0}$, it is natural to define the equilibrium density matrix
\be\label{eq:eqstate}
\rho_{\equil} = \dfrac{1}{\mathcal{N}}
\,e^{-\frac{\beta}{4}H_0} \rho(\GS) e^{-\frac{\beta}{4}H_0} \,.
\ee
This equilibrium state is normalized by
\be\label{eq:normalizationdef} 
\mathcal{N}\, \equiv Z(\beta)^{1/2}\,\langle \TFD|\GS\rangle\;,
\ee 
i.e. by the norm of state at late times. These definitions imply that if $\ket{\GS}=\vert\textrm{TFD}\rangle$ for some effective temperature $\beta$, then $\rho_\equil$ is the thermal density matrix at such temperature. In general $\rho_\equil$ will be only ensemble equivalent to the thermal density matrix.

Again, after the gap time $t \gg t_\ast$, we can approximate
\be \label{expeff}
e^{-tH_{\text{eff}}^{a\bar{b}}} = \ket{\text{GS}}\bra{\text{GS}} + e^{-t\Egap}\,\Pi_{\text{gap}} + \cdots\,,
\ee 
where $\Pi_{\text{gap}}$ is the orthogonal projector onto the subspace of first excited states. For holographic CFTs, we expect such gap time to be $\Egap t_\star = O(1)$ since that is the number of first excited states over the putative traversable wormhole.

If we define the superoperator
\be \label{2super}
\Phi_t^{a\bar{b}} \equiv\, \left(\, e^{-\frac{\beta}{4}H_0}\otimes e^{-\frac{\beta}{4}H_0}\right)\,e^{-tH_{\text{eff}}^{a\bar{b}}}\, \left(\, e^{-\frac{\beta}{4}H_0}\otimes e^{-\frac{\beta}{4}H_0}\right)\;,
\ee
then at late times we can approximate $\Phi_t^{a\bar{b}}\approx \mathcal{N}^2\,\vert\rho_\equil\rangle\langle\rho_\equil\vert$, see Fig \ref{fig:Heff_tinf}.
\begin{figure}[h]
    \centering
    \includegraphics[width=0.5\linewidth]{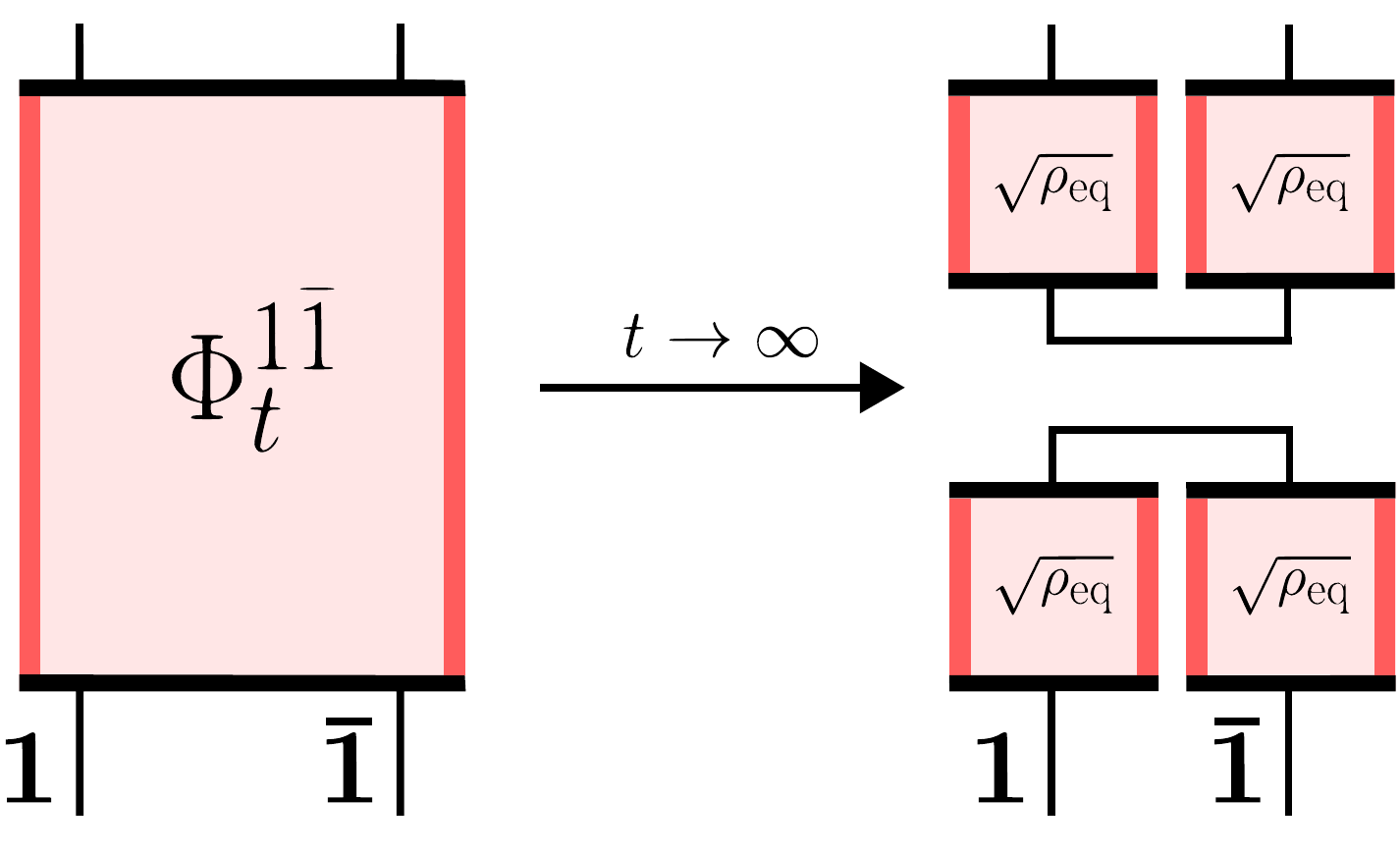}
    \caption{At times much greater than the gap time $t_\star$, the two-replica interacting superoperator $\Phi_t^{a\bar{b}}$ simplifies to a time independent equilibrium state determined by $\rho_\equil$. We have omitted a normalization factor $\mathcal{N}^2$ in the Figure. }
    \label{fig:Heff_tinf}
\end{figure}
Corrections are easily computed using the subleading in time terms in \eqref{expeff}. Then, on quite general grounds we obtain the following universal expansion of the replicated partition functions
\be\label{mainsemi}
Z_n(t) \approx  \text{Tr}(\rho_\equil^n)^2 + n   \frac{e^{-\Egap t}}{\mathcal{N}^2} \Bigl (\langle e^{-\frac{\beta}{4}H_0} \rho_\equil^{n-1} e^{-\frac{\beta}{4}H_0} \vert \Pi_{\text{gap}} \vert e^{-\frac{\beta}{4}H_0} \rho_\equil^{n-1} e^{-\frac{\beta}{4}H_0} \rangle - \text{Tr}(\rho_\equil^n)^2 \langle e^{-\frac{\beta}{2}H_0} \vert \Pi_{\text{gap}} \vert e^{-\frac{\beta}{2}H_0} \rangle \Bigr )  \;,
\ee
where we notice that
\be
\mathcal{N}=Z_1(t) \approx \frac{1}{Z(\beta)}  \left (\text{Tr} \left(e^{-\frac{\beta}{4}H_0}\, \rho (\GS)\,e^{-\frac{\beta}{4}H_0}\right)^2  +  e^{-\Egap t} \langle e^{-\frac{\beta}{2}H_0} \vert \Pi_{\text{gap}} \vert e^{-\frac{\beta}{2}H_0} \rangle \right ).
\ee
In particular the late time stationary value of $Z_n$ is computed by the $n$-Renyi entropy of the equilibrium state, see Fig. \ref{fig:Heff_tinf_Z3}. This generalizes the universal result in \cite{Balasubramanian:2022gmo,Balasubramanian:2022lnw}, which arises in the special case in which $\rho_\equil\approx\rho_\beta$.

From this formula, the dimension of the Hilbert space spanned by these microstates is the expected one by now, i.e.
\be 
d_{\Omega}(t) = \text{min}\lbrace \Omega, \text{rank}(\rho_{\text{eq}})^2\rbrace\;,
\ee 
shown for times greater than the gap time $t_\star$ but expected form smaller times as well.

\begin{figure}[h]
    \centering
    \includegraphics[width=1\linewidth]{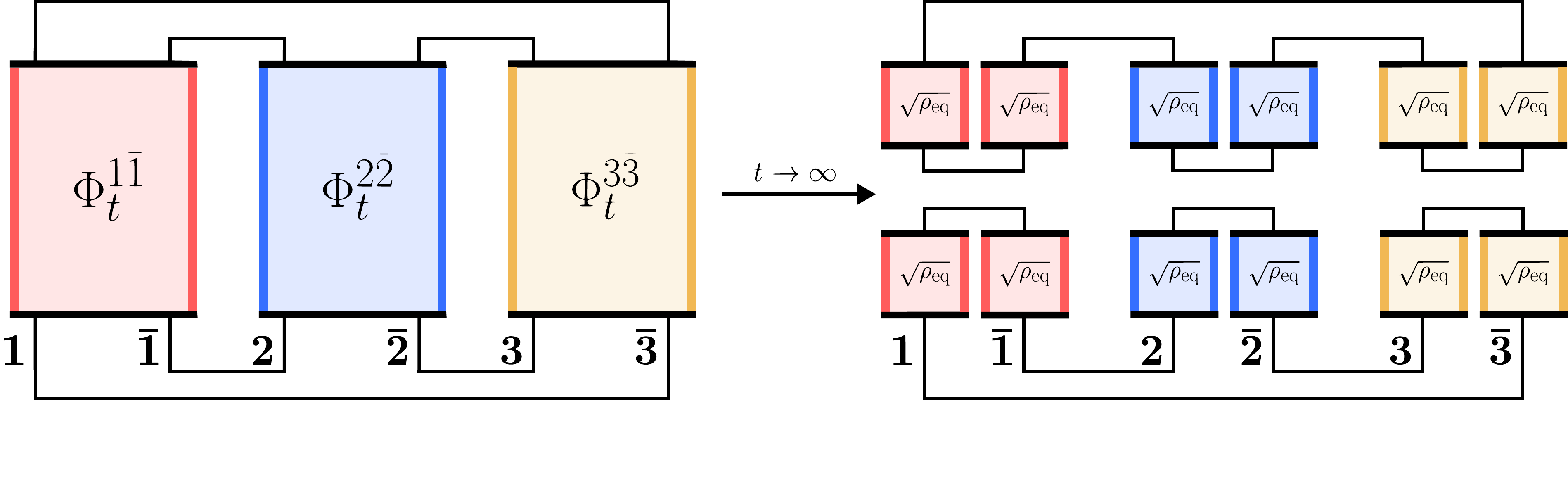}
    \caption{Circuit diagram describing the evaluation of $Z_3$. The large time expansion of the two-replica interacting superoperator $\Phi_t^{a\bar{b}}$ trivializes the computation. We have omitted a normalization factor $\mathcal{N}^{2n}$ in the Figure.}
    \label{fig:Heff_tinf_Z3}
\end{figure}

Since we do not have a precise example of traversable wormhole in higher dimensions we cannot properly build the wormholes that compute $Z_n$. But we expect a wormhole very similar to the JT bird-cage constructed in the previous section, namely we have the $n$-bars of the cage which are the gravitational duals to the effective two-replica evolution operator $e^{-tH_{\text{eff}}^{a\bar{a}}}$, for $a=1,\cdots ,n$, and we have to glue those at future and past infinites with the caps, see Fig \ref{eq:figurewhJT}. Then the action of such wormholes should take the form \eqref{Zngrav}, just by changing the $l$-basis to whatever semiclassical basis is convenient in the higher dimensional theory.

\section{Conclusions}\label{SecVI}

In the introduction we described various caveats pertaining to the recent statistical explanations of black hole entropy and the Page curve. Let us start now the discussion by briefly describing how the approach presented in this article overcomes them.

\begin{enumerate}
    \item {\it Atypical families of black hole microstates}. The microstates studied above were constructed with Brownian motions that, at each time, have computable randomness \cite{Magan:2024aet,Magan:2025hce}. In particular, at large times they become arbitrarily close to typical states. The ability to universally derive \eqref{eq:bhe} by counting diverse families of caterpillars then addresses this question.

    \item {\it Microscopics of shell operators and overlaps}.  While the precise microscopic definition of the shell operator is not under fair control, the microstates used in this article are perfectly defined in the quantum theory through the disordered, time-dependent couplings, associated with particular driving operators. On one hand, this construction grants microscopic control (at least in theory) over overlap moments. On the other hand, this construction is universally applicable to any quantum system.

 \item {\it Implicit statistics}. While the statistics of shell operators is only computable in gravity, the statistics of the Brownian microstates is computable at both sides, through a solid universal and sensible basis, controlled by the mapping of the problem to an effective statistical mechanical system. Also, from the bulk perspective in quantum gravity, each realization of the random couplings prepares a distinct caterpillar with specific semiclassical but erratic features. Then, in this framework, one considers an explicit ensemble of states, draws a finite number of instances, and asks for the Hilbert space dimension the states span. Hence, there is no obvious ``factorization puzzle'' for the overlap moments of the ensemble, which are the inherently ensemble-averaged objects at both sides, from which we extract the dimension. Equivalently, in this scenario we have a precise map between the ensemble in the boundary and the ensemble in the bulk.

\item {\it Variance of the Hilbert space dimension}. The microscopic control over the Brownian microstates allows to show there is no variance in the counting method. This is intuitively clear since the Brownian motion creates absolutely continuous probability distributions in the Hilbert space, and getting two parallel vectors has zero probability. In finite temperature scenarios, the ``variance'' then only comes from the definition of the energy window and not from the counting method, and it can be precisely obtained if the spectrum statistics is known. This will be developed elsewhere.

 \item {\it Large mass limit}. Finally, the physical continuous parameter playing the role of the shell proper mass is the Brownian time $t$. In the present scenario, the replica partition functions can be computed at finite $t$, solving the present issue. In fact, as shown in Appendix \ref{app:resol}, we can generalize the Schwinger-Dyson equation used in \cite{Penington:2019kki} so that we include scenarios in which the  one point function of the inner product is non-diagonal. These required generalizations do not affect the computation of the Hilbert space dimension.

  \item {\it Semiclassical limit}. We have performed the counting microscopically in various many-body quantum systems, and we have also carried it out in the semiclassical limit in the infinite-temperature Brownian SYK and in JT. In all cases it has been immediate to verify that the replica partition functions display the universal late time structure $Z_n(t)\approx \text{Tr}(\rho_\equil^n)^2 +\# n e^{-\Egap t}+\# n^2 e^{-2\Egap t}+\cdots$. The semiclassical limit does not interfere with this structure, in particular it leaves the critical $n\rightarrow 0$ limit intact. This limit implies that the dimension is always $\text{rank}(\rho_{\text{eq}})^2$. The semiclassical limit just affects the precision with which we compute $\Egap$, and whether we can write $\rho_\equil$ as a discrete sum or as a continuous integral. The power of the gravitational path integral is here simply understood as a universal behavior in quantum mechanics.
 
\end{enumerate}

In summary, constructing a linear basis for the Hilbert space of a many-body quantum system is, in principle, straightforward: absent fine-tuning, any collection of $\Omega \geq d$ states will span the full space of dimension $d$. Collections of such states induced from Brownian circuits are specially convenient in this regard. In this work, we have shown that the same idea applies to black holes. This construction  demonstrates that the state-counting derivation of black hole entropy with the gravitational path integral is remarkably robust. Also, it is worth noticing that there is no need for long ER bridges to build bases. In particular, the present results shows that ER bridges with lengths of the order of the scrambling time suffice. Moreover, the results in the quantum mechanical many-body systems suggest that ER bridges of any size suffice.

We end by describing some interesting avenues for future research. First, at a technical level, an interesting and important open problem, already emphasized in \cite{Magan:2024aet,Magan:2025hce}, is the explicit construction of caterpillars and their associated average geometries over couplings, including the disk and replica wormhole saddle points, in higher dimensional Einstein gravity. From the structure of the effective Hamiltonian, the average case is closely related to the explicit construction of eternal traversable wormholes from time-independent double-trace perturbations in higher dimensions. In principle, by choosing suitable drive operators of the Brownian Hamiltonian, like heavy enough operator sourcing localize worldlines, this should be possible. Second, it would also be interesting to better understand if the ground state of the two-replica effective Hamiltonian can provide a high-fidelity thermofield double state in the large-$N$ limit for generic chaotic quantum systems \cite{Maldacena:2018lmt,Cottrell:2018ash,toappear}. In fact, this problem has practical applications in quantum simulation too. Finally, as remarked in \cite{Magan:2024aet,Magan:2025hce}, we find important to explore the application of these new constructions to important open problems in black hole physics. The first concerns the physicality of the black hole interior, given the extreme macroscopic linear dependencies between different geometries and the fact that most of them are almost indistinguishable from random states. The second concerns the validity of the equivalence principle at the event horizon. In all cases, renewed hope stems from the microscopic control we have within this construction, as well as technical control at finite circuit time.

\section*{Acknowledgments}

We are indebted to Martin Sasieta for discussions and collaboration on different stages of this project. We also wish to thank Brian Swingle for discussions. E.B thanks hospitality of Instituto Balseiro where part of this work was developed. The work of J.M is supported by a Ramón y Cajal fellowship from the Spanish Ministry of Science. In the first stages of this project, the work of J.M was supported by Conicet, Argentina. The work of L.M is supported by Conicet, Argentina.

\appendix

\section{Derivation of the effective Hamiltonian}
\label{app:Heff}

This appendix gives a comprehensive derivation of the effective Hamiltonian that arises from averaging over Brownian couplings. Within the path integral formalism, such relations follow from interpreting the couplings as Hubbard–Stratonovich fields and performing the Gaussian integrals (see e.g. \cite{Cardy:2018sdv}). In fact, we already did this in the JT theory, in \eqref{eq:avgcouplingsJT}. Here we instead take a broader perspective and perform the derivation directly in the operator formalism.

Given a scalar or operator-valued functional $A[g_\alpha(t)]$, its ensemble average over the white-noise correlated couplings is defined as
\be\label{eq:ensemblecouplings}
    \overline{A} \equiv \int \prod_{\alpha}\mathcal{D} g_\alpha(t) \,e^{-\frac{1}{2 J} \int \text{d}t \sum_{\alpha} g^2_\alpha(t')} A[g_\alpha(t)]\,,
\ee
where the measure $\prod_{\alpha}\mathcal{D} g_\alpha(t)$ is normalized so that $\overline{1} = 1$. 

\subsection*{Infinite temperature}

For the discussion in Section \ref{Sec:ensemblesinftemp}, the distribution of Brownian couplings defines the Brownian Hamiltonian \eqref{eq:BrownianH}, which in turn induces an ensemble of unitaries $U(t)$ via the definition \eqref{eq:timeorderedinfinite} and associated Choi states $\ket{U(t)}$.

The relevant moments for deriving \eqref{eq:id1} and \eqref{eq:squareoverlapavginftemp} are
\be\label{eq:appendixAmoments}
 \overline{U(t)}\,,\qquad   \overline{ U(t) \otimes U(t)^*}\,.
\ee
From the Markovian nature of the Brownian couplings it follows that $(\overline{U(t)})^2 = \overline{U(2t)}$. In terms of these two moments, the quantities of interest in Section \ref{Sec:ensemblesinftemp} can be written as
\begin{gather} 
\overline{\bra{U_1(t)}\ket{U_2(t)}} = \frac{1}{d}\text{Tr}(\overline{U(2t)})\,,\label{eq:appaid1}\\[.2cm]\qquad Z_n(t) = \bra{\eta} \overline{ U(t)^1 \otimes U(t)^{\bar{1}}\,^*}\otimes \cdots \otimes \overline{ U(t)^n \otimes U(t)^{\bar{n}}\,^*}\ket{\eta}\,,\label{eq:appaid2}
\end{gather}
where in the second expression the superscripts indicate the replicas on which the unitaries act.

We therefore need to evaluate the explicit averages over couplings in \eqref{eq:appendixAmoments}. Expanding the time-ordered exponential gives
\begin{align}          \label{Udet}
    \notag U(t) &=\sum_{n=0}^\infty \frac{(-\iw)^n}{n!} \mathsf{T} \left \lbrace\prod_{i=1}^n \int_0^t \text{d}t_i \sum_{\alpha_i = 1}^K g_{\alpha_i}(t_i) \mathcal{O}_{\alpha_i} \right \rbrace \\
     &= \sum_{n=0}^\infty (-\iw)^n \int_0^t \text{d}t_1 \int_0^{t_1} \text{d}t_2 \ldots \int_0^{t_{n-1}} \text{d}t_n  \sum_{\alpha_1 \ldots \alpha_n = 1}^K g_{\alpha_1}(t_1) \ldots g_{\alpha_n}(t_n) \mathcal{O}_{\alpha_1} \ldots \mathcal{O}_{\alpha_n}\,.
\end{align}
Taking the ensemble average,
\be\label{eq:avgU(t)2}
    \overline{U(t)} = \sum_{n=0}^\infty (-\iw)^n \int_0^t \text{d}t_1 \ldots \int_0^{t_{n-1}} \text{d}t_n  \sum_{\alpha_1 \ldots \alpha_n = 1}^K \overline{g_{\alpha_1}(t_1) \ldots g_{\alpha_n}(t_n) } \,\mathcal{O}_{\alpha_1} \ldots \mathcal{O}_{\alpha_n}\,.
\ee

Since the couplings in \eqref{eq:ensemblecouplings} are Gaussian and white-noise correlated, all odd moments vanish. For the even moments, Wick’s theorem gives
\be
    \overline{g_{\alpha_1}(t_1) \ldots g_{\alpha_{2n}}(t_{2n})} 
         =  J^n \sum_{ [\pi]\in \mathcal{P}_{2n}} \delta_{\alpha_{\pi(1)} \alpha_{\pi(2)}} \delta (t_{\pi(1)} - t_{\pi(2)}) \ldots \delta_{\alpha_{\pi(2n-1)} \alpha_{\pi(2n)}} \delta (t_{\pi(2n-1)} - t_{\pi(2n)})\,,\label{eq:Wickcouplings}
\ee
where $\mathcal{P}_{2n} = \text{Sym}(2n)/H$ denotes the set of all pairings, expressed as the coset by the subgroup $H < \text{Sym}(2n)$ that preserves the standard pairing $\lbrace \lbrace 1,2\rbrace, \ldots, \lbrace 2n-1,2n\rbrace \rbrace $.

Substituting \eqref{eq:Wickcouplings} into \eqref{eq:avgU(t)2}, the time ordering implies that the only nonvanishing contribution among the pairings $[\pi]$ is the identity element $[e]$.  This gives
\be
    \overline{U(t)} = \sum_{n=0}^\infty \left (-J \right )^n \int_0^t \text{d}t_1 \ldots \int_0^{t_{2n-1}} \text{d}t_{2n} \sum_{\alpha_1 \ldots \alpha_{2n} = 1}^K \delta_{\alpha_1 \alpha_2} \delta (t_1 - t_2) \ldots \delta_{\alpha_{2n-1} \alpha_{2n}} \delta (t_{2n-1} - t_{2n}) \mathcal{O}_{\alpha_1} \ldots \mathcal{O}_{\alpha_{2n}}\,.
\ee
Now, using the identity
\be
    \int_0^{y} \text{d}x \; f(x) \; \delta(y - x) = \frac{1}{2} f(y)\,,
\ee
that follows from an ``even'' definition of the delta function, we obtain
\be
    \int_0^t \text{d}t_1 \ldots \int_0^{t_{2n-1}} \text{d}t_{2n} \; \delta (t_1 - t_2) \; \delta (t_3 - t_4) \ldots \delta (t_{2n-1} - t_{2n}) = \frac{1}{n!} \left (\frac{t}{2} \right )^n\,.
    \label{deltas}
\ee
The resulting expression for the average of $U(t)$ is an Euclidean evolution for time $t$:
\begin{align}
    \notag \overline{U(t)} &= \sum_{n=0}^\infty \left (-\frac{Jt}{2} \right )^n \frac{1}{n!} \sum_{\alpha_1 \ldots \alpha_{2n} = 1}^K \delta_{\alpha_1 \alpha_2} \ldots \delta_{\alpha_{2n-1} \alpha_{2n}} \mathcal{O}_{\alpha_1} \ldots \mathcal{O}_{\alpha_{2n}}  \\
     &= \sum_{n=0}^\infty \left (-\frac{Jt}{2} \right )^n \frac{1}{n!} \left (\sum_{\alpha = 1}^K \mathcal{O}^2_{\alpha} \right )^n = e^{-tH_{\eff}}\,,
\end{align}
with the single-replica effective Hamiltonian
\be 
H_{\eff} = \frac{J}{2} \displaystyle\sum_{\alpha = 1}^{K} \mathcal{O}^2_\alpha\,.
\ee
Substituting this into \eqref{eq:appaid1} yields the identity \eqref{eq:id1}.

The analogous identity for $U(t) \otimes U(t)^*$ follows from the identification
\be
    U(t)^1 \otimes U(t)^{\bar{1}}\,^*=\mathsf{T} \left\lbrace e^{-\iw \int_0^t \text{d}t' H(t')}\right\rbrace\,,\qquad   H(t) = \sum_{\alpha = 1}^{K} g_\alpha (t) \mathcal{O}^{1\bar{1}}_\alpha\,,
\ee
with $\mathcal{O}^{1\bar{1}}_\alpha =  \mathcal{O}^{1}_\alpha - \mathcal{O}^{\bar{1}}_\alpha\,^*$. Repeating the same steps as above gives
\be\label{eq:appAUotimesU*}
    \overline{U(t)^1 \otimes U(t)^{\bar{1}}\,^*} =  e^{-tH^{1\bar{1}}_{\eff}}\,,
\ee
with the two-replica effective Hamiltonian
\be
    H^{1\bar{1}}_{\eff} = \frac{J}{2} \sum_{\alpha = 1}^{K} \left (\mathcal{O}^{1}_\alpha - \mathcal{O}^{\bar{1}}_\alpha\,^* \right )^2\,.
\ee
Substituting this into \eqref{eq:appaid2} reproduces the identity \eqref{eq:Znetaoverlapinftemp}.

From these expressions, it follows directly that Lemma~\ref{lemma:densitymatrix} holds.
\begin{proof}[Proof of Lemma~\ref{lemma:densitymatrix}]
Taking the Choi state version of \eqref{eq:appAUotimesU*}, we find
\be
\overline{\ket{U(t)}\bra{U(t)}} = \big(e^{-t H_{\eff}^{a\bar{b}}}\big)^{T_\perp} = \rho_t\,.
\ee
The ensemble average acts as a CPTP map, ensuring that $\rho_t$ is a normalized density matrix.
\end{proof}

\subsection*{Finite temperature}

We now turn to the more general case where a time-independent, non-random term is added to the Brownian Hamiltonian \eqref{eq:BrownianH}. The time-dependent Hamiltonian then takes the form
\be\label{eq:HdetH0}
H(t) = H_0 + \sum_{\alpha = 1}^{K} g_\alpha(t) \mathcal{O}_\alpha = \sum_{\alpha = 1}^{K+1} g_\alpha(t) \mathcal{O}_\alpha\,,
\ee
where $H_0$ is a generic Hermitian operator. In the last expression we have set $g_{K+1}(t) \equiv 1$ for all times and $\mathcal{O}_{K+1} \equiv H_0$. As before, we expand the time-evolution operator using \eqref{eq:HdetH0}:
\begin{align}
    \notag U(t) &= \sum_{n=0}^\infty \frac{(-\iw)^n}{n!} \mathsf{T} \left \lbrace \prod_{i=1}^n \int_0^t \text{d}t_i \sum_{\alpha_i = 1}^{K+1} g_{\alpha_i}(t_i) \mathcal{O}_{\alpha_i} \right \rbrace \\
         &= \sum_{n=0}^\infty (-\iw)^n \int_0^t \text{d}t_1 \int_0^{t_1} \text{d}t_2 \ldots \int_0^{t_{n-1}} \text{d}t_n  \sum_{\alpha_1 \ldots \alpha_n = 1}^{K+1} g_{\alpha_1}(t_1) \ldots g_{\alpha_n}(t_n) \mathcal{O}_{\alpha_1} \ldots \mathcal{O}_{\alpha_n}\,.
\end{align}
Now, in each term of the sum over $n$, there can be $p = 0, \ldots ,n$ insertions of $H_0$. There is a total of $n-p$ random couplings in $g_{\alpha_1}(t_1) \ldots g_{\alpha_n}(t_n)$. From the Gaussianity of the random couplings, only terms with even $n-p$ survive the averaging. This requires us to treat separately the cases with $n$ even ($p$ even) and $n$ odd ($p$ odd). Furthermore, when performing the time integrals, the only nonvanishing contributions are those in which all of the $\delta$-functions appear in the form $\delta (t_i - t_{i+1})$, exactly as in the $H_0 = 0$ case. These contractions generate quadratic factors in $\mathcal{O}_{\alpha_i}$. 

It is also straightforward to see that \eqref{deltas} generalizes to the case with $\frac{n-p}{2}$ $\delta$-functions as
\be
\int_0^t \text{d}t_1 \ldots \int_0^{t_{n-1}} \text{d}t_n \; \delta (t_i - t_{i+1}) \; \ldots \delta (t_j - t_{j+1}) = \frac{t^{\frac{n+p}{2}}}{2^{\frac{n-p}{2}} \left (\frac{n+p}{2} \right )!}\,,
\ee
where we note that there are now $n$ integrals, since $n$ may be odd when $p$ is odd. This result is independent of which $\frac{n-p}{2}$ time variables $t_i, \ldots, t_j$ appear in the arguments of the $\delta$’s.

 It is therefore convenient to introduce the generalization of the binomial theorem to non-commuting operators:
\be
\left (A+B \right )^n = \sum_{p = 0}^n \sum_{\text{perm}} A^p B^{n-p} \, ,
\ee
where $\sum_{\text{perm}} A^p B^{n-p}$ denotes the $\binom{n}{p}$ distinct permutations of terms containing $p$ factors of $A$ and $n-p$ factors of $B$.

With this in mind, we can write
\begin{align}
\notag\overline{U(t)} &= \sum_{n_1=0}^\infty (-1)^{n_1} \sum_{p_1 = 0}^{n_1}    \frac{t^{n_1+p_1}}{(n_1+p_1)!}  \sum_{\text{perm}} H_0^{2p_1} \left (H_1 \right )^{n_1-p_1} \\
&-\iw \sum_{n_2=0}^\infty (-1)^{n_2} \sum_{p_2 = 0}^{n_2} \frac{t^{n_2+p_2+1}}{(n_2+p_2+1)!}  \sum_{\text{perm}} H_0^{2p_2+1} \left (H_1 \right )^{n_2-p_2}\,,
\end{align}
where we have separated the contributions corresponding to even $n = 2n_1$ ($p = 2p_1$) and odd $n = 2n_2+1$ ($p = 2p_2+1$), and moreover $H_1 = \tfrac{J}{2} \sum_{\alpha = 1}^{K} \mathcal{O}^2_\alpha$. The first sum contains terms with an even number of $H_0$ insertions, while the second one contains those with an odd number.

Making the change of variables $n_1 + p_1 = m_1$, $n_1-p_1 = l_1$, $n_2+p_2+1 = m_2$, and $n_2-p_2 = l_2$, we obtain
\begin{align}
\notag \overline{U(t)} &= \sum_{m_1=0}^\infty \frac{(-t)^{m_1}}{m_1!} \sum_{(m_1-l_1) \;\text{even}}^{m_1} \; \sum_{\text{perm}} (\iw H_0)^{m_1-l_1} \left (H_1 \right )^{l_1} \\
\notag &+ \sum_{m_2=1}^\infty \frac{(-t)^{m_2}}{m_2!} \sum_{(m_2-l_2) \;\text{odd}}^{m_2} \; \sum_{\text{perm}} (\iw H_0)^{m_2-l_2} \left (H_1\right )^{l_2} \\
&= \sum_{n=0}^\infty \frac{(-t)^{n}}{n!} \left (\iw H_0 + H_1 \right )^n = e^{-tH_{\eff}}\,,
\end{align}
for the single-replica effective Hamiltonian
\be
H_{\eff} = \iw H_0 + \frac{J}{2} \sum_{\alpha = 1}^{K} \mathcal{O}^2_\alpha\,.
\ee
We emphasize that the same notation is used for the effective Hamiltonian as in the previous case; the distinction should be clear from context.

The analogous identity for $U(t) \otimes U(t)^*$ follows from the identification obtained analogously to the case $H_0 = 0$. We start from
\begin{equation*}
U(t)^1 \otimes U(t)^{\bar{1}}\,^* = \mathsf{T} \left\lbrace  e^{-\iw \int_0^t \text{d}t’ H(t’)} \right\rbrace \,, \qquad H(t) = H_0^{1\bar{1}} + \sum_{\alpha = 1}^{K} g_\alpha (t)\, \mathcal{O}^{1\bar{1}}_\alpha\,,
\end{equation*}
where $H_0^{1\bar{1}} =  H_0^{1} - H_0^{\bar{1}}\,^{*}$. Thus we arrive at
\be
\overline{U(t)^1 \otimes U(t)^{\bar{1}}}\,^*  = e^{-t H_{\eff}^{1\bar{1}}}\,,
\ee
for the two-replica effective Hamiltonian
\be
    H_{\eff}^{1\bar{1}} = \iw  (H_0^{1} - H_0^{\bar{1}}\,^*) + \frac{J}{2} \sum_{\alpha = 1}^{K} \left ( \mathcal{O}^{1}_\alpha - \mathcal{O}_\alpha^{\bar{1}}\,^* \right )^2\,.
\ee

When we perform gradual cooling for $W(t)$ given by \eqref{eq:coolrcircuit}, the system evolves in real time under the Hamiltonian \eqref{eq:coolrhamilt}, which has the same form as \eqref{eq:HdetH0} with the substitution $H_0 \rightarrow -\iw  H_0$. Making the corresponding substitution in the formulas above, we obtain the finite-temperature single-replica and two-replica effective Hamiltonians for the moments:
\begin{gather}
     \overline{W(t)} = e^{-tH_{\eff}}\,,\qquad  H_{\eff} = H_0  + \frac{J}{2} \sum_{\alpha = 1}^{K} \mathcal{O}^2_\alpha\,,\\[.2cm]
      \overline{W(t)^1\otimes W(t)^{\bar{1}}\,^*}  = e^{-tH^{1\bar{1}}_{\eff}}\,,\qquad H_{\eff}^{1\bar{1}} = H_0^{1} + H_0^{\bar{1}}\,^* + \frac{J}{2} \sum_{\alpha = 1}^{K} \left ( \mathcal{O}^{1}_\alpha - \mathcal{O}_\alpha^{\bar{1}}\,^* \right )^2\,.
\end{gather}
Once more, the same notation is used for the effective Hamiltonians as for the infinite-temperature ones; the intended meaning should be evident from the context. Taking the average of \eqref{eq:numfintemppreavg}, \eqref{eq:normfintemppreavg} or considering \eqref{eq:znfintemp}, and substituting these expressions, we arrive at \eqref{eq:avoverlapfintemp} and \eqref{eq:Znetaoverlapfintemp}.

\section{Statistics of the rank of the Gram Matrix: numerics}\label{app:rank}

In this appendix we numerically analyze the statistics of the rank of the Gram matrix $G_{ij} = \braket{\Psi_{i}}{\Psi_{j}}$  with respect to the number of states $\Omega$ for different ensembles. Equivalently, we seek to understand the statistics of the dimension of the Hilbert space spanned by the $\Omega$ states in the ensemble, as discussed in section \ref{Sec:countingframeworkdim}. The procedure to calculate the statistics consist on drawing a set of $\Omega$ vectors from an ensemble, construct the Gram Matrix corresponding to this set, and then calculate it rank. To take the average, we repeat this procedure several times for different realizations. Below we confirm the general discussion given in section \ref{Sec:countingframeworkdim}.

\subsubsection*{Haar Example}
We start from the ensemble of Haar random states. More precisely, we take a Hilbert space $\mathcal{H}$ of dimension $d = \text{dim}(\mathcal{H}) = 100$. From it we take a reference vector $\ket{v_{0}} \in \mathcal{H}$ sampling $d$ random real numbers. The set of vectors is constructed by drawing $\Omega$ unitary matrices $U_{i}$ from the Haar distribution, and then acting over the references vector $\ket{v_{0}}$ to obtain
\be
    \mathsf{F}_\Omega^{\text{Haar}} = \{ \ket{\Psi_{i}} = U_{i} \ket{v_{0}} : \ i = 1, \dots, \Omega \} \ .
\ee
With this ensemble it is straightforward to calculate the Gram matrix and its rank. We repeat this process several times to calculate the mean $\overline{\text{rank}(G)}$ and the variance $\sigma_{G} = \overline{\text{rank}(G)^2}-\overline{\text{rank}(G)}^2$ for a given value of $\Omega$.
In figure \ref{fig:Rank2Local} we plot both quantities.

\begin{figure}[H]
    \centering
    \includegraphics[width=0.6\linewidth]{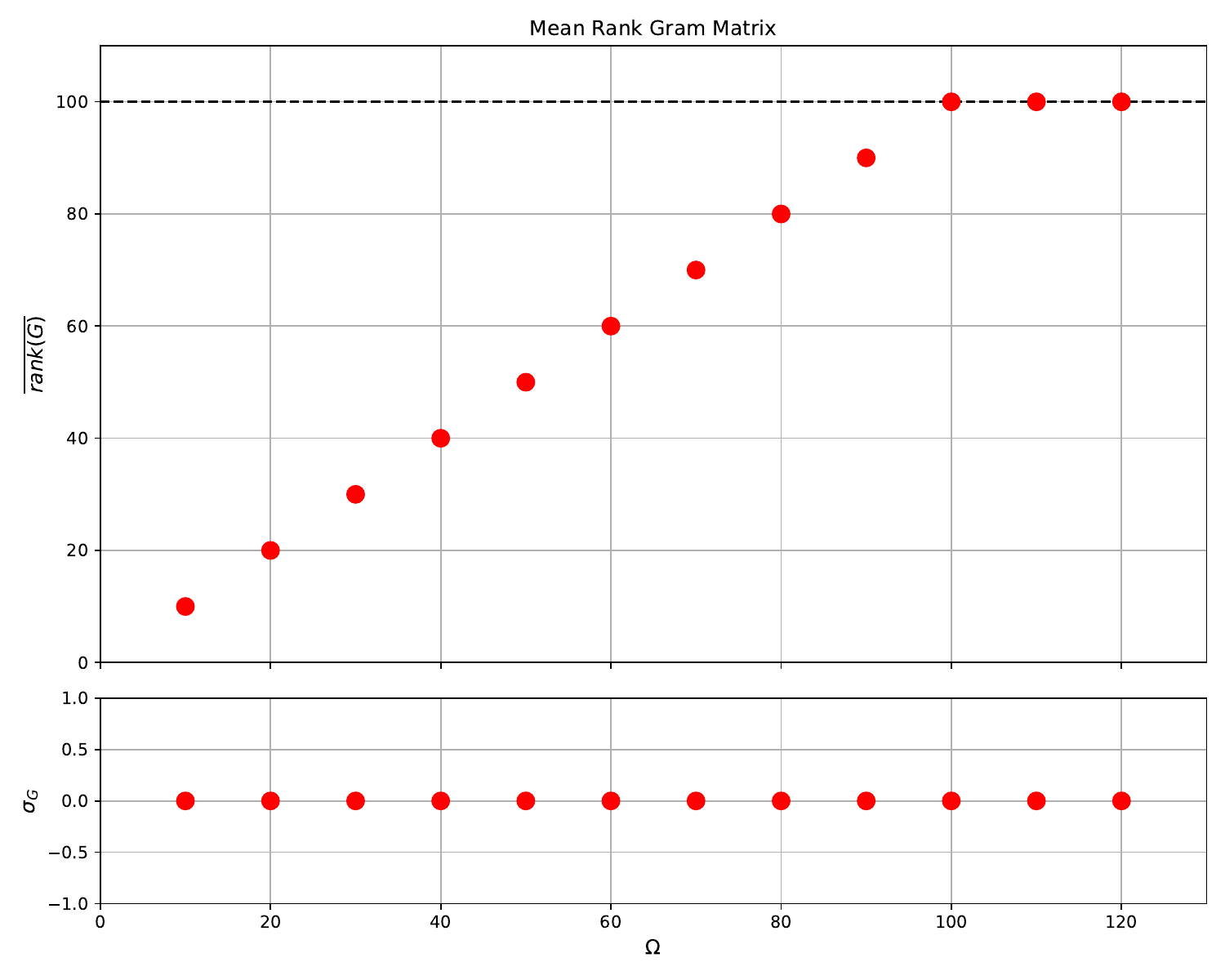}
    \caption{
        Plot of the mean rank $\overline{\rank(G)}$ and standard deviation $\sigma_{G}$ of the Gram Matrix as a function of $\Omega$.
        The black dashed line correspond to the full dimension of the spaces.
    }
    \label{fig:Rank2Local}
\end{figure}

This numerical results can be summarized by
\be
    d_{\Omega} = \overline{\rank(G)} = \min \{\Omega, d \} \ , \quad
    \sigma_{G}^{2} = \overline{\rank(G)^{2}} - \overline{\rank(G)}^{2} = 0 \ .
\ee
The first part is well known by now. More interesting is the behavior of the variance, which is exactly zero, even at finite $\Omega$ and $d$. This of course follows from the general observation described in section \ref{Sec:countingframework}, i.e. the probability of having linear dependence has measure zero $\Omega<d$.

\subsubsection*{Brownian spin cluster example}

We now analyze the statistics of the rank of the Gram matrix for the Brownian spin cluster worked out in section \ref{S:brownianspincluster}. Using the Hamiltonian defined in \eqref{eq:hamiltonianspincluster}, we construct the time evolved states at time $t = M \delta t$ as
\be
    \ket{\Psi_{i}(t)} = \prod_{l=1}^{M} (e^{-iH_{l} \delta t} \otimes I) \ket{\mathbf{1}} \,,
\ee
and the ensemble at time $t$ as
\be
    \mathsf{F}_\Omega^{t} = \{ \ket{\Psi_{i}(t)} | \ i = 1, \dots, \Omega \} \ .
\ee
With this ensemble it is straightforward to calculate the Gram matrix and its rank. We repeat this process several times to calculate $\overline{\text{rank}(G)}$ and the variance $\overline{\text{rank}(G)^2}-\overline{\text{rank}(G)}^2$ for a given value of $\Omega$.
For the calculations we take $q=2$, $N=4$, $\delta t =1$ and $J = \frac{1}{10 K}$ which gives a $t_{\text{gap}} \propto 1 / E_{\text{gap}} = 16.75$.
In figure \ref{fig:RankGramqLocalAll} we plot both quantities as a function of $\Omega$ for different values of the time $t$.

\begin{figure}[H]
    \centering
    \includegraphics[width=0.6\linewidth]{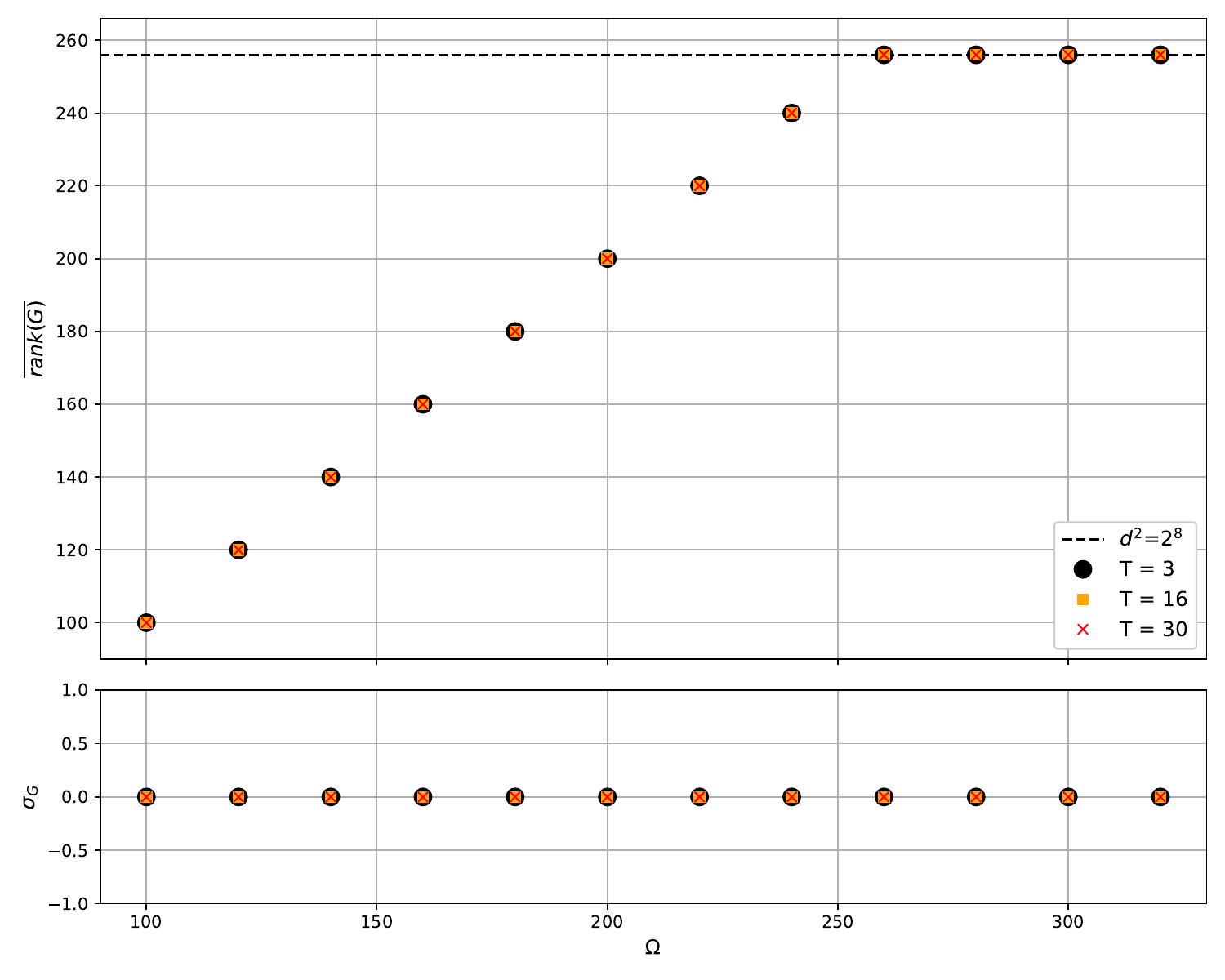}
    \caption{
        Plot of $\overline{\rank(G)}$ and its standard deviation $\sigma_{G}$ of the Gram Matrix as a function of $\Omega$ at different times $t$. The black dashed line is the dimension of the Hilbert space.
    }
    \label{fig:RankGramqLocalAll}
\end{figure}

Again, these numerical results verify that for every $t$ and $\Omega$ we have
\be
    d_{\Omega} = \overline{\rank(G)} = \min \{ \Omega, d \} \ , \quad
    \sigma_{G}^{2} = \overline{\rank(G)^{2}} - \overline{\rank(G)}^{2} = 0 \,,
\ee
showing that the variance is zero for the state ensembles induced by the Brownian circuits, even for small $\Omega$ and $d$, in accordance with the discussion in section \ref{Sec:countingframework}.

\section{Resolvent computations in Brownian circuits}\label{app:resol}

In this Appendix we recover the finite $\Omega$ formula \eqref{eq:transitiondimension} from resolvent type computations, expanding on the methods of \cite{Penington:2019kki}. The objective is to test the validity of the planar approximation and the assumed vanishing of the average inner product between different microstates.

\subsection{Resolvent counting methods}

We first review the standard approach. Consider the resolvent of the Gram matrix $G_{ij}=\langle \Psi_i\vert\Psi_j\rangle$ of overlaps, where $\vert\Psi_j\rangle$ are a set of $\Omega$ states. This is defined as
\be\label{eq:resolapp}
R_{ij}(\lambda) \equiv \left( \frac{1}{\lambda \mathbf{1} - G} \right)_{\! ij}=\frac{1}{\lambda}\delta_{ij}+\sum\limits_{n=1}^{\infty}\,\frac{1}{\lambda^{n+1}}(G^n)_{ij}\,.
\ee
where we have expanded the resolvent around large $\lambda$. To find the dimension spanned by the states we can proceed in two analog manners. Denote by $R(\lambda) = \sum_{i=1}^\Omega R_{ii}(\lambda)$ the trace of the resolvent and by $\lbrace \lambda_i: i=1,...,\Omega\rbrace $ the spectrum of eigenvalues of $G$. There is a well-known mathematical relation between the density of eigenvalues $D(\lambda) = \sum_{i=1}^\Omega \delta(\lambda-\lambda_i)$ of a matrix such as $G$, and the discontinuity of $R(\lambda)$ along the imaginary axis, 
\be\label{eq:densitydisc}
D(\lambda)=\lim\limits_{\epsilon \rightarrow 0}\frac{1}{2\pi \iw}\left(R(\lambda-\iw\epsilon)-R(\lambda+\iw\epsilon)\right)\,.
\ee
So if one finds $R(\lambda)$, then the dimension follows directly from the integral of the density of eigenvalues over the positive real axis
\be\label{eq:dimspan}
d_\Omega = \int_{0^+}^\infty\text{d}\lambda\, D(\lambda) \,.
\ee 
Alternatively, we can use the fact that the trace of the resolvent of the Gram matrix can be written as
\be \label{eq:kerg}
R(\lambda)=\sum\limits_{i=1}^{\Omega}\,\frac{1}{\lambda  -\lambda_i}=\frac{\text{Ker}(G)}{\lambda}+\sum\limits_{\lambda_i>0}\,\frac{1}{\lambda  -\lambda_i}\,,
\ee
where we are using that $G$ is Hermitian and positive semi-definite. Then the trace of the resolvent has a simple pole at $\lambda=0$ with residue equal to $\text{Ker}(G)$. The dimension follows directly from this, together with \eqref{eq:Gramdimensionrankker},
\be \label{eq:domega}
d_\Omega = \Omega - \text{Res}_{\lambda=0} \,R(\lambda)\,.
\ee
The goal now is to use these methods to compute the average dimension $\overline{d_\Omega}$ spanned by a randomly chosen $\Omega$ states, where the overline denotes averaging. Since the methods above work for individual instances of $G$, we just need to use the average resolvent $\overline{R(\lambda)}$. The difference is now that $\overline{R(\lambda)}$ typically has branch cuts instead of simple poles, but both \eqref{eq:densitydisc} and \eqref{eq:domega} are still valid for the average. More concretely, notice that the average over the second term $\sum_{\lambda_i>0}\,\frac{1}{\lambda  -\lambda_i}$ cannot contribute to the pole at $\lambda=0$ to the averaged resolvent since no instance in the ensemble shows that pole. The average of the first term is just the average of $\text{Ker}(G)$ divided by $\lambda$. Then the residue at $\lambda =0$ of the averaged resolvent provides the desired result.

So we seek the averaged resolvent $\overline{R(\lambda)}$. The first simplification occurs in the large-$\Omega$ expansion, with $ d\Omega^{-1}$ held fixed. Assuming that
\be \label{asszn}
\overline{\bra{\Psi_{i_1}}\ket{\Psi_{i_2}} \bra{\Psi_{i_2}}\ket{\Psi_{i_3}}... \bra{\Psi_{i_n}}\ket{\Psi_{i_{n+1}}}}= \delta_{i_{1},i_{n+1}}\,\overline{\bra{\Psi_{i_1}}\ket{\Psi_{i_2}} \bra{\Psi_{i_2}}\ket{\Psi_{i_3}}... \bra{\Psi_{i_n}}\ket{\Psi_{i_{1}}}}\equiv \delta_{i_{1},i_{n+1}} \,Z_n\;,
\ee
for $i_i\neq i_2\neq \cdots\neq i_n$, which in particular says
\be
\overline{\bra{\Psi_{i_1}}\ket{\Psi_{i_2}} }=\delta_{i_{1},i_{2}}\,,
\ee
it was shown in \cite{Penington:2019kki}, following earlier \cite{speicher2009freeprobabilitytheory,Cvitanovic:1980jz}, that the average resolvent satisfies the following Schwinger-Dyson (SD) equation 
\be\label{eq:SchwingerDyson}
    \overline{R(\lambda)} = \frac{\Omega}{\lambda} + \frac{1}{\lambda} \sum_{n=1}^{\infty} Z_n \overline{R(\lambda)}^{\,n} \ .
\ee
The SD equation \eqref{eq:SchwingerDyson} is derived in a ``planar approximation'' that only considers diagrams whose index contractions can be drawn on the plane. It relies on the presence of two large parameters in the calculation: the number of states $\Omega$ and the total dimension $d$. Repeated indices include suppression factors of $\Omega^{-1}$, since they effectively lower the number of traces, while at the same time they are enhanced by factors of $d$, since they increase the number of loops. The relevant parameter is $d\,\Omega^{-1}$. It is simple to check in explicit examples that planar diagrams have the dominant scalings in the limit $\Omega,d\rightarrow\infty$ with $d\,\Omega^{-1}=\text{constant}$. 

Note the expansion of the resolvent in terms of the moments of the Gram matrix \eqref{eq:resolapp} holds in the case where $\lambda$ is bigger than all the eigenvalues of $G$. But expression \eqref{eq:domega} instructs us to inspect $\overline{R(\lambda)}$ in the limit $\lambda \rightarrow 0$. Indeed, stepping on the SD equation above, we formally have
\be
d_\Omega = - \lim\limits_{\lambda \to 0}\, \sum_{n=1}^{\infty} Z_{n} \overline{R(\lambda)}^{n}\,.
    \label{eq:domega2}
\ee
We can bypass this problem by further analyzing the SD equation for the resolvent in the case that the states $\vert \psi_i\rangle$ are drawn from some ensemble. In this scenario, we define the average state
\be 
\rho\equiv  \overline{\ket{\Psi}\bra{\Psi}}\,.
\ee
This state is, by construction, a normalized density matrix. We denote its $k$ different non-zero eigenvalues by $\lbrace p_\alpha: \alpha = 1,...,k\rbrace$, and their multiplicities by $\lbrace \omega_\alpha: \alpha = 1,...,k\rbrace$, so that $\sum_{\alpha=1}^{k}\, \omega_{\alpha} p_{\alpha} = 1$. 
We can write the coefficients of the SD equation \eqref{eq:SchwingerDyson} in terms of the R\'enyi entropies of this density matrix
\be 
Z_n=\textrm{Tr}\,\rho^n=\sum_{\alpha=1}^{k} \omega_{\alpha} p_\alpha^n\,.
\ee
The SD equation \eqref{eq:SchwingerDyson} then can be formally resummed
\be\label{eq:resummedSD}
    \lambda \overline{R(\lambda)}  = {\Omega} +  \sum_{\alpha=1}^{k} \sum_{n=1}^{\infty} \omega_{\alpha} \left( p_\alpha \overline{R(\lambda)}\right)^{n} = {\Omega} +  \sum_{\alpha=1}^{k} \omega_{\alpha} \frac{p_\alpha\,\overline{R(\lambda)}}{1-p_\alpha \,\overline{R(\lambda)}}\,.
\ee
In this form, we can analyze the limit $\lambda \rightarrow 0$. First, if $\overline{R}(\lambda)$ goes to infinity as $\lambda$ goes to zero, we can combine this equation with \eqref{eq:kerg} to obtain
\be 
d_\Omega = \sum_{\alpha=1}^{k} \,\omega_{\alpha} = \lim_{n \to 0} Z_n\;.
\label{Z0}
\ee
The limit defines $\text{rank}(\rho)$ by analytic continuation of the R\'enyi entropies. But \eqref{eq:kerg} also shows that $R(\lambda)$ diverges as $\lambda \rightarrow 0$ only if $\text{Ker}(G) >0$. Therefore, again using \eqref{eq:kerg} we have
\be 
\text{Ker}(G) = \Omega - d_\Omega > 0 \Rightarrow \Omega > d_\Omega.
\ee
Equivalently
\be  
d_\Omega = \lim\limits_{n \to 0} Z_n \,\,\,\, \text{only if} \,\,\,\,\Omega > \lim\limits_{n \to 0} Z_n\;.
\ee
On the other hand, if $R(\lambda)$ is continuous in $\lambda = 0$, eq.\eqref{eq:kerg} shows that $d_\Omega = \Omega$. We finally conclude
\be
d_\Omega = \text{min} \left \{\Omega, \lim_{n \to 0} Z_n \right \},
\ee
where we notice that if $\Omega = \Omega_0 > \lim\limits_{n \to 0} Z_n$ and $d_\Omega = \Omega_0$, by increasing $\Omega > \Omega_0$ (by adding more states to the set of generators), eventually $R(\lambda)$ will diverge and we shall have $d_\Omega = \lim\limits_{n \to 0} Z_n < \Omega_0$. But this is not possible since it would mean that adding more states to a set of generators would end up generating a smaller subspace.

For completeness, we note that another formula for the dimension follows by defining the function
\be
    f(z) = \text{Analytic extension of} \  \sum_{n=1}^{\infty} Z_{n} z^{n}\,,
\ee
whose Taylor expansion around $z=0$ coincides with the series $\sum_{n=1}^{\infty} Z_{n} z^{n}$. Using the previous form for $Z_n$ we have
\be\label{secondd}
     d_{\Omega} = - \lim_{x \rightarrow \infty} f(x) \ .
\ee
We now expand on the application of this method to the case of interest in this paper, namely the case the states are drawn from Brownian circuits.

\subsection{Planar resolvent vs exact numerics}
\label{sS:planarnumerics}

We recall that the previous analysis applies when \eqref{asszn} holds. In particular, it holds when
\be\label{onepo}
\overline{\bra{\Psi_{i_1}}\ket{\Psi_{i_2}} }=\delta_{i_1,i_2}\,.
\ee
This is a problem for specific constructions in quantum gravity. For example, for the shells used in \cite{Balasubramanian:2022gmo,Balasubramanian:2022lnw}, although small, inner products between shell states of different masses are non-zero since there is a small amplitude for annihilation and creation of dust particles due to gravitational effects. This forces one to go to the limit in which all masses and all mass differences diverge.

In the case of interest of this paper, i.e. states drawn from Brownian circuits, we face a similar situation. The advantage of the present framework is that such non-vanishing ``one-point functions'' $\overline{\bra{\Psi_{i_1}}\ket{\Psi_{i_2}} }$ can be explicitly computed using the methods developed above. Indeed, from the results of Appendix \ref{app:Heff}, we know that such amplitudes decay exponentially fast with the circuit time. Then we should be able to verify numerically the validity of the planar approximation and the SD equation, at least after some time scale. In this appendix, using specific examples, we compare the exact averaged resolvent obtained numerically using \eqref{eq:kerg}, with the resolvent obtained by the SD equation in the form \eqref{eq:resummedSD}.

To compute $\overline{R(\lambda)}$ numerically we use \eqref{eq:kerg}. We first perform multiple iterations in which we sample $\Omega$ vectors from the ensemble and construct the corresponding Gram matrix. For each sample, we extract the eigenvalues of the Gram matrix, use them to build the resolvent via \eqref{eq:kerg}, and evaluate it for several values of the parameter $\lambda$. Since the resolvent has poles along the real axis, we shift $\lambda \rightarrow \lambda + i\epsilon$ with small real constant $\epsilon$ to avoid divergences. Finally, we average the resulting values of $R(\lambda)$ across all iterations to obtain $\overline{R(\lambda)}$.

To compute the numerical planar resolvent $\overline{R(\lambda)}$ we use \eqref{eq:resummedSD}. We first construct the density matrix $\rho$ by sampling $M$ vectors $\ket{\Psi_i}$ from the ensemble and computing the average: $\overline{\ketbra{\Psi}} = \frac{1}{M} \sum_{i=1}^{M} \ketbra{\Psi_i} = \rho$. We then compute the eigenvalues $p_i$ of $\rho$ and solve \eqref{eq:resummedSD} for $\overline{R(\lambda)}$ at a given value of $\lambda$. Since \eqref{eq:resummedSD} is a polynomial equation in $\overline{R(\lambda)}$, its solution is generally not unique. We therefore expect that the resolvent obtained from the Gram matrix corresponds to one of the branches of solutions derived from this method.

\subsubsection*{Haar example}

We start with Haar random states, namely states constructed from unitaries drawn with the Haar measure. More precisely, we take a Hilbert space $\mathcal{H}$ of dimension $d = \text{dim}(\mathcal{H}) = 100$, and fix a reference vector $\ket{v_{0}} \in \mathcal{H}$ by sampling $d$ random real numbers. To compute $\overline{R}(\lambda)$ using \eqref{eq:kerg}, we draw $\Omega$ unitary matrices $U_{i}$ with a Haar distribution. Then we let those unitaries act over the reference vector $\ket{v_{0}}$. This leads to the ensemble of states
\be
    \mathsf{F}_{\Omega}^{\text{Haar}} = \{ \ket{\Psi_{i}} = U_{i} \ket{v_{0}} | \ i = 1, \dots, \Omega \} \,.
\ee
With this ensemble we can construct a Gram matrix and calculate the resolvent $R(\lambda)$ as was explained before. We repeat this process $N = 1000$ to obtain $\overline{R(\lambda)}$.

To compute $\overline{R}(\lambda)$ using \eqref{eq:resummedSD} it is well known that for Haar random unitaries $U$ it follows
\be
    \overline{\ketbra{\Psi}}
    = \frac{1}{d} \Tr \left( \ketbra{v_{0}} \right) 
    = \frac{1}{d} \,.
\ee
Then the eigenvalues are trivially $p_{i} = 1/d$, and the mean resolvent takes the form
\be
    \lambda \overline{R(\lambda)} = \Omega + d \frac{\overline{R(\lambda)}}{d - \overline{R(\lambda)}} \,,
\ee
which is a quadratic equation in $\overline{R(\lambda)}$ and therefore has an explicit solution.

In figure \eqref{fig:ResolventHaarN} we plot the results obtained from the two methods, showing a fair match already for relatively small $\Omega$ and $d$.

\begin{figure}[h]
    \centering
    \includegraphics[width=0.8\linewidth]{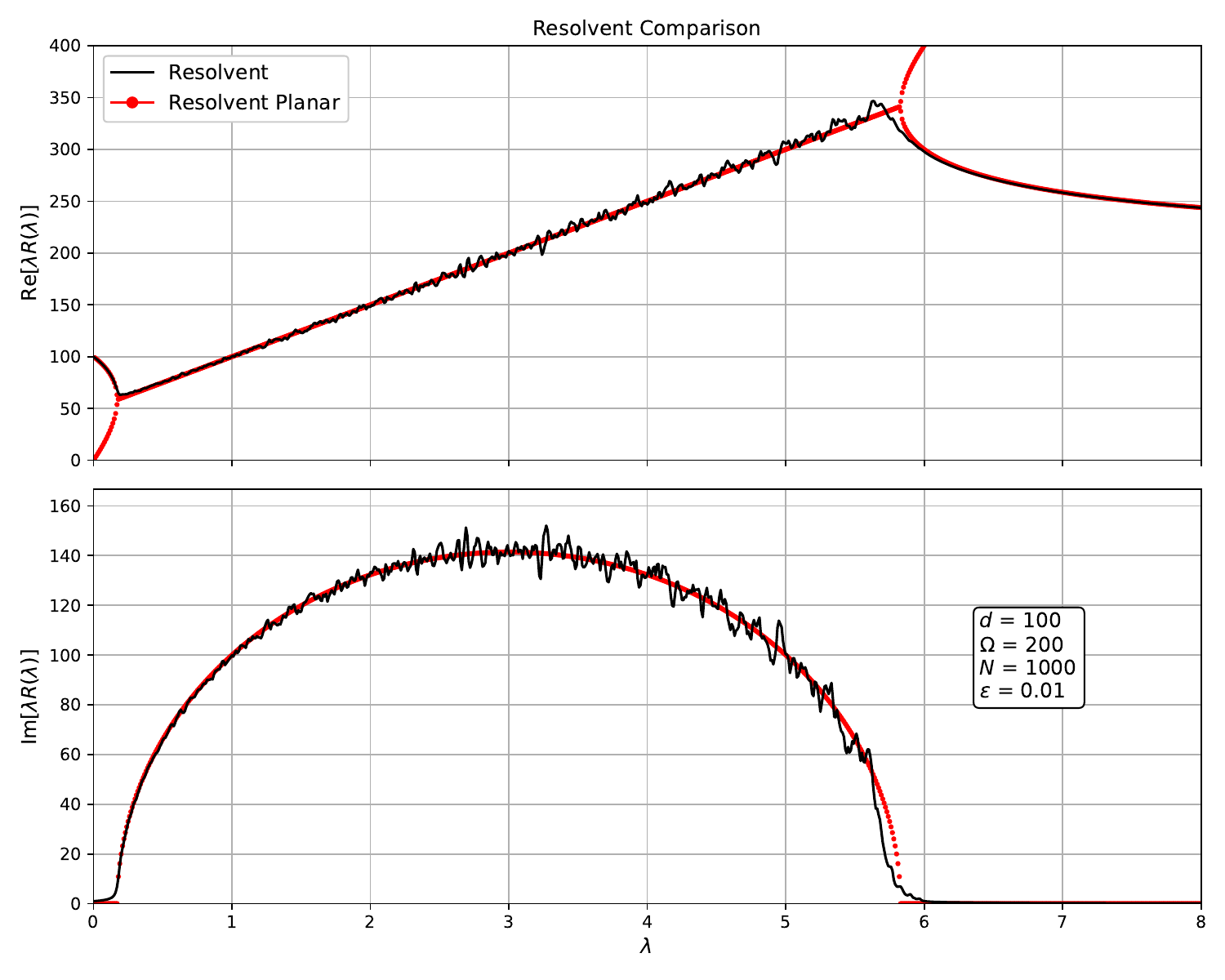}
    \caption{
    Plot of the real and imaginary part of the mean resolvent $\lambda \overline{R}(\lambda)$ calculated from the two methods described in the text for the Haar ensemble. The black solid line corresponds to the resolvent calculated with the Gram matrix \eqref{eq:kerg} and the red point curve with \eqref{eq:resummedSD}, which only takes into account the planar diagrams.
    The parameters of the calculation are $d = 100$, $\Omega = 200$, $\epsilon = 0.01$, with averages taken over 1000 draws.
    }
    \label{fig:ResolventHaarN}
\end{figure} 

\subsubsection*{Brownian GUE Evolution example}

We now analyze the time-dependent Brownian GUE. We seek to understand whether the planar approximation is affected at sufficiently small times.

We start from a Hilbert space $\mathcal{H}$ of dimension $d = \text{dim}(\mathcal{H}) = 100$. From it we choose a reference vector $\ket{v_{0}} \in \mathcal{H}$ sampling $d$ random real numbers. The state ensemble $\mathsf{F}_{\Omega}^{t}$ at a time $t$ is
\be\label{eq:brownianGUEapp}
    \mathsf{F}_{\Omega}^{t} = \{ \ket{\Psi_{i}(t)} | \ i = 1, \dots, \Omega \} \ .
\ee
The states $\ket{\Psi_{j}(t)}$ are constructed as
\be
    \ket{\Psi_{j}(t)} = \prod_{l=1}^{M} e^{-iH_{l} \delta t} \ket{v_{0}} \ ,
\ee
where $H_{l}$ is an Hermitian matrix taken from the Gaussian unitary ensemble (GUE) with zero mean and standard deviation equal to one and $t = M \delta t$.

Following the same procedure explained before, we calculate the mean resolvent using the two different methods. We plot these computations for different times in Figs. \eqref{fig:ResolventGUET5}, \eqref{fig:ResolventGUET1} and \eqref{fig:ResolventGUET01}. For the computations we take $\Omega = 20$, perform the average over 1000 draws, and analyze three different times $t=5$, $t=1$ and $t=0.1$.
As might have been anticipated, we observe a good matching already for very low values of $\Omega$ and $d$. But, maybe unexpectedly given the comments above, we find a fairly good matching even at small times. We expand on this in the next section.

Summarizing, the numerical analysis let us conclude that the planar approximation works well for Brownian circuits and that $\mathsf{F}_{\Omega}^{t}$ spans the full Hilbert spaces at every time $t$, so that
\begin{equation}
    d_{\Omega} = \Omega - \text{Res}_{\lambda = 0} R(\lambda) = \Omega - \lim_{\lambda \rightarrow 0} \lambda R(\lambda) = d \ .
\end{equation}

\subsection{Generalized Schwinger-Dyson equation}

The numerical analysis of the previous section led us to the conclusion that the planar approximation works well at all times. This seems at odds with the fact that for random circuits 
\be
A_n \equiv\overline{\braket{\Psi_{i}}{\Psi_{i_{1}}} \braket{\Psi_{i_{1}}}{\Psi_{i_{2}}} \cdots \braket{\Psi_{i_{n-1}}}{\Psi_{j}}}\;,
\ee
does not vanish for $i\neq i_l\neq j$. Let us provide an explanation for this observation.

\begin{figure}[H]
    \centering
    \includegraphics[width=0.6\linewidth]{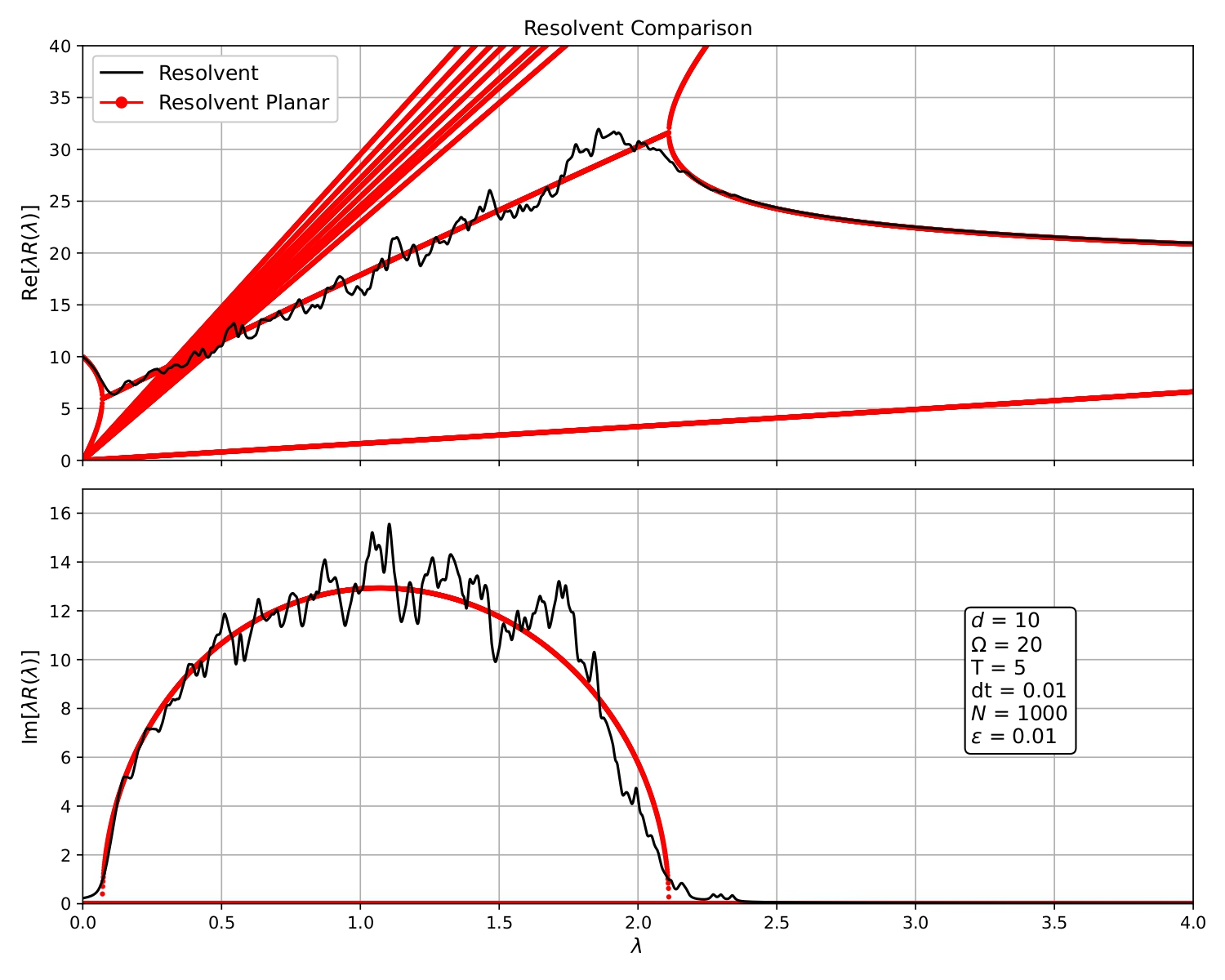}
    \caption{
    Plot of real and imaginary part of the mean resolvent $\lambda \overline{R}(\lambda)$, calculated from different methods for the Brownian GUE.
    The black solid line correspond to the resolvent calculated with the Gram matrix \eqref{eq:kerg} and the red point curve with \eqref{eq:resummedSD}, which only takes in account the planar diagrams.
    The parameters used in the calculation are $T=5$, $dt=0.01$, $d = 10$, $\Omega = 20$, $\epsilon = 0.01$, with averages taken over 1000 draws.
    }
    \label{fig:ResolventGUET5}
\end{figure}

\begin{figure}[H]
    \centering
    \includegraphics[width=0.6\linewidth]{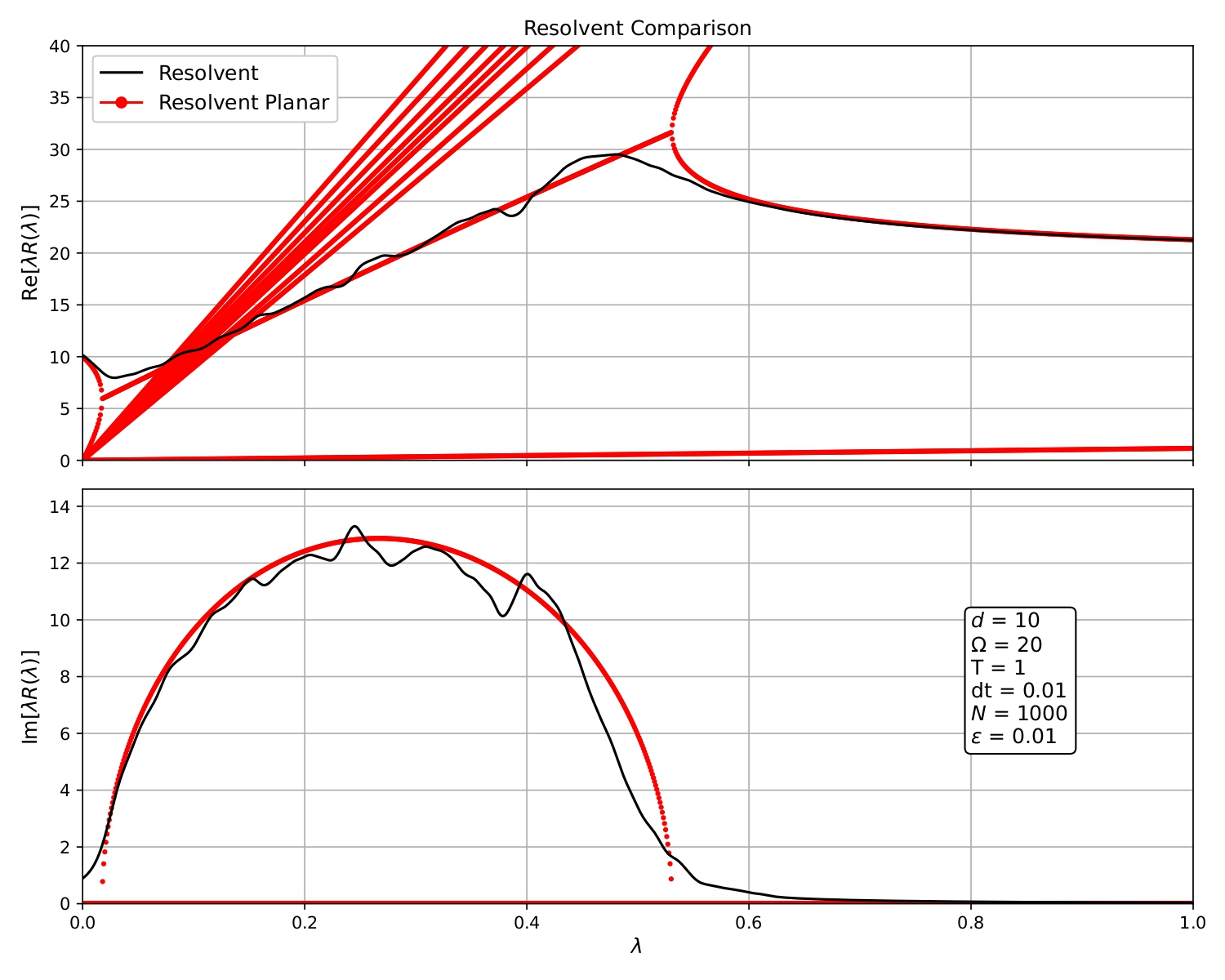}
    \caption{
    Plot of real and imaginary part of the mean resolvent $\lambda \overline{R}(\lambda)$, calculated from different methods for the Brownian GUE.
    The black solid line correspond to the resolvent calculated with the Gram matrix \eqref{eq:kerg} and the red point curve with \eqref{eq:resummedSD}, which only takes in account the planar diagrams.
    The parameters used in the calculation are $T=1$, $dt=0.01$, $d = 10$, $\Omega = 20$, $\epsilon = 0.01$, with averages taken over 1000 draws.
    }
    \label{fig:ResolventGUET1}
\end{figure}

\begin{figure}[H]
    \centering
    \includegraphics[width=0.6\linewidth]{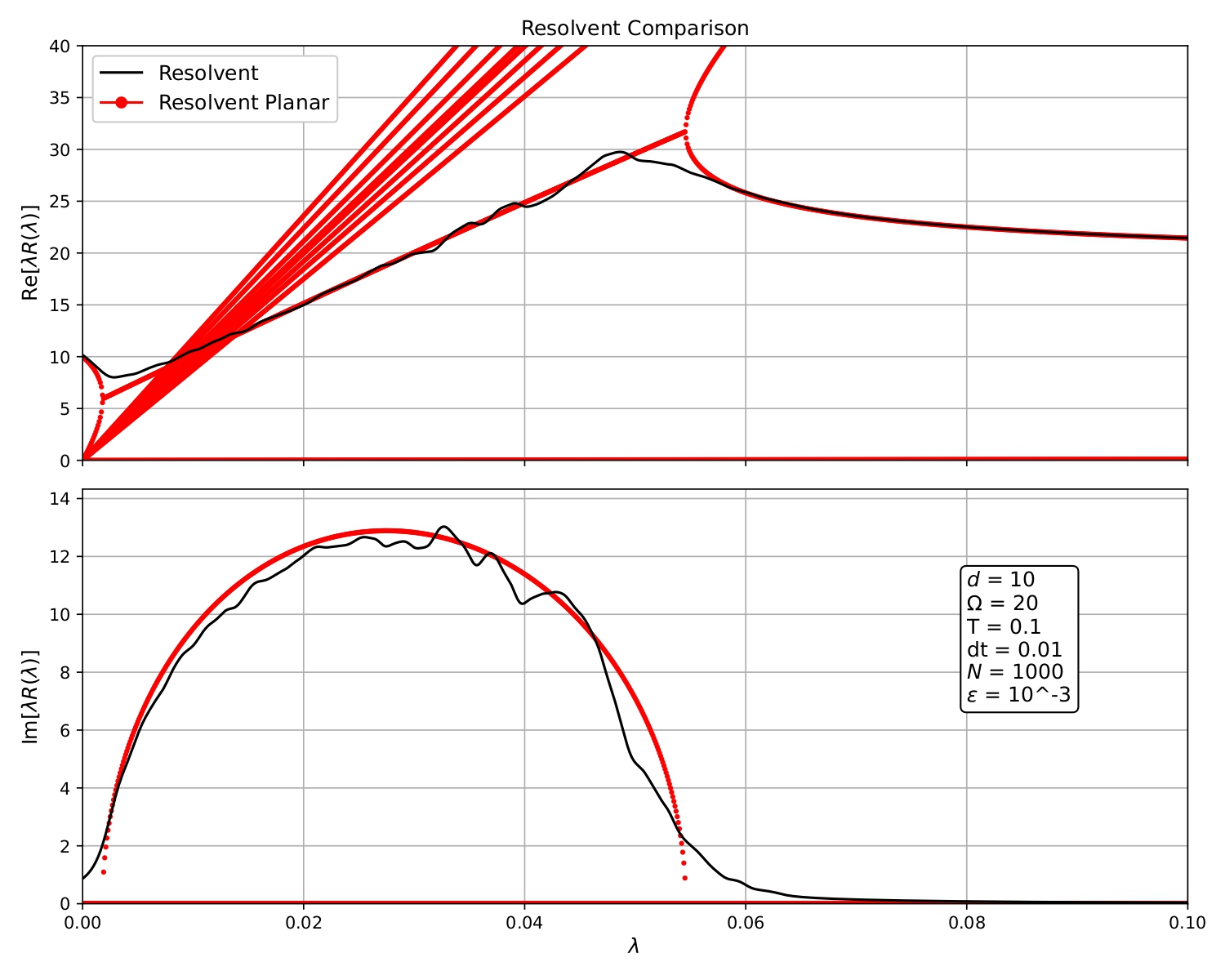}
    \caption{
    Plot of real and imaginary part of the mean resolvent $\lambda \overline{R}(\lambda)$, calculated from different methods for the Brownian GUE.
    The black solid line correspond to the resolvent calculated with the Gram matrix \eqref{eq:kerg} and the red point curve with \eqref{eq:resummedSD}, which only takes in account the planar diagrams.
    The parameters used in the calculation are $T=0.1$, $dt=0.01$, $d = 10$, $\Omega = 20$, $\epsilon = 0.001$, with averages taken over 1000 draws.
    }
    \label{fig:ResolventGUET01}
\end{figure}

What we need to do is to generalize the original SD equation \ref{eq:SchwingerDyson} from \cite{Penington:2019kki} to the case where $A_n$ above does not vanish. To find such generalization we assume the following factorization in the thermodynamic limit
\bea 
\overline{\braket{\Psi_{i}}{\Psi_{i_{1}}}  \cdots \braket{\Psi_{i_{n-1}}}{\Psi_{i}}\braket{\Psi_{i}}{\Psi_{i_{n}}}  \cdots \braket{\Psi_{i_{n+k-1}}}{\Psi_{i_{n+k}}}}\simeq \nonumber\\ \overline{\braket{\Psi_{i}}{\Psi_{i_{1}}}  \cdots \braket{\Psi_{i_{n-1}}}{\Psi_{i}}}\,\,\overline{\braket{\Psi_{i}}{\Psi_{i_{n}}}  \cdots \braket{\Psi_{i_{n+k-1}}}{\Psi_{i_{n+k}}}}\;,
\eea
where again $i\neq i_l$. This factorization is the same as the one that appears in free probability and the generalized Eigenstate Thermalization Hypothesis \cite{Pappalardi:2022aaz}. Here such assumption might be justified since the Gram matrix is a particular type of random matrix, and we remark that it is implicit in the previous analysis of the standard SD. In fact, when such inner products are computed in e.g. gravity in the semiclassical limit, it is immediate that corrections to such factorization are exponentially subleading. Assuming such factorization the generalization reads
\be\label{gensd}
\overline{R_{ij}(\lambda)} = \frac{1}{\lambda} \delta_{ij} + \frac{1}{\lambda} \sum_{n=1}^\infty \overline{R(\lambda)}^{n-1} \left (Z_n \overline{R_{ij}(\lambda)} +  A_n \frac{\overline{R(\lambda)}}{\Omega} (1-\delta_{ij}) \right )\;.
\ee
The idea is that these $A_n$ generalize $Z_n$ to the cases with different external indices. The internal indices are similar to the $Z_n$ ones, explaining that the factor $\overline{R(\lambda)}^{n-1}$ is the same for both. The extra factor $\frac{\overline{R(\lambda)}}{\Omega}$ provides some of the terms with $A_n$ times a typical term of the diagonal case, the rest being provided by the $Z_n \overline{R_{ij}(\lambda)}$ term.

Although this is more complicated than the original SD, the modification does not change the trace of the resolvent $\overline{R(\lambda)}$, i.e. the modification only affects the non-diagonal elements due to the factor $\left(1-\delta_{ij} \right)$. This explains why the numerical analysis of the previous section, which only looked at the trace of the resolvent, provided good results at all times.

Even if the trace, and therefore the computation of the dimension, remains unaffected by the fact that  $A_n$ does not vanish, it is interesting to test the validity of the new equation for the non-diagonal terms. This can be achieved since we have an explicit expression for the correction
\be\label{eq:resummedAn}
    \sum_{n=1}^\infty A_n \,\overline{R(\lambda)}^{n} =\overline{R(\lambda)} \,\, \overline{\bra{\Psi}} \left ( \mathbb{1} - \rho \overline{R(\lambda)} \right )^{-1} \overline{\ket{\Psi}}\;, 
\ee
where we remind that $\rho$ is the average state. Using eqs. \eqref{eq:resummedSD} and \eqref{eq:resummedAn} we can rewrite the matrix elements of the resolvent as
\begin{equation}
    \overline{R(\lambda)_{ij}} = \frac{\overline{R(\lambda)}}{\Omega} \delta_{ij} + \left( \frac{\overline{R(\lambda)}}{\Omega} \right)^{2} \overline{\bra{\Psi}} \left( \mathbb{1} - \rho \overline{R(\lambda)} \right)^{-1} \overline{\ket{\Psi}} \ (1 - \delta_{ij}) \ .
\label{eq:resolventij}
\end{equation}
Therefore, we have that the diagonal elements and the non-diagonal elements are all equal separately between them, as is expected on general grounds from the lack of structure of the matrix $G$.

We now test this expression for the case of the Brownian GUE ensemble defined in \eqref{eq:brownianGUEapp}, using similar methods for the simulation as the ones explained in \ref{sS:planarnumerics}. To compute the mean resolvent, we construct several instances of $G$, then we compute its resolvent for different values of $\lambda$ using \eqref{eq:resolapp}, and finally we take the average over the different draws. Because of the lack of structure of the resolvent, we also implement an average between the non-diagonal elements of the resulting mean resolvent, to get our final result. In this simulations, we evaluate at values of $\lambda$ in regions far from the locations of eigenvalues of $G$. The reason behind this is that in order to preserve the Hermitian structure of the resolvent, we have to put $\epsilon = 0$, where $\epsilon$ was defined in \ref{sS:planarnumerics}. The numerical average of the trace of the resolvent $\overline{R(\lambda)}$ and the density matrix $\rho$ is computed as in \ref{sS:planarnumerics}. The average vector $\overline{\vert \Psi \rangle}$ is computed by drawing several vectors from \eqref{eq:brownianGUEapp} and taking the average. 

In figures \ref{fig:ResolventNonDiagGUET01}, \ref{fig:ResolventNonDiagGUET1}, \ref{fig:ResolventNonDiagGUET10} we compare, for different times, the two computations of the non-diagonal element (ND) of the mean resolvent $\overline{R_{ij}(\lambda)}$, one doing the numerically simulation and the other using the formula \eqref{eq:resolventij} that follows from the generalized SD equation \eqref{gensd}. We find a perfect matching for the various times.

\begin{figure}[H]
    \centering
    \includegraphics[width=0.9\linewidth]{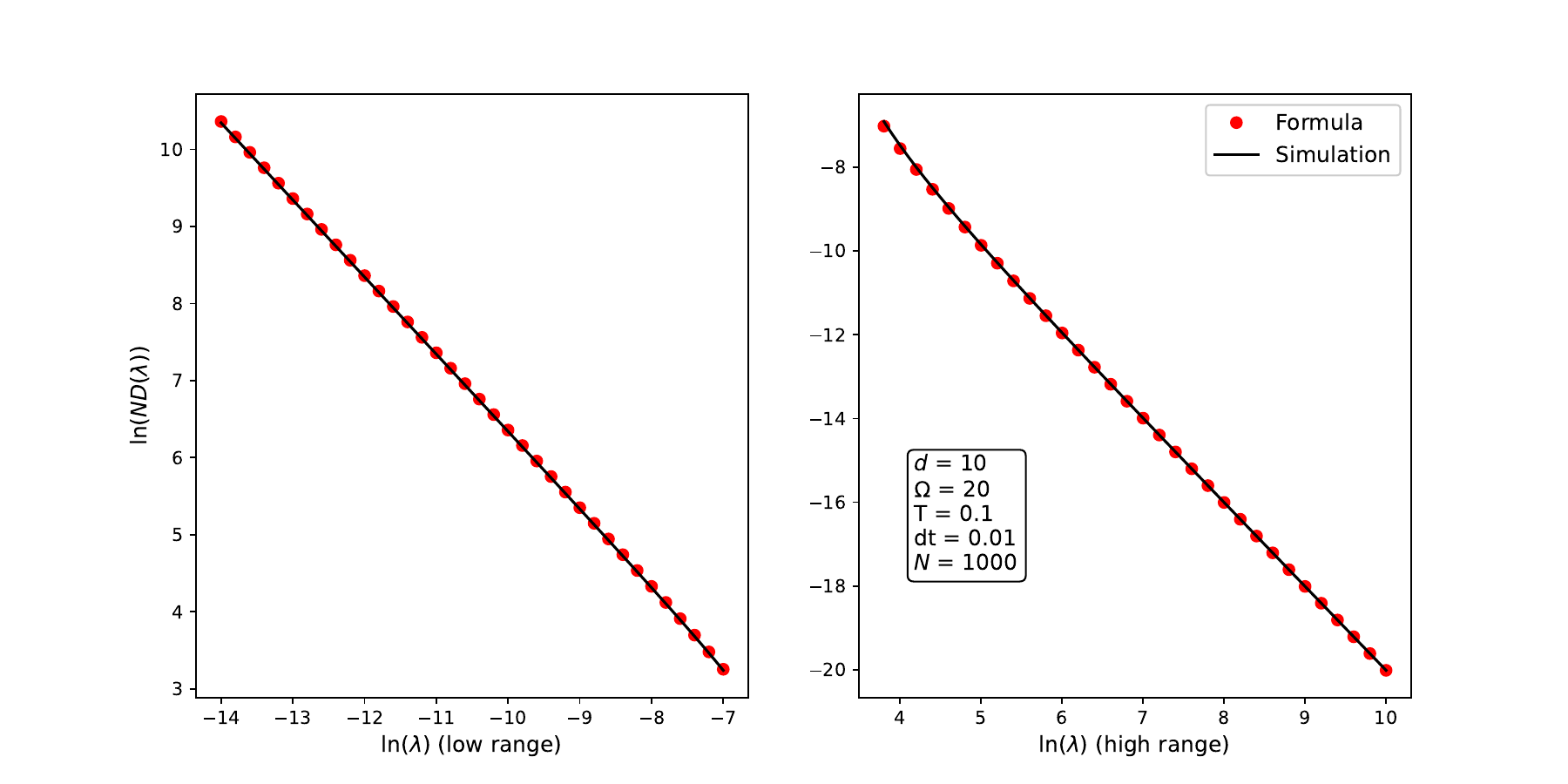}
    \caption{
    Plot of the non-diagonal elements of  $\overline{R(\lambda)_{ij}}$ (all are equal on average) calculated from different methods in the Brownian GUE for values of $\lambda$ far from the location of eigenvalues of the Gram matrix $G$.
    The black solid line correspond to one calculated numerically and the red point curve with the generalized SD equation \eqref{eq:resolventij}.
    The parameters used in the calculation are $T=0.1$, $dt=0.01$, $d = 10$, $\Omega = 20$, with averages taken over 1000 draws.
    }
    \label{fig:ResolventNonDiagGUET01}
\end{figure}

\begin{figure}[H]
    \centering
    \includegraphics[width=0.9\linewidth]{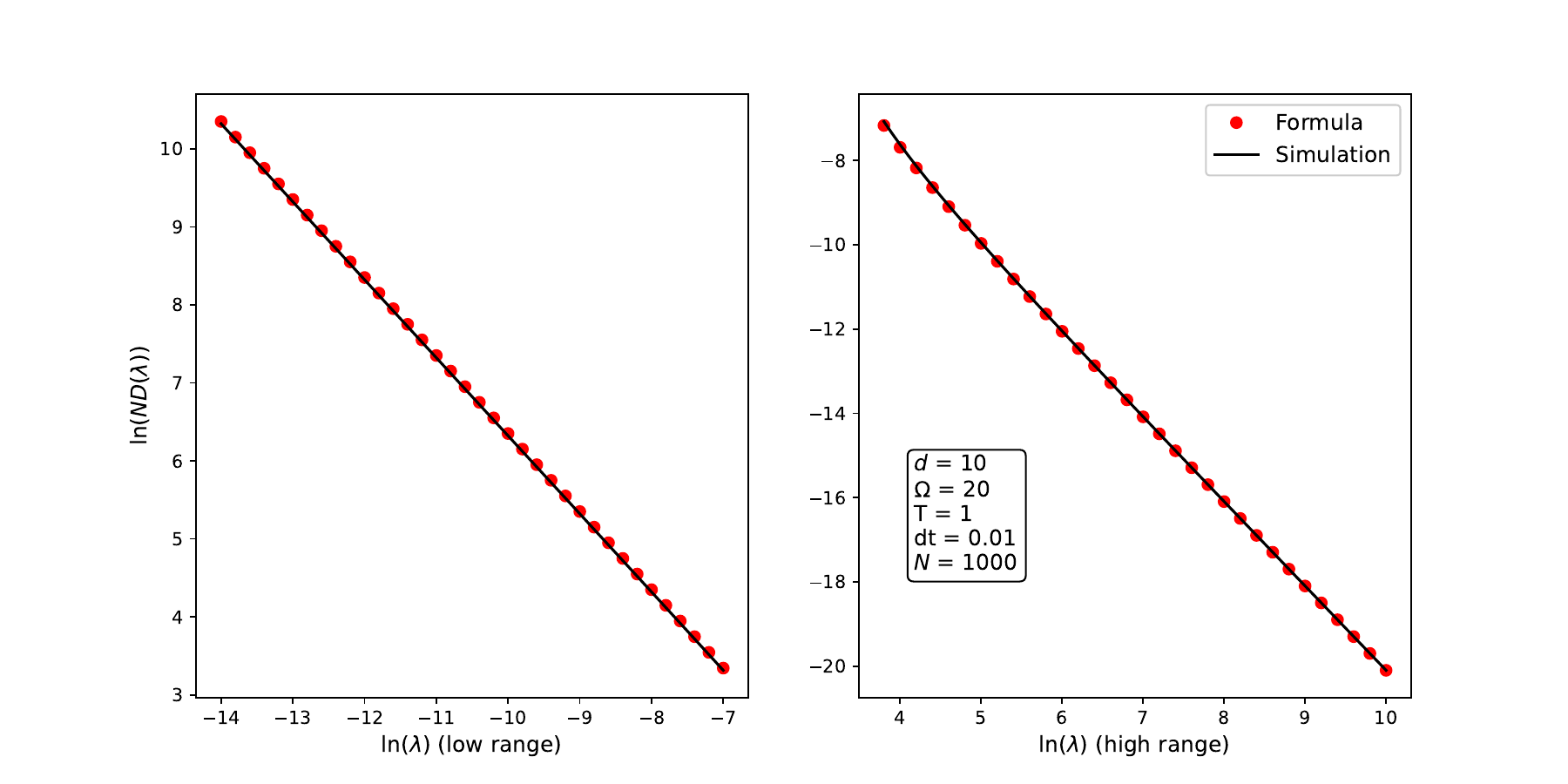}
    \caption{
    Plot of the non-diagonal elements of  $\overline{R(\lambda)_{ij}}$ (all are equal on average) calculated from different methods in the Brownian GUE for values of $\lambda$ far from the location of eigenvalues of the Gram matrix $G$.
    The black solid line correspond to one calculated numerically and the red point curve with the generalized SD equation \eqref{eq:resolventij}.
    The parameters used in the calculation are $T=1$, $dt=0.01$, $d = 10$, $\Omega = 20$, with averages taken over 1000 draws.
    }
    \label{fig:ResolventNonDiagGUET1}
\end{figure}

\begin{figure}[H]
    \centering
    \includegraphics[width=0.9\linewidth]{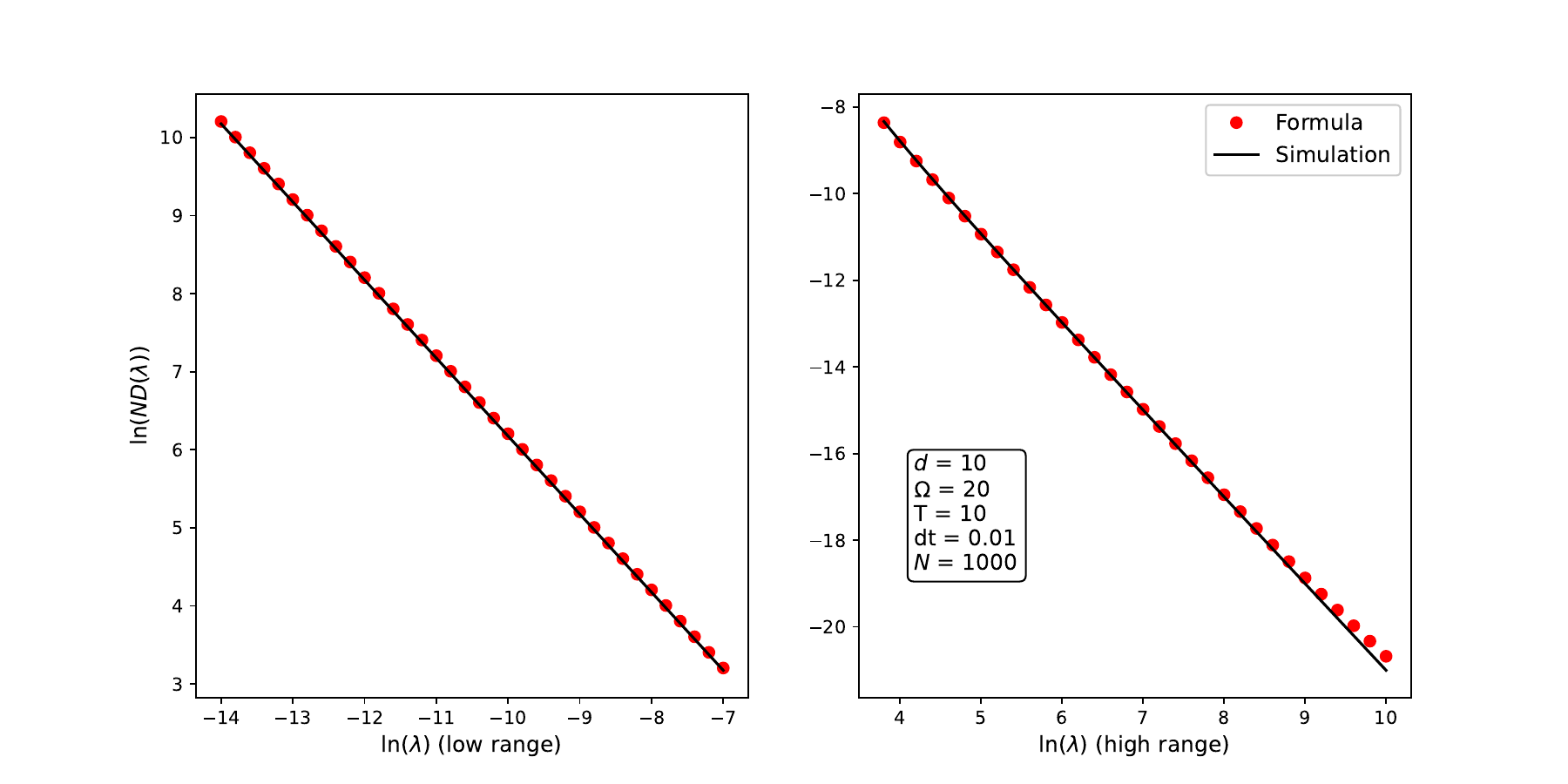}
    \caption{
    Plot of the non-diagonal elements of  $\overline{R(\lambda)_{ij}}$ (all are equal on average) calculated from different methods in the Brownian GUE for values of $\lambda$ far from the location of eigenvalues of the Gram matrix $G$.
    The black solid line correspond to one calculated numerically and the red point curve with the generalized SD equation \eqref{eq:resolventij}.
    The parameters used in the calculation are $T=10$, $dt=0.01$, $d = 10$, $\Omega = 20$, with averages taken over 1000 draws.
    }
    \label{fig:ResolventNonDiagGUET10}
\end{figure}

\section{$Z_n(t)$ for Brownian Spin Cluster and Brownian SYK}
\label{app:exactBrownian}

In this appendix, we supply detailed explanations for both the energies of the effective Hamiltonians and the replicated partition functions for the Brownian spin cluster and Brownian SYK, described in sections \ref{S:brownianspincluster} and \ref{sec:BrownianSYKexact}. Also, we provide numerical simulations of such ensembles and compute their partition functions, verifying the analytic expressions.

\subsection*{Brownian Spin Cluster}

We start with:

\begin{proof}[Proof of Lemma \ref{lemma:bspin1}]
Fix a string $P_\alpha$ of length $|\alpha|=s$ and choose a string $P_\gamma$ of length $|\gamma|=q$. Suppose they overlap on $m$ sites. There are $\binom{s}{m}$ ways to choose the overlap locations for a fixed $P_\alpha$. The remaining $q-m$ sites of $P_\gamma$ lie outside $\alpha$ and do not affect $k(\gamma,\alpha)$, producing a degeneracy $3^{q-m}\binom{N-s}{\,q-m\,}$ upon summing over $\gamma$ in \eqref{eq:specBspinHeff}. On the $m$ overlapping sites, let $l\le m$ be the number of positions where the Pauli matrices differ; each such position has two choices (the two Pauli matrices distinct from the fixed one on $P_\alpha$), giving $\binom{m}{l}2^l$. Only strings with odd $k(\gamma, \alpha)=l$ contribute (with a factor of $2J$) to \eqref{eq:specBspinHeff}, and $\sum_{l\ \mathrm{odd}}^{m}\binom{m}{l}2^l=\frac{1}{2}\!\left(3^m-(-1)^m\right)$. Putting everything together yields 
\be\label{eq:explicitEs}
    E(s) = J \sum_{m=0}^{q} \binom{N-s}{q-m} 3^{q-m} \binom{s}{m} \left (3^m- \left(-{1}\right)^m \right )\,,
\ee
where we adopt the convention ${s\choose m}=0$ for $m>s$ (relevant only when $s<q$). Now, the first term in \eqref{eq:explicitEs} can be resummed using Vandermonde's identity $\sum_{m=0}^{q} \binom{s}{m}\binom{N-s}{q-m}= {N\choose q}$. The second term is the Kravchuk polynomial $K_q(s;x,N)\;\equiv\;\sum_{m=0}^{q}(-1)^m\binom{s}{m}\binom{N-s}{q-m}\!\left(\frac{1-x}{x}\right)^{\!m}$ with $x=\frac{3}{4}$, which admits a representation in terms of the hypergeometric function $K_q(s;x,N) = {N \choose q} {}_2F_1\!\left(\!-q,\,-s;\,-N;\,\frac{1}{x}\right)$, and this yields \eqref{eq:specBspinHefffinal}. 
\end{proof}

As an observation, we can generalize the idea of $q$-locality and consider driving operators of size $s \leq q$. This will only change the energies of the effective Hamiltonian, but it will be diagonalized by the same eigenvectors, namely the Pauli strings. These energies are
\be
\label{eq:explicitEs2}
 E(s) = J \sum_{m=0}^q \binom{N-s}{q-m} 3^{q-m} \sum_{j=0}^{m} \binom{s}{j} \left (3^j - (-1)^j \right ) \, ,
\ee
where we see we recover the previous case just by considering the term $j=m$. Given that the energies still depend only on the size of the Pauli strings, the partition function will be exactly the same as the previous case, just changing (\ref{eq:explicitEs}) for (\ref{eq:explicitEs2}). 

\begin{proof}[Proof of Lemma \ref{lemma:bspin2}]
In \eqref{eq:ZntBrownianspin} each site $i = 1, \ldots, N$ can host $m_i = 0, \ldots, n$ Pauli matrices, since we have $n$ Pauli strings, and each string may either insert a Pauli matrix on a given site or leave it untouched. However, not all configurations contribute: due to the trace, the product of the $m_i$ Pauli matrices at each site must equal the identity (up to an overall phase that vanishes upon taking the absolute value squared). 

There are $3^{m_i}$ possible products of $m_i$ Pauli matrices on a given site $i$, and we are interested in counting how many of these yield the identity up to a phase. 
This number is given by
\be\label{eq:Npauli}
a_{m_i} = \frac{1}{4}\left(3^{m_i} + 3(-1)^{m_i}\right)\,,
\ee
which follows from the simple recurrence relation: if $a_{m_i}$ products yield the identity, then the remaining $3^{m_i} - a_{m_i}$ are proportional to Pauli matrices. For $m_i+1$, only those configurations that were nontrivial at step $m_i$ can produce the identity, leading to $a_{m_i+1} = 3^{m_i} - a_{m_i}$ with $a_0 = 1$, whose solution is precisely the expression above.

 Using this last result, we can write an exact formula for $Z_n(t)$ in terms of multinomial coefficients, starting from (\ref{eq:ZntBrownianspin}). The idea is explicitly putting all operators by summing over indices that indicate in which sites each operator acts non-trivially. By adding the $a_{m_j}$ (\ref{eq:Npauli}) factors ($j = 1, \ldots,N$), we ensure to keep only with the product of $n$ operators that have non-zero trace (i.e., are the identity up to a phase). This is given by
\be
Z_n(t) = \frac{1}{4^{N(n-1)}} \sum_{k_0 + \ldots +  k_N = n} \binom{n}{k_0, \ldots,k_N} \left (\prod_{j=1}^N a_{m_j} \right ) e^{-tE(0) k_0} \ldots e^{-tE(N) k_N} \, ,
\ee
where $k_i$ counts all the chains of size $i$. This means that in fact for each $k_i$ we have $\binom{N}{i}$ indices, each one counting the number of chains of size $i$ that act on $i$ particular sites of the spin chain. So, we have $2^N$ indices in total. The index $m_j$ is the sum of all the indices that act on the site $j$. For example for $N = 3$, we have $k_1 = n_1 + n_2 + n_3$, $k_2 = n_{12} + n_{13} + n_{23}$, $k_3 = n_{123}$, $k_0 = n - k_1 - k_2 - k_3$ and $m_1 = n_1 + n_{12} + n_{13} + n_{123}$, $m_2 = n_2 + n_{12} + n_{23} + n_{123}$, $m_3 = n_3 + n_{13} + n_{23} + n_{123}$, where the sub-indices indicate in which sites the operator acts non-trivially.

Then, by expanding the product of the $a_{m_j}$ we find terms of the form
\be
\frac{1}{4^N} \left (3^s \left (-1 \right )^{m_1 + \ldots +  m_{s}} + 3^{m_{s+1} + \ldots +  m_N} \right ) \, ,
\label{eq:termprodaN}
\ee
where $s = 0, \ldots,N$. Given that all sites are equivalent, the product of the $a_{m_j}$ is equal to the sum of these terms with a factor $\binom{N}{s}$ (because they are dummy indices). So, the whole problem comes down to understanding how many indices do $m_1, \ldots,m_s$ share. This question simplifies because we are interested in separating the indices according to their size. Thus, given $s$ sites, from all the indices of size $l$, $\binom{s}{j} \binom{N-s}{l-j}$ are shared by $j$ of those sites. For $s = N$, it is trivial that all indices of size $l$ are shared by $l$ sites. In (\ref{eq:termprodaN}) we observe we can express $m_{s+1} + \ldots +  m_N = \sum_j m_j - m_1 - \ldots,m_s$. Then, the negative terms are equal to the ones in the first term of (\ref{eq:termprodaN}). Taking all this into account, we can write
\be
Z_n(t) = \frac{1}{4^{Nn}} \sum_{s = 0}^N \binom{N}{s} 3^s \left ( \sum_{l=0}^N  3^l e^{-tE(l)} \sum_{j=0}^s \left (-3 \right )^{-j} \binom{s}{j} \binom{N-s}{l-j} \right )^n \, ,
\label{eq:Znqlocal}
\ee
where we notice that inside the parenthesis, the exponents are just the number of sites that share the indices of size $l$, $3^l$ coming from $\sum_j m_j$ and $(-1)^j ~ 3^{-j} = (-3)^{-j}$ coming from $m_1 + \ldots +  m_s$. Therefore, we obtain the eigenvalues (\ref{eq:ptsspin}) with their corresponding degeneracies, just comparing with (\ref{eq:lemmabrownianZn(t)}).
\end{proof}

As an observation, this replicated partition function satisfies a recurrence relation in $N$. This is a direct consequence of the fact that the energies depend only on the size of the Pauli strings. Thus, if we define $e^{-tE(s)} \equiv x_s$ and make explicit the dependence on $N$ of the partition function, $Z_n(t) \to Z_{n,N}(x_0, \ldots x_N)$, we have

\begin{align} \label{eq:ZnNrecspin}
\notag Z_{n,N+1}(x_0,\ldots, x_{N+1}) = \frac{1}{4^n} & \left [ Z_{n,N}(x_0+3x_1,x_1+3x_2,\ldots, x_N+3x_{N+1}) + \right . \\
& \left . + 3 Z_{n,N}(x_0 - x_1,x_1 - x_2,\ldots, x_N - x_{N+1}) \right ] \, ,
\end{align}
where we defined $Z_{n,0}(x_0) \equiv x_0^n$.

We now numerically simulate the Brownian spin cluster described in \ref{S:brownianspincluster} and verify equation \eqref{eq:Znqlocal}.
To calculate $Z_{n}(t)$ numerically we used that (see section \ref{sec:overlapgeneralexp}, lemma \ref{lemma:densitymatrix})
\begin{equation}
    Z_{n}(t) = \Tr(\rho(t)^{n}), \quad
    \rho(t) = \overline{\ketbra{\Psi(t)}} \ .
\label{eq:Znsimu1}
\end{equation}
Here the states $\ket{\Psi(t)}$ are
\begin{equation}
    \ket{\Psi(t)} = \prod_{l=1}^{M} (e^{-i H_{l} \delta t} \otimes I) \ket{\mathbf{1}}\;,
\end{equation}
where we discretize the time as $t = M \delta t$, and each $H_{l}$ is a particular realization of the Hamiltonian \eqref{eq:hamiltonianspincluster}. 

More concretely, for the simulation we take the number of spins $N=4$, $q=2$, $\delta t = 1$, $J= 1/10$ and to take the average in \eqref{eq:Znsimu1} we draw $10000$ vectors $\ket{\Psi(t)}$. In this case, we have that the gap time is $t_{\text{gap}} \propto 1/E_{\text{gap}} = 16.75$. In figure \ref{fig:znqlocalsimu} we plot the values of $Z_{n}(t)$ as a function of $n$ for times $t = \{ 10, 17, 30 \}$, finding perfect agreement with the analytic result. 

\begin{figure}[h]
    \centering
    \includegraphics[width=0.9\linewidth]{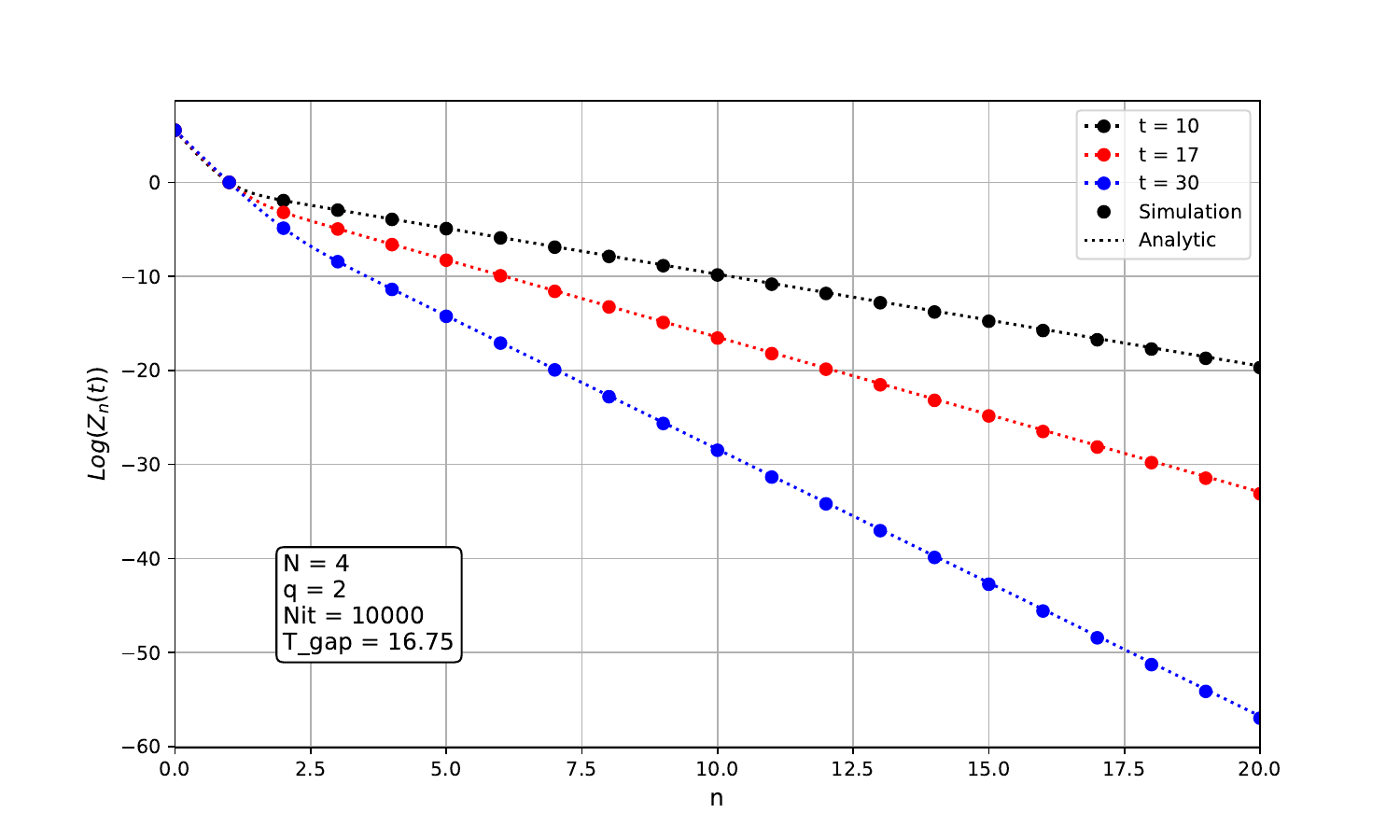}
    \caption{
        Plot of $\log(Z_{n}(t))$ as a function of $n$ for different values of $t$. The dotted lines are the result from the analytic expression \eqref{eq:Znqlocal}. The solid dots are the results of the simulation. The parameters of the simulation were $N=4$, $q=2$, $J=1/10$ and the average was taken over 10000 draws.
    }
    \label{fig:znqlocalsimu}
\end{figure}

\subsection*{Brownian SYK}

\begin{proof}[Proof of Lemma \ref{lemma:bSYK1}]
The state $\psi_\alpha \ket{\mathbf{1},\pm}$ is an eigenstate of \eqref{eq:HeffSYK} with energy
\be\label{eq:specSYKHeff}
E_\alpha = J \sum_{|\gamma| = q } \left (1- (-1)^{k(\gamma, \alpha)} \right)  \,,
\ee
since $\psi_\alpha\psi_\gamma = (-1)^{k(\gamma, \alpha)}\psi_\gamma\psi_\alpha$, and $k(\gamma, \alpha)$ is the number of common sites between both strings. To evaluate the sum, we consider a string $\psi_\alpha$ of length $|\alpha|=s$ and want to count the possible $\psi_\gamma$ strings of length $|\gamma|=q$ with fixed $k(\gamma, \alpha)$. There are $\binom{s}{m}\binom{N-s}{\,q-m\,}$ such strings.
Then the spectrum of \eqref{eq:HeffSYK} is
\be\label{eq:specSYKHeff2}
E(s) = J \sum_{m=0}^{q} \binom{N-s}{\,q-m\,} \binom{s}{m}\left (1- (-1)^{m} \right)  \,.
\ee
As for the spin cluster, the first term in \eqref{eq:specSYKHeff2} can be resummed using Vandermonde's identity. The second term is the Kravchuk polynomial $K_q(s;x,N)$ with $x=\frac{1}{2}$, which admits a representation in terms of the hypergeometric function $K_q(s;x,N) = {N \choose q} {}_2F_1\!\left(\!-q,\,-s;\,-N;\,\frac{1}{x}\right)$, and this yields \eqref{eq:specSYKHefffinal}. 
\end{proof}

\begin{proof}[Proof of Lemma \ref{lemma:syk2}]
In (\ref{eq:ZnSYKNeven}), we observe that $Z^\pm_n(t)$ can be computed analogously as in the spin case, just changing the number of relevant operators in each site (\ref{eq:Npauli}) for $a^\pm_m = \frac{1}{2} (1 \pm (-1)^m)$, given that in this case we have only one possible traceless operator acting on each site instead of three. Both partition functions end up being equal because of the energy symmetry $E(s) = E(N-s)$, so comparing with (\ref{eq:Znqlocal}), we arrive at

\be
Z_n(t) = \frac{1}{2^{Nn}} \sum_{s\text{ even}}^{N} \binom{N}{s}\left [ \sum_{l=0}^{N} e^{-tE(l)} \sum_{j=0}^{s} (-1)^j \binom{s}{j} \binom{N-s}{l-j} \right ]^n \, .
\label{eq:ZnSYKexact}
\ee
As mentioned before, although presented for even $N$, this expression is actually valid for all $N$. Let us then expand on the case of $N$ odd. The idea is that we can represent irreducibly $N$ Majorana fermions in a Hilbert space of dimension $d = 2^{[N/2]}$. For $N$ odd, the last Majorana fermion is the fermion parity operator of $N-1$, i.e., is proportional to the product of the other $N-1$ fermions (for $N=1$ we can trivially represent the Majorana fermion as $1$). Of course, now the product of the $N$ Majorana fermions will be proportional to the identity. Thus, the ground state of the effective Hamiltonian (\ref{eq:HeffSYK}) is non-degenerate and the Brownian circuit will generate the whole (duplicated) Hilbert space of dimension $d^2 = 2^{N-1}$. A basis of eigenstates of (\ref{eq:HeffSYK}) is $\vert \psi_\alpha \rangle = \frac{1}{\sqrt{d}} \psi_\alpha \vert \mathbf{1} \rangle$, with energy given by (\ref{eq:specSYKHefffinal}) for $s = \abs{\alpha}$. Just like the $N$ even case, we shall restrict to $\abs{\alpha} < N/2$, due to the fact that chains of longer size are in a one-to-one relation with shorter ones. Then, the replicated partition function is
\begin{align}
\notag Z_n(t) &= \frac{1}{d^{2n}} \sum_{\alpha_1, \ldots, \alpha_n}^{\abs{\alpha} < N/2} e^{-t(E_{\alpha_1} + \ldots +  E_{\alpha_n})} \abs{\text{Tr}(\psi_{\alpha_1} \ldots \psi_{\alpha_n})}^2 \\
&= \frac{1}{d^{2n}} \frac{1}{2^n} \sum_{\alpha_1, \ldots, \alpha_n}^{\abs{\alpha} \leq N} e^{-t(E_{\alpha_1} + \ldots +  E_{\alpha_n})} \abs{\text{Tr}(\psi_{\alpha_1} \ldots \psi_{\alpha_n})}^2 \, ,
\label{ZnSYKNodd}
\end{align}
where we extended the sum to all chains adding the factor $1/2^n$. This formula is analogous to (\ref{eq:ZntBrownianspin}), but now we have $2$ different cases that result in a non-vanishing trace: all $N$ sites empty or full. So, we shall sum two terms as we did in (\ref{eq:ZnSYKNeven}), using $a^\pm_m = \frac{1}{2} (1 \pm (-1)^m)$. As for $N$ even, these two end up being equal due to the energy symmetry. So it is direct to compute (\ref{ZnSYKNodd}) and find it coincides with (\ref{eq:ZnSYKexact}). 
\end{proof}

Also analogous to the spin case, we can obtain a recurrence relation for this partition function. Given that it is written in variables $x_s \equiv e^{-tE(s)}$, we may not consider the energy symmetry. Thus, for each partition function $Z_{n,N}^\pm(x_0, \ldots,x_N)$, we have
\begin{align}
\notag Z_{n,N+1}^\pm(x_0, \ldots, x_{N+1}) = \frac{1}{2^n} &\left [ Z_{n,N}^\pm(x_0 + x_1 , \ldots, x_N + x_{N+1}) \right .\\
& \left . \pm Z_{n,N}^\pm(x_0 - x_1 , \ldots, x_N - x_{N+1}) \right ] \, ,
\end{align}
with the initial condition $Z_{n,0}^\pm(x_0) \equiv x_0^n$, being then $Z_n(t) \equiv Z_{n,N} = (Z_{n,N}^+ + Z_{n,N}^-)/2$. Of course, if we use the energy symmetry, we recover $Z_{n,N} = Z_{n,N}^+ = Z_{n,N}^-$, so we only need to compute one of these partition functions.

Finally, we simulated the SYK model described in \ref{sec:BrownianSYKexact} to verify the analytic expression \eqref{eq:ZnSYKexact}.
For this, we used the Jordan–Wigner representation of the Majorana fermions
\begin{equation}
    \psi_{j} = \sigma_{3}^{\otimes \frac{j-1}{2}} \otimes \sigma_{1} \otimes I^{\otimes \frac{N - j - 1}{2}}, \quad \psi_{j+1} = \sigma_{3}^{\otimes \frac{j-1}{2}} \otimes \sigma_{2} \otimes I^{\otimes \frac{N - j - 1}{2}} \ ,
\end{equation}
where $N$ is the quantity of Majorana fermions and it is even. 
To calculate $Z_{n}(t)$ we used that (see section \ref{sec:overlapgeneralexp}, lemma \ref{lemma:densitymatrix})
\begin{equation}
    Z_{n}(t) = \Tr(\rho(t)^{n}), \quad
    \rho(t) = \overline{\ketbra{\Psi(t)}} \, .
\label{eq:Znsimu2}
\end{equation}
Here the states $\ket{\Psi(t)}$ are
\begin{equation}
    \ket{\Psi(t)} = \prod_{l=1}^{M} (e^{-i H_{l} \delta t} \otimes I) \ket{\mathbf{1}}\, ,
\end{equation}
where we discretize the time as $t = M \delta t$ and each $H_{l}$ is a particular realization of the Hamiltonian \eqref{eq:timedepHSYK}. 

\begin{figure}[h]
    \centering
    \includegraphics[width=0.9\linewidth]{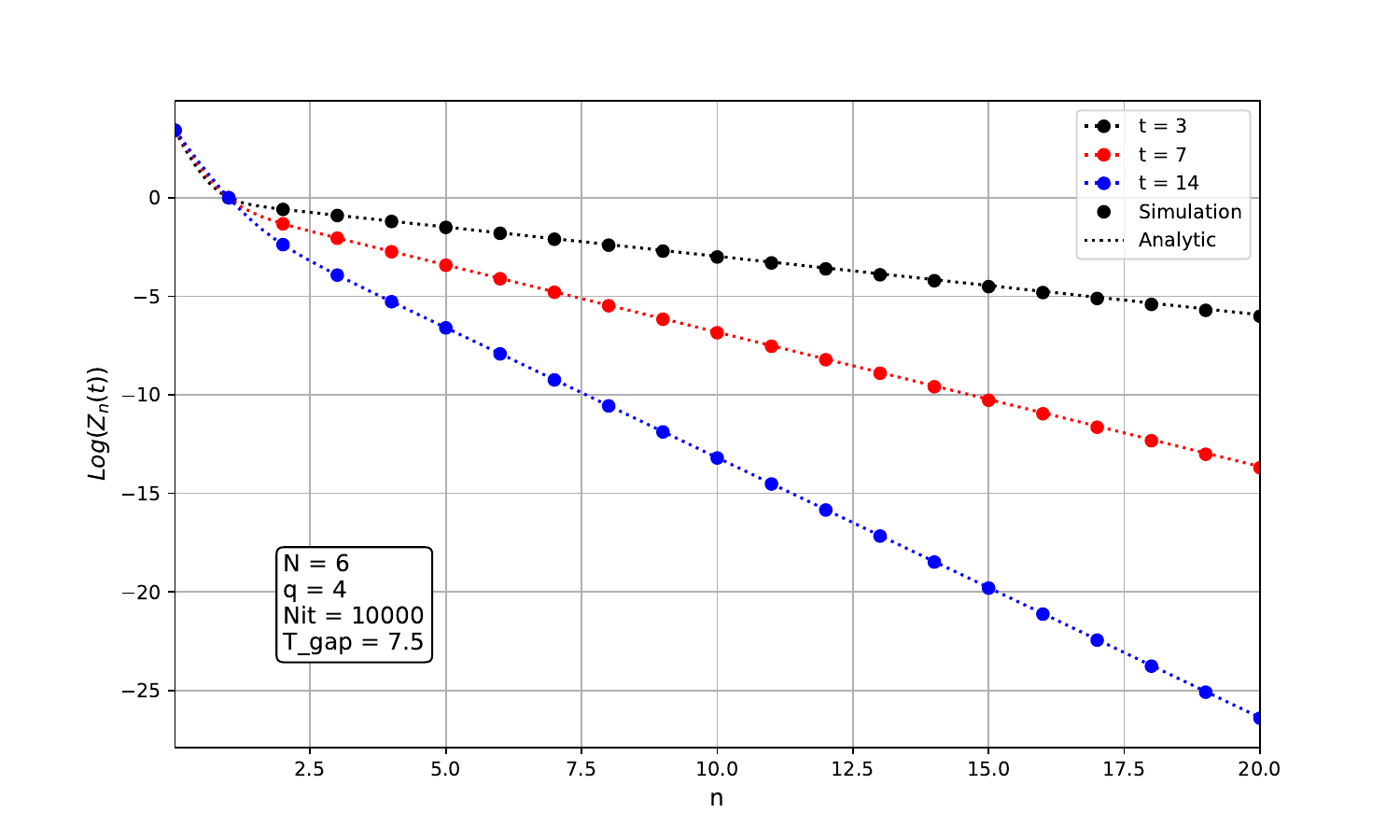}
    \caption{
        Plot of $\log(Z_{n}(t))$ as a function of $n$ for different values of $t$. The dotted lines are the result from the analytic expression \eqref{eq:ZnSYKexact}. The dots are the result calculated from the simulation. The parameters of the simulation were $N=6$, $q=4$, $J=1/10$ and the average was taken over 10000 draws.
    }
    \label{fig:znsylsimu}
\end{figure}

More concretely, for the simulation we take the number of fermions $N=6$, $q=4$, $\delta t = 1$, $J= 1/10$ and perform the average in \eqref{eq:Znsimu2} with $10000$ vectors $\ket{\Psi(t)}$. In this case, we have that the gap time is $t_{\text{gap}} \propto 1/E_{\text{gap}} = 7.5$. In figure \ref{fig:znsylsimu} we plot the values of $Z_{n}(t)$ as a function of $n$ for times $t = \{ 3, 7, 14 \}$, and find perfect agreement with the analytic expression.

\section{On the path integral for Brownian SYK}
\label{app:ferdet}

In this appendix we provide the details of the computation of the Brownian SYK path integral described in section \ref{subsec:largeNsyk}. To do this we start by proving a general useful result to compute fermionic path integrals:
\begin{lemma}
    Given a general operator $A(\tau) \in \mathbb{C}^{n \times n}$, the renormalized determinant
\be
D(A) = \abs{\partial_\tau - A(\tau)} \, ,
\ee
with boundary condition $\mathbf{v}(t) = B\mathbf{v}(0)$, $\mathbf{v} \in \mathbb{C}^n$ and $B \in \text{GL}(n,\mathbb{C})$ reads
\be\label{pxa}
D(A) =  \frac{P(1,A)}{\sqrt{P(0,A)}} \,\hspace{1cm} \textrm{with} \hspace{1cm} P(x,A)\equiv \abs{x \mathbb{1}_n - B^{-1} \exp{\int_0^t d\tau A(\tau)}}\;.
\ee
\label{GYB_lemma}
\end{lemma}

\begin{proof} The associated eigenstate equation is
\be
\left ( \partial_\tau - A(\tau) \right ) \mathbf{v}(\tau) = \lambda \mathbf{v}(\tau) \, , \;\;\; \lambda \in \mathbb{C} \, ,
\ee
whose general solution is
\be
\mathbf{v}(\tau) = \exp{\lambda \tau + \int_0^\tau d\tau' A(\tau')} \mathbf{v}(0) \, .
\ee
By imposing the boundary condition, we find
\be
\mathbf{v}(t) = \exp{\lambda t + \int_0^t d\tau A(\tau)} \mathbf{v}(0) = B \mathbf{v}(0) \Rightarrow \abs{x \mathbb{1}_n - B^{-1} \exp{\int_0^t d\tau A(\tau)}} = P(x,A) = 0 \, ,
\label{defPA}
\ee
for $x \equiv e^{-\lambda t}$. This tells us that for each $x_j(A)$ ($j = 1, \ldots,n$) root of $P(x,A)$, we have an eigenvalue $\lambda_{j,k}(A) = -\frac{1}{t} \ln{(x_j(A))} + \iw 2 k \pi$.
In order to regularize the determinant, we should divide by the free case, so
\be
\frac{D(A)}{D(0)} = \prod_{j = 1}^n \prod_{k \in \mathbb{Z}} \frac{\lambda_{j,k}(A)}{\lambda_{j,k}(0)} = \prod_{j = 1}^n \prod_{k \in \mathbb{Z}} \left [1 + \frac{\iw \ln{\left (x_j(A)/x_j(0) \right )}}{\iw \ln{\left (x_j(0) \right )} + 2 k \pi} \right ] = \prod_{j = 1}^n \frac{\sin{\left (\iw \ln{\sqrt{x_j(A)}} \right )}}{\sin{ \left (\iw \ln{\sqrt{x_j(0)}}\right )}} \, ,
\ee
where we used that
\be
\prod_{k \in \mathbb{Z}} 1 - \frac{x}{a + 2 k \pi} = \frac{\sin{\left ((a-x)/2 \right )}}{\sin{(a/2)}} \, .
\ee
Thus, we obtain
\be
\frac{D(A)}{D(0)} = \prod_{j = 1}^n \sqrt{\frac{-x_j(0)}{-x_j(A)}} \left [\frac{1-x_j(A)}{1-x_j(0)} \right ] \, ,
\ee
from which we can extract the regularized determinant
\be
D(A) = \prod_{j=1}^n \frac{1 - x_j(A)}{\sqrt{-x_j(A)}} = \frac{P(1,A)}{\sqrt{P(0,A)}} \, ,
\label{GYB}
\ee
and where we note the minus inside the square root is a convention, in order to obtain $D(0) = 2^n$ for $B = -\mathbb{1}_n$.
\end{proof}
This lemma is the generalization to fermions of the well-known Gelfand-Yaglom formula in the bosonic scenario \cite{Gelfand:1959nq}.

Given this result, we only need to compute $P(x,\mathbb{\Sigma})$ as defined by (\ref{pxa}) to find the determinant $D_n(\mathbb{\Sigma})$ (\ref{Majoranadet}). The fermionic mass matrix that arises in section \ref{subsec:largeNsyk} is block diagonal and can be explicitly written in terms of Pauli matrices and the Lagrange multipliers $\Sigma_m(\tau)$ as
\be
\mathbb{\Sigma}_{ij}(\tau) = \delta_{i,k-1} \mathbb{\Sigma}'_{kl}(\tau) \delta_{l-1,j} \, , \;\;\; \mathbb{\Sigma}' = \bigoplus_{m=1}^n -\iw \sigma_y \frac{\Sigma_m(\tau)}{2} \, .
\label{massmatrix}
\ee
In turn, the boundary conditions can be written in the same basis as
\be
\Psi_a(t) = \left ( \bigoplus_{m=1}^n \iw \sigma_y \right ) \Psi_a(0) \, , 
\ee
where the vector of Majorana fermions is ordered as
\be 
\Psi_a^\dagger(\tau) = (\psi_a^{1,1}(\tau) \; \psi_a^{1,2}(\tau) \; \psi_a^{2,1}(\tau) \; \ldots \; \psi_a^{n,2}(\tau) )\;.
\ee
To exponentiate (\ref{massmatrix}), we work in the basis $\mathbf{v}(t)$ where the components of $\mathbf{v}(\tau)$ are cyclically rotated in relation to $\Psi_a(\tau)$, being $\mathbf{v}^\dagger(\tau) = (v^{n,2}(\tau) \; v^{1,1} (\tau) \; v^{1,2}(\tau) \; \ldots \; v^{n,1}(\tau))$, leading to
\be
\exp{{\int_0^t d\tau \mathbb{\Sigma}'(\tau)}} = \bigoplus_{m=1}^n \cos{\left ({\int_0^t d\tau \frac{\Sigma_m(\tau)}{2}} \right )} - \iw \sigma_y \sin{\left ({\int_0^t d\tau \frac{\Sigma_m(\tau)}{2}} \right )} \, .
\ee
In this basis, the matrix $B \in \text{GL}(2n,\mathbb{C})$ takes the form
\be
B = 
\begin{pmatrix}
0 & \ldots & \ldots & \ldots & \ldots & \ldots & \ldots & \ldots & -1 \\
0 & 0 & 1 & 0 & \ldots & \ldots & \ldots & \ldots & 0 \\
0 & -1 & 0 & 0 & 0 & \ldots & \ldots & \ldots & 0 \\
0 & 0 & 0 & 0 & 1 & 0 & \ldots & \ldots & 0 \\
0 & 0 & 0 & -1 & 0 & 0 & \ldots & \ldots & 0 \\
\vdots & & & & & \ddots & & & \vdots \\
0 & \ldots & \ldots & \ldots, & \ldots & \ldots & 0 & 1 & 0 \\
0 & \ldots & \ldots & \ldots & \ldots & \ldots & -1 & 0 & 0 \\
1 & 0 & \ldots & \ldots & \ldots & \ldots & \ldots & \ldots & 0 \\
\end{pmatrix}
\, ,
\ee
which has as inverse $B^{-1} = - B$. Therefore,
\be
P(x,\mathbb{\Sigma}) = \abs{x \mathbb{1}_{2n} + B \exp{{\int_0^t d\tau \mathbb{\Sigma}'(\tau)}}} \, .
\ee
This determinant has the form
\be
P(x,\mathbb{\Sigma}) = 
\begin{vmatrix}
x & c_2 & -s_2 & 0 & \ldots & \ldots & \ldots & \ldots & \ldots & 0 \\
-c_1 & x & 0 & \ldots & \ldots & \ldots & \ldots & \ldots & 0 & -s_1 \\
0 & 0 & x & c_3 & -s_3 & 0 & \ldots & \ldots & \ldots & 0 \\
\vdots & & & \ddots & & & & & &  \vdots \\
0 & \ldots & \ldots & 0 & x & c_k & -s_k & 0 &\ldots & 0 \\
0 & \ldots & 0 & -s_{k-1} & -c_{k-1} & x & 0 & \ldots & \ldots & 0 \\
\vdots & & & & & & \ddots & & & \vdots \\
\vdots & & & & & & & \ddots & & \vdots \\
-s_1 & 0 & \ldots & \ldots & \ldots & \ldots &  \ldots & 0 & x & c_1 \\
0 & \ldots & \ldots & \ldots & \ldots &\ldots & 0 & -s_n & -c_n & x
\end{vmatrix}
 \, ,
\ee
where $c_j \equiv \cos{ \left (\frac{\int_0^t d\tau \Sigma_j (\tau)}{2} \right )}$, $s_j \equiv \sin{ \left (\frac{\int_0^t d\tau \Sigma_j (\tau)}{2} \right )}$. Essentially, we have $x$ in the diagonal elements and $c_i \; -s_i$ to the right if the row is odd and $-s_p \; -c_p$ to the left if it is even. The subindices $i = 2, \ldots, n, 1$ and $p = 1, \ldots, n$ run over odd and even rows respectively. So we obtain $P(0,\mathbb{\Sigma}) = 1$ and $P(1,\mathbb{\Sigma}) = 2f_n(t)$, given by (\ref{poldetMajorana}).

\section{Sampling states from a time-independent Hamiltonian}
\label{app:H}

For a single drive operator ($K=1$), which we denote $\mathsf{h}\equiv \Op_1$, the time-evolution operator can be written as a time-independent evolution
\be\label{eq:unitarytimeindep}
U(t) = e^{-\iw t_g \mathsf{h}}\,,\qquad t_g = \int_0^t \text{d}t \,g(t)\,.
\ee
Sampling different realizations of the Brownian coupling $g(t)$ at fixed time $t$ amounts to considering the unitary \eqref{eq:unitarytimeindep}, and the associated state $\ket{e^{-\iw t_g \mathsf{h}}}$ for different values of the dimensionless ``Hamiltonian time'' $t_g$. As in a conventional Brownian motion we have $\overline{t_g}=0$ and $\overline{t_g^2}=Jt$. 

Assume that $\mathsf{h}$ has eigenvalues $\varepsilon_i$ and eigenstates $\ket{\varepsilon_i}$. The two-replica effective Hamiltonian can then be diagonalized as
\be\label{eq:effH2timeindep}
H_{\eff}^{a\bar{b}} = \frac{J}{2} \sum_{i,j=1}^{d^2} (\varepsilon_i - \varepsilon_j)^2 \ket{\varepsilon_i}\ket{\varepsilon_j} \bra{\varepsilon_i}\bra{\varepsilon_j}\,.
\ee
In this case, the moments take the form
\be\label{eq:momentZntimeindep}
Z_n(t) = \dfrac{1}{d^n}\sum_{i_1,\ldots ,i_n = 1}^d \exp(-t\frac{J}{2} \left[(\varepsilon_{i_1}-\varepsilon_{i_2})^2 +…+(\varepsilon_{i_n}-\varepsilon_{i_1})^2 \right]) = \dfrac{1}{d^n} \text{Tr} (K_t^n)\,,
\ee
where $K_t$ has matrix elements
\be
(K_t)_{ij}= e^{-t \frac{J}{2}(\varepsilon_i -\varepsilon_j)^2}\,,\qquad (i,j=1,…,d)\,.
\ee

To compute the dimension spanned by the family $\mathsf{F}_\Omega^t =\{\ket{U_1(t)},...,\ket{U_\Omega(t)}\}$ drawn from the ensemble at fixed Brownian time $t$, we use \eqref{eq:transitiondimension} and take the limit $n\rightarrow 0$ of \eqref{eq:momentZntimeindep}. The maximum dimension is set by the rank of this matrix,
\be
d_\Omega(t) = \min \lbrace \Omega, \text{rank}(K_t)\rbrace \,.
\ee
If there are no degeneracies, $\varepsilon_i \neq \varepsilon_j$ for $i\neq j$, as is expected for quantum chaotic Hamiltonians, this implies
\be
d_\Omega(t) = \begin{cases}
1\, & t=0\,,\\[2pt]
\min\{\Omega,d\}\, & t>0\,.
\end{cases}
\ee
Here we have used that $K_t$ has full rank for $t>0$ and unit rank at $t=0$. This follows from its positive-definiteness, since
\be
\bra{\psi} K_t \ket{\psi}=\sum_{i,j} \psi_i^* \psi_j e^{-t\frac{J}{2}(\varepsilon_i-\varepsilon_j)^2}
=\int_{-\infty}^\infty \dfrac{\text{d}\omega}{\sqrt{2\pi tJ}}\,e^{-{\omega^2}/{2tJ}}\,\Bigl|\sum_i \psi_i e^{\iw \omega \varepsilon_i}\Bigr|^2\,>0\,,
\ee
where in the second step we have introduced the ``Hubbard–Stratonovich'' variable $\omega$. If the spectrum is non-degenerate, the integrand is always positive for a non-zero vector.

This improves the related considerations in \cite{Magan:2024nkr,Banerjee:2024fmh}, where time evolution of TFDs was used to construct bases of black hole microstates. There the argument needed large times and large time separations. Here the Brownian couplings demonstrate the basis generation is immediate.

The reason the dimension is $d$ (and not $d^2$) is that the states
\be
|e^{-\iw t_g H}\rangle = \dfrac{1}{\sqrt{d}} \sum_{i}e^{-\iw t_g \varepsilon_i} \ket{\varepsilon_i}\ket{\varepsilon_i^*},
\ee
all lie within the diagonal subspace, $(\mathsf{h}^{a}-\mathsf{h}^{\bar{b}}\,^*)|e^{-\iw t_g \mathsf{h}}\rangle=0$. Moreover, if the spectrum of $\mathsf{h}$ had degeneracies, the states could be restricted to a proper subspace, which would be signaled by the rank of $K_t$.

Finally, we note that the gap of the effective Hamiltonian \eqref{eq:effH2timeindep} in this case is exponentially small, set by the minimum eigenvalue separation in the spectrum of $\mathsf{h}$. Consequently, one must wait for parametrically longer timescales than those in Section \ref{sec:overlapgeneralexp} for $Z_2(t)$ to decay sufficiently, if one wants the different draws of the family to be distinguishable.

\section{Counting in a toy model of the caterpillars}
\label{app:toymodel}

In this appendix we seek to perform the state counting of the toy models of black hole caterpillars described in Sec. \ref{sec:semiclgeom}.  To this end we start by constructing a toy model of the eternal traversable wormhole that naturally appears in the interior of caterpillars using thin shell operators. We allow a general scenario in which the local temperature inside the wormhole varies as $\beta(x)$, where $x$ parametrizes the position along the ER bridge. We shall start from a sequence of $n$ operator insertions
\be
S_n(\beta(x_i),m(x_i)) = \mathcal{O}_{m_1} e^{-\tilde{\beta}_1 H} \ldots \mathcal{O}_{m_{n-1}} e^{-\tilde{\beta}_{n-1} H} \mathcal{O}_{m_n} \; ,
\ee
along with periods of Euclidean evolution. We then take the limits $n \to \infty$ while $m_i \sim \frac{1}{\sqrt{n}} \to 0$. The idea is that for any $\beta(x)$, we can introduce a density of thin shell operators $m(x)$ that support such traversable wormhole. We will make this discussion in $d=3$ spacetime, where analytical expressions are easily obtained, following e.g. \cite{Balasubramanian:2022gmo}, but the construction is valid for any dimension. Given that the relation between the temperature and the mass of the black hole is
\be
M(x) = \frac{1}{2G} \frac{\pi^2}{\beta^2(x)} \, ,
\ee
the Israel junction conditions, see \cite{Climent:2024trz}, impose
\be
\frac{dm^2(x)}{dx} \geq \frac{\abs{M'(x)}}{2G} \;.
\ee
We can choose any shell mass density over this lower bound. Fixing $m(x)$, the consistency relation for the preparation temperature $\tilde{\beta}(x)$ is
\be
\frac{d\tilde{\beta}(x)}{dx} = \frac{1}{2M(x)} \frac{dm(x)}{dx} \sqrt{1-\frac{M'(x)}{2G \frac{dm^2(x)}{dx}}} \, ,
\ee
which, given $M(x)$ and $m(x)$, can be integrated to a particular function determining the preparation temperatures along the ER bridge.

The previous two simple relations show we can model a traversable wormhole of arbitrary interior local temperatures using thin shell operators. Nevertheless, if we want to construct the corresponding microstates (by adding the external geometries) and perform the counting to obtain the Bekenstein-Hawking entropy, it will be challenging to find the higher-moment classical solutions. Equivalently, there is a transition regime between the interior caterpillar and exterior geometry, both for microstates and for the wormholes computing the statistics of overlaps, which is typically involved. See for example the caterpillar in JT gravity described in \cite{Magan:2024aet}, whose transition regime is described by a bounce solution, while the interior caterpillar and exterior geometry asymptote to known simple solutions.

To circumvent this problem, we can insert heavy-shell operators $\mathcal{O}_m$ in both extremes of the ER bridge, and then add the external black holes of temperature $\beta$. We depict this in Fig. \ref{fig:dibujobeta_beta4}. The external ultra-massive shells allow a simple gluing with any exterior black hole geometry, since such high mass will always satisfy the constraints ($m^2>\textrm{constant}\,\Delta M$) coming from the Israel junction conditions. 

By doing so, we can easily compute the replicated partition function and obtain 
\be\label{repuniold}
Z_n = \frac{Z(n\beta)^2}{Z(\beta)^{2n}} \, ,
\ee
coinciding with \cite{Balasubramanian:2022gmo,Climent:2024trz}, from which black hole entropy can be derived. This expression arises because the heavy shell operator at the end of the ER bridge not only allow the gluing with the exterior geometry, but also conveniently factorizes out the action corresponding to the ER bridge and the exterior. When calculating the normalized replicated partition functions, the action corresponding to the ER bridge cancels out with the normalization of the states, see  \cite{Balasubramanian:2022gmo} for a detailed description, leaving only expression \ref{repuniold}.

\begin{figure}
    \centering
    \includegraphics[width=1\linewidth]{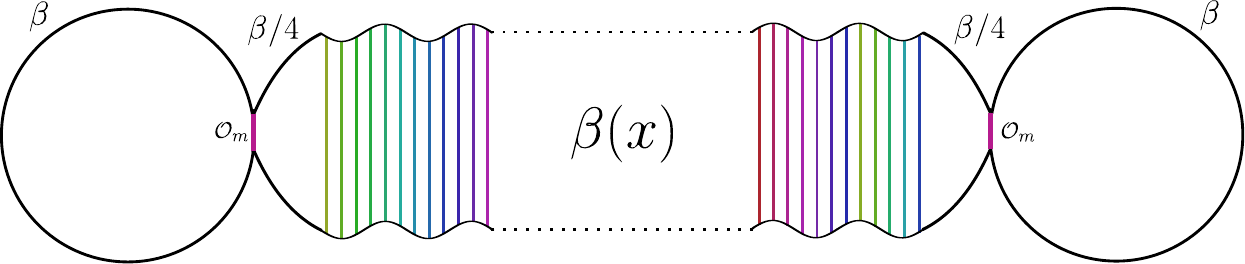}
    \caption{Toy model of black hole caterpillars valid in arbitrary dimensions. The interior is a long wormhole with an arbitrary position dependent local proper temperature, supported by shell insertions at the Euclidean boundaries. The ultra-heavy shells $\mathcal{O}_m$ at each extreme of the ER bridge allows a smooth gluing.}
    \label{fig:dibujobeta_beta4}
\end{figure}

\bibliographystyle{ourbst.bst}
\bibliography{bibliographydraft.bib}

\end{document}